\documentclass[
    aps,
    superscriptaddress,
    onecolumn,
    10pt,
    prx,
    nofootinbib
]{revtex4-2}

\usepackage{amsmath,amssymb,amsthm,bm,mathrsfs,physics}

\usepackage{graphicx}
\usepackage{subcaption}

\usepackage{tikz}
\usetikzlibrary{arrows.meta,positioning}

\usepackage{enumitem}
\usepackage[normalem]{ulem}

\usepackage[
    colorlinks=true,
    linkcolor=blue,
    citecolor=blue,
    urlcolor=blue
]{hyperref}
\usepackage[capitalise]{cleveref}
\usepackage{soul,xcolor}
\usepackage{thmtools}
\usepackage{thm-restate}

\newtheorem{theorem}{Theorem}
\newtheorem{lemma}{Lemma}
\newtheorem{proposition}{Proposition}
\newtheorem{definition}{Definition}

\newtheorem{corollary}{Corollary}

\newcommand{\Haar}{\mathrm H}
\newcommand{\PM}{\mathrm{PM}}
\newcommand{\TVD}{\operatorname{TVD}}

\newcommand{\E}{\mathbb E}

\newcommand{\cP}{\mathcal P}
\newcommand{\cE}{\mathcal E}

\newcommand{\cW}{\mathcal W}

\newcommand{\nocontentsline}[3]{}
\let\origcontentsline\addcontentsline
\newcommand\stoptoc{\let\addcontentsline\nocontentsline}
\newcommand\resumetoc{\let\addcontentsline\origcontentsline}

\newtheoremstyle{namedstatement}
    {\topsep}       
    {\topsep}       
    {\itshape}      
    {}              
    {\bfseries}     
    {.}             
    {.5em}          
    {\thmnote{#3}}  

\theoremstyle{namedstatement}
\newtheorem*{namedstatementInner}{}

\theoremstyle{plain}

\newenvironment{restatement}[1]{%
    \begin{namedstatementInner}[#1]%
}{%
    \end{namedstatementInner}%
}

\begin{document}

\title{Strong unitary designs in optimal depth and space}

\author{Teodor Parella-Dilmé}
\affiliation{
ICFO -- Institut de Ciències Fotòniques,
The Barcelona Institute of Science and Technology,
Avinguda Carl Friedrich Gauss 3,
08860 Castelldefels (Barcelona), Spain
}

\author{Júlia Barberà-Rodríguez}
\affiliation{
ICFO -- Institut de Ciències Fotòniques,
The Barcelona Institute of Science and Technology,
Avinguda Carl Friedrich Gauss 3,
08860 Castelldefels (Barcelona), Spain
}

\author{Salvatore F.~E.~Oliviero}
\affiliation{
Dahlem Center for Complex Quantum Systems,
Department of Physics,
Freie Universität Berlin,
Arnimallee 14,
14195 Berlin, Germany
}

\author{Antonio A.~Mele}
\affiliation{
Dahlem Center for Complex Quantum Systems,
Department of Physics,
Freie Universität Berlin,
Arnimallee 14,
14195 Berlin, Germany
}

\begin{abstract}
Unitary designs provide finite-moment approximations to Haar-random unitaries, with wide-ranging applications across physics and quantum information, from scrambling and black-hole dynamics to foundational primitives in quantum algorithms. Strong unitary designs capture a more demanding operational notion of approximation, requiring indistinguishability from Haar randomness even for quantum algorithms that may access a unitary not only in the forward direction, but also through its inverse, transpose, and complex conjugate. Motivated by the physical requirement that scrambling arise within the system itself, Schuster, Ma, Lombardi, Brandão, and Huang~\cite{SchusterEtAlStrong2025} left open whether strong unitary designs can be generated in logarithmic depth using only the system qubits.
For every fixed design order $k$ and measurable-error tolerance, we construct strong approximate unitary $k$-designs in optimal $\Theta(\log n)$ all-to-all circuit depth using only the $n$ original system qubits.
Our new ingredient is a logarithmic-depth Pauli-mixing bound for the perfect-matching ensemble, whose layers pair the qubits uniformly at random and apply independent random two-qubit gates. This bound controls the mixed forward-reverse two-query case, which we combine with existing design and gluing results to obtain strong unitary designs of arbitrary fixed order.
\end{abstract}

\maketitle

\stoptoc

\section{Introduction}
Quantum information scrambling describes how information initially stored in a
small part of a system spreads across many degrees of freedom and becomes
encoded in nonlocal correlations~\cite{SekinoSusskind2008,LashkariEtAl2013}. The
information remains present globally but becomes difficult to recover through
local measurements or other restricted observations. Scrambling is expected in
complex many-body systems and plays an important role in the approach to
thermalization in isolated quantum systems~\cite{Deutsch1991,Srednicki_1994,Rigol_2008}.
It is also central to black-hole dynamics, where it governs how rapidly
information falling into a black hole is dispersed and how it may later be
recovered~\cite{HaydenPreskill2007}. Related mechanisms underlie decoupling and
recovery protocols in quantum information~\cite{BrownFawzi2012,BrownFawzi2015}.
The central question is how quickly physical dynamics can scramble information.
Answering it requires a way of detecting scrambling in the first place, and of
saying when it is complete.

In many-body systems, scrambling is closely connected to the growth of initially
simple operators into increasingly nonlocal ones~\cite{ShenkerStanford2014,RobertsStanford2015,NahumVijayHaah2018}.
Out-of-time-order correlators, or OTOCs, probe this growth through the
increasing failure of an initially local operator to commute with distant
observables~\cite{LarkinOvchinnikov1969,Kitaev2015,RobertsStanford2015}. Their
measurement often relies on protocols combining forward and reverse
evolution~\cite{SwingleEtAl2016,GarttnerEtAl2017}. The same behaviour is
reproduced by Haar-random unitaries, and higher-order OTOCs make the comparison
quantitative: reproducing the first $k$ moments of the Haar measure already
brings the $2k$-point functions close to their Haar
values~\cite{HosurEtAl2016,RobertsYoshida2017}. These quantities have also been subject of study in recent near-term quantum advantage proposals~\cite{king2025simplifiedversionquantumotoc2,abanin2025constructiveinterferenceedgequantum}. 

A generic Haar unitary, however, is exponentially costly to describe or
implement. Approximate unitary designs avoid this difficulty by reproducing a
finite collection of Haar moments up to a controlled error, rather than the full
distribution~\cite{Emerson_2003,Gross_2007,dankert2005efficientsimulationrandomquantum}.
This makes it possible to compare realistic circuit ensembles with Haar
evolution at the level accessible to finite-moment experiments, and efficient
random-circuit constructions make the corresponding Haar averages accessible
without requiring the implementation of a generic Haar-random
unitary~\cite{HarrowLow2009,BrandaoHarrowHorodecki2016}. Finite-moment Haar
replacement also appears in fidelity-estimation protocols~\cite{Dankert_2009},
randomized benchmarking~\cite{MagesanGambettaEmerson2011}, and decoupling
arguments~\cite{HaydenPreskill2007,BrownFawzi2015}. These guarantees, however,
are stated for experiments that query the evolution in the forward direction
only.

In the standard query model, the experimenter accesses $U$, while $U^\dagger$
enters only through the adjoint structure of probabilities and expectation
values. By contrast, the measurement of an OTOC, for example, may apply the reverse evolution
as a separate physical step~\cite{SwingleEtAl2016}, and the efficient
Hayden--Preskill decoder uses a conjugated evolution~\cite{YoshidaKitaev2017}.
Such access is genuinely different: the appearance of $U^\dagger$ in an
expectation value does not provide a separate black box implementing the reverse
evolution. For approximate designs, a bound on the usual forward moments need
not be preserved when separate queries to $U^\dagger$ or $U^*$ reshuffle the
moment indices.

A notion of design adequate to these protocols must therefore extend the Haar
replacement to the broader query model. In the measurable-error formulation
considered here, the ensemble must remain indistinguishable from Haar measure to
bounded-query experiments that may separately access $U$, $U^\dagger$, $U^T$,
and $U^*$, while interleaving these calls with arbitrary quantum operations and
quantum memory~\cite{SchusterEtAlStrong2025}. In a black-box setting, access to
one of these transformations does not generally provide access to the others.

With this notion in place, one can ask how quickly a circuit reaches this
stronger form of Haar-like scrambling: how much depth is required before it
becomes a strong design? Schuster, Ma, Lombardi, Brand\~{a}o, and
Huang~\cite{SchusterEtAlStrong2025} conjectured that both strong unitary designs
and strong pseudorandom unitaries can be generated in logarithmic depth using
only the physical qubits. This restriction matters if the circuit is meant to
model closed-system dynamics: additional qubits enlarge the system's Hilbert
space and introduce resources absent from the dynamics being modeled.

Logarithmic depth is the natural target because it is the earliest possible
asymptotic scale in the all-to-all circuit model. A layer of bounded-size gates
can enlarge the support of a local operator by only a constant factor, so
system-wide operator spreading requires depth
$\Omega(\log n)$~\cite{SchusterEtAlStrong2025}. The same work established this
lower bound for strong designs and gave constructions that either require
additional qubits or incur an additional logarithmic factor when restricted to
the system alone. It remained open whether the optimal logarithmic depth could
be reached on the original system, for fixed design order and fixed
measurable-error tolerance.

We resolve this question for strong unitary designs with a particularly simple
random-circuit architecture. Each layer draws a uniformly random perfect
matching of the $n$ qubits, with $n$ even, and applies independent Haar-random
two-qubit gates to the matched pairs. Our result also implies that OTOC-growth can occur in logarithmic time, within such architecture using solely the system-qubits.

The analysis rests on a uniform Pauli-mixing result for this ensemble. Starting
from any nonidentity Pauli operator, the induced distribution over Pauli strings
approaches the corresponding Haar-induced distribution in total variation after
a number of layers logarithmic in the system size and in the inverse target
accuracy. The bound holds uniformly over the initial Pauli operator and
throughout the stated error range. Typical-input or average-case mixing would not
suffice, since the strong-design reduction requires control of the relevant
two-query experiment for every initial Pauli. We obtain this worst-case estimate
by reducing the evolution to a Markov chain on Pauli supports and constructing a
monotone grand coupling.

This uniformity is what controls two-query experiments containing one call from
$\{U,U^T\}$ and one from $\{U^\dagger,U^*\}$, which are not covered by ordinary
unitary designs. Combining the estimate with an independent weak design for the
ordinary forward-query sector, and then applying the strong-gluing framework to
reach higher orders, gives approximate strong unitary designs of every fixed
order and fixed measurable-error tolerance on the original $n$-qubit system,
with optimal all-to-all depth $\Theta(\log n)$. Strong-design behaviour is
therefore reached on the same asymptotic depth scale as that required for an
initially local operator to acquire system-wide support.

\Cref{sec:comb} defines the strong-design model, its measurable-error criterion,
and the construction framework. \Cref{sec:Main_results} states the formal
results, while \Cref{sec:methods} explains the proof strategy. The appendices
provide the technical background and imported inputs, establish the Pauli-mixing
estimate, and assemble the main constructions. \Cref{sec:discussion} returns to
the scope of the result and the remaining open directions.

\section{Background}\label{sec:comb}
\subsection{Approximate strong designs}
A unitary $k$-design is an ensemble of unitaries $\mathcal{E}$ that cannot be distinguished from Haar-random unitaries by any experiment making at most $k$ queries. The standard notion of a unitary $k$-design, which we refer to as a weak $k$-design, restricts the distinguisher to $k$ forward queries to the same sampled unitary $U\sim\mathcal E$. Strong $k$-designs require more. After drawing $U$, a $k$-query experiment may have access to $U$, $U^\dagger$, $U^T$, or $U^*$. Throughout, these are ordinary oracle calls. Controlled versions require an additional convention, since they make the global phase of the sampled unitary observable; we state this convention below. The ensemble is therefore a strong $k$-design if every such experiment remains indistinguishable from Haar-random~\cite{SchusterEtAlStrong2025}.

Exact unitary designs, which reproduce the first $k$ moments of the Haar measure exactly, are available only in special settings. In practice, one therefore works with approximate designs. The choice of approximation error is important because it determines which experiments are guaranteed to exhibit Haar-like behaviour. Three inequivalent notions appear in the literature, all defined in
detail in Appendix~\ref{app:preliminaries}. Two such notions, the \textit{additive} and \textit{relative} errors, are defined through the $(p,q)$ mixed-moment channel of $\cE$
\begin{equation}\label{eq:mixed-twirl-main}
    \Phi_{\cE}^{(p,q)}(X)
    :=
    \E_{U\sim\cE}\!\left[
        \bigl( U^{\otimes p} \otimes U^{*,\otimes q} \bigr)
        X
        \bigl( U^{\dagger,\otimes p} \otimes U^{T,\otimes q} \bigr)
    \right].
\end{equation}
where $p$ is the number of forward-time queries \{$U$, $U^T$\}, and $q$ is the number of time-reversal queries \{$U^*$,$U^{\dagger}$\}, with $p+q=k$. The \emph{additive error} represents the simplest of the three notions. It controls the experiments in which all $k$ uses of the sampled unitary occur in a single parallel call, while allowing the input state to be entangled with an arbitrary reference system. In its strong version, it is defined as $\varepsilon_{a}:=\max_{p+q=k}\|\Phi_{\cE}^{(p,q)}-\Phi_{\mu_H}^{(p,q)}\|_\diamond$. The weak version is restricted to the forward-time sector, corresponding to $p=k$, $q=0$.

The \emph{relative error} imposes the most demanding approximation guarantee. It controls any measurement outcome multiplicatively: if an event occurs with probability $p$ for the Haar measure, then its probability under an $\varepsilon_{r}$-relative-error design $\cE$ is between $(1-\varepsilon_{r})p$ and $(1+\varepsilon_{r})p$, even when $p$ is arbitrarily small. It also provides adaptive security:  it covers experiments that interleave the $k$ queries with arbitrary operations and memory, even for events no physical experiment can efficiently resolve. In its strong version, a relative error $\varepsilon_{r}$ requires
\begin{equation}\label{eq:strong-relative-main}
    (1-\varepsilon_{r})\Phi_{\mu_H}^{(p,q)}
    \preceq_{\mathrm{CP}}
    \Phi_{\cE}^{(p,q)}
    \preceq_{\mathrm{CP}}
    (1+\varepsilon_{r})\Phi_{\mu_H}^{(p,q)}
    \qquad
    \text{for every }p+q=k,
\end{equation}
where $\Psi\preceq_{\mathrm{CP}}\Phi$ means that $\Phi-\Psi$ is completely positive. Its weak version again restricts to $p=k$, $q=0$.

While relative error guarantees security against arbitrary adaptive experiments, its multiplicative control of the outcome probabilities is
unnecessarily restrictive: it demands accuracy even on outcomes so rare that a
bounded number of queries would never produce a single one of them. Instead of looking at the multiplicative ratio between the two outcome probabilities, one can look instead at their absolute difference, which is what an experiment can actually notice. This is the \emph{measurable error}: over the class of adaptive experiments, it bounds how
well a single run can distinguish $\mathcal{E}$ from the Haar measure.

Quantum combs provide a natural framework for describing these mixed-query experiments~\cite{ChiribellaEtAl2009}. Their $k$ oracle slots specify the sequence of queries, while the interleaving maps describe the adversary's operations and carry its memory from one query to the next.
More explicitly, let $\mathsf{Q}$ denote the $n$-qubit register on which the sampled unitary acts, and let $\mathsf{M}$ be an arbitrary finite-dimensional memory register held by the adversary. The mixed-query experiment is specified by a choice of queries $\bm s=(s_1,\ldots,s_k)\in\{1,\dagger,T,*\}^k$, which record which of the four available transformations is used at each of the $k$ oracle calls. For $s_i\in\{1,\dagger,T,*\}$, we denote the corresponding channel on $\mathsf{Q}$ by $\mathcal V_{s_i}^U(X):=U^{s_i}X(U^{s_i})^\dagger$. At the $i$-th oracle call, this channel $\mathcal V_{s_i}^U$ acts on $\mathsf{Q}$ while the memory register is left unchanged.

The most general way to connect several such calls is through a $k$-slot quantum comb $\mathcal W$. Starting from an initial state $\rho_0:=|0\rangle^{\otimes (n+m)}\langle 0|^{\otimes (n+m)}$ on the query and memory registers $\mathsf{Q}\otimes \mathsf{M}$ with $m=\log_2 \operatorname{dim}\mathsf{M}$, the comb applies a sequence of channels $\{\mathcal{W}_{i}\}_{i=1}^{k+1}$ on $\mathsf{Q} \otimes \mathsf{M}$.
For a fixed unitary $U$ and choice of queries $\bm s$, the state produced before the final measurement is
\begin{equation}\label{eq:comb-output-main}
    \rho_{\mathcal W,\bm s}^{U}
    :=
    \mathcal W_{k+1}
    \circ
    \bigl(
        \mathcal V_{s_k}^{U}
        \otimes
        \operatorname{id}_{\mathsf{M}}
    \bigr)
    \circ
    \mathcal W_{k}
    \circ\cdots\circ
    \mathcal W_2
    \circ
    \bigl(
        \mathcal V_{s_1}^{U}
        \otimes
        \operatorname{id}_{\mathsf{M}}
    \bigr)
    \circ
    \mathcal W_1
    (\rho_0).
\end{equation}
The general comb formalism is illustrated in~\cref{fig:comb}. We call this a fixed-word adaptive comb. The choice of queries $\bm s$ is fixed before the experiment. The comb may adapt its interleaving operations and retain arbitrary quantum memory, but later query orientations cannot be chosen adaptively.

We now consider an  ensemble $\cE$ of unitaries. For a fixed word $\bm s$ and comb $\mathcal W$, averaging over this ensemble gives the output state $\E_{U\sim\cE}\left[\rho_{\mathcal W,\bm s}^{U}\right]$. We compare this state with the output obtained by averaging the same comb over the Haar measure $\mu_H$, that is $\E_{U\sim \mu_H}\left[\rho_{\mathcal W,\bm s}^{U}\right]$. 
We say that an ensemble $\cE$ is a strong $\varepsilon_m$-approximate unitary $k$-design in measurable error if
\begin{equation}\label{eq:measurable-error-main}
 \operatorname{Err}_{k}^{\rm full}(\mathcal{E}):=\sup_{\bm s,\mathcal W}
\left\|
\E_{U\sim\cE}
\left[\rho_{\mathcal W,\bm s}^{U}\right]
-
\E_{U\sim \mu_H}
\left[\rho_{\mathcal W,\bm s}^{U}\right]
\right\|_1
\leq
\varepsilon_m,
\end{equation}
where the supremum ranges over all fixed query words $\bm s$ and all finite-memory adaptive $k$-slot combs $\mathcal W$. 

Phase randomization gives a controlled-query version. Let $\widetilde{\cE}$ be the distribution of $\widetilde U=e^{i\phi}U$, where $U\sim\cE$ and $\phi$ is uniform on $[0,2\pi)$ and independent of $U$. For every $\bm{s}\in\{1,\dagger,T,*\}$, the identity $\mathcal V_s^{e^{i\phi}U}=\mathcal V_s^U$ shows that $\widetilde{\cE}$ has the same ordinary-query guarantee as $\cE$. Refs.~\cite{Sheridan_Maslov_Mosca_2009,tang2026controlledunitarieshelpful} give the same guarantee for controlled versions of $\widetilde U$, $\widetilde U^\dagger$, $\widetilde U^T$, and $\widetilde U^*$, as long as every orientation uses the same sampled phase. This applies only to $\widetilde{\cE}$.

Note that this definition for measurable error also covers experiments using fewer than $k$ queries: one may append padding queries acting on a maximally mixed register and discard their outputs.

\begin{figure}[t]
  \centering
  \includegraphics[width=0.9\linewidth]{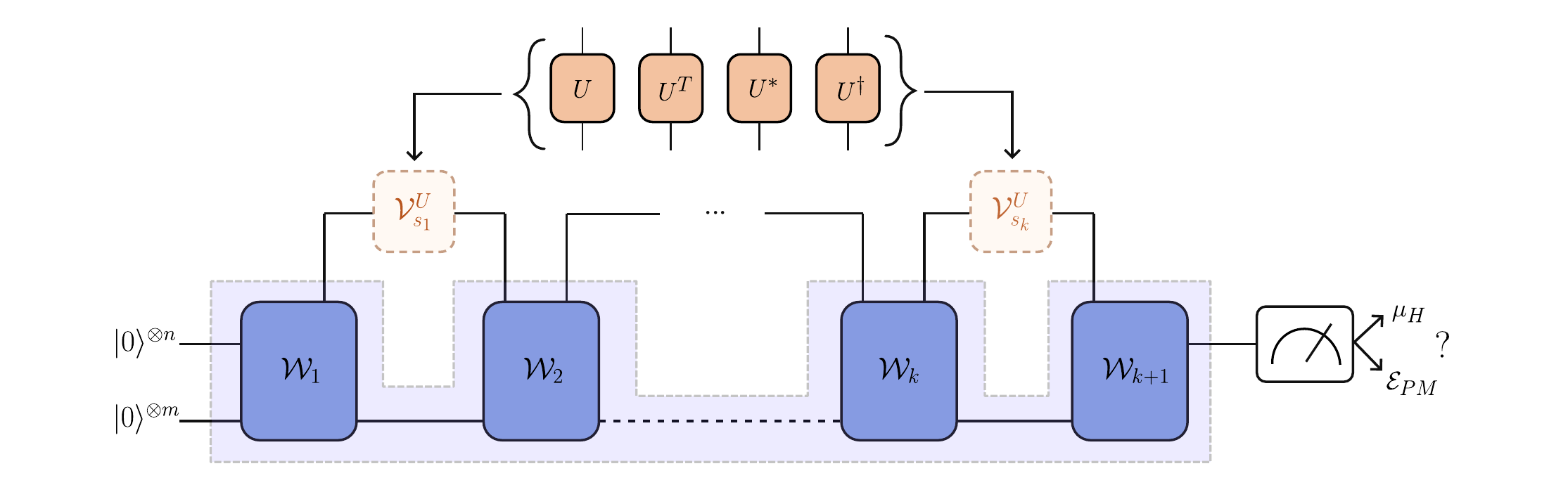}
  \caption{A mixed-query distinguisher represented as a quantum comb $\mathcal{W}$. A single unitary $U$ is sampled once from an ensemble and used throughout the experiment: in each of the $k$ slots, the adversary may place an oracle query of $U$, $U^\dagger$, $U^T$, or $U^*$. The comb models the arbitrary operations $\mathcal{W}_1 \cdots \mathcal{W}_{k+1}$ the distinguisher may apply between queries, and arbitrary retention of quantum memory.}
  \label{fig:comb}
\end{figure}

Two normalizations of \cref{eq:measurable-error-main} appear in the literature and in the imported statements we use throughout our work. We write $ \operatorname{Err}_{k}^{\rm full}$ for the full trace norm, as displayed
above, and  $\operatorname{Err}_{k} :=\frac{1}{2}  \operatorname{Err}_{k}^{\rm full}$ for the half trace distance, which is the operationally normalized
distinguishing advantage. We will make precise which one we are using in each case. All our final statements about measurable error are quoted in the full trace norm of Eq.~\eqref{eq:measurable-error-main}. 

\subsection{Assembling high order designs: strong gluing}
\label{sec:assembling}
Having introduced the operational notion of an approximate strong unitary design, we now explain how strong designs of higher order can be assembled from simpler components via strong gluing. We consider two brickwork layers, with bricks consisting of strong $k$-designs, sandwiched between two global strong-$2$-design shells (throughout this paper, a shell is an ensemble acting on all $n$
system qubits), as shown in~\cref{fig:strong-k}. We partition the system into $m=\lfloor n/\xi \rfloor$ patches of $\xi$ qubits, and the last patch absorbs the remainder $\zeta=n-m\xi$ qubits. Each strong $k$-design in the brickwork has support $2\xi$ qubits (except the last one, which has support $2\xi+\zeta$), and overlaps with each neighbour on a support of
$\xi=\Theta(\log(nk/\varepsilon))$ qubits. In this precise configuration, the strong gluing \cref{lem:gluing_strong_random_unitaries} lets us glue together two overlapping strong $k$-designs, into a unique strong $k$-design, at a cost controlled by the overlap size, the local error, and the shell error:

\begin{lemma}[Gluing strong random unitaries, adapted from Lemma~8 of Ref.~\cite{FolkertsmaEtAl2026}, and Lemma~1 of Ref.~\cite{SchusterEtAlStrong2025}]
\label{lem:gluing_strong_random_unitaries}
Let ($a$, $b$, $c$, $d$) be four quantum registers, where $a$, $b$, $c$ are at least size $\xi$. Consider
the unitary ensemble $U_1 = D_{abcd}U_{ab}U_{bc}C_{abcd}$, where $C_{abcd}$ and $D_{abcd}$ are strong
$\varepsilon_2$-approximate unitary $2$-designs and $U_{ab}$, $U_{bc}$ are strong
$\varepsilon_{ab}$- and $\varepsilon_{bc}$-approximate unitary $k$-designs on the indicated registers. Consider another ensemble $U_2 = D_{abcd}U_{abc}C_{abcd}$, where $U_{abc}$ is Haar
random. Then any $k$-query experiment whose queries can be chosen from $\{U, U^\dagger, U^{T}, U^*\}$
cannot distinguish $U_1$ from $U_2$ up to
$\varepsilon = \varepsilon_{ab} + \varepsilon_{bc} + O\!\left(k^{2}/2^{(3/16)\xi}\right)
+ O\!\left(k^{5/8}\varepsilon_{2}^{1/8}\right)$ in trace distance.
\end{lemma}

Iterating the strong gluing \cref{lem:gluing_strong_random_unitaries} across the patches allows us to transform the whole two-layer brickwork into an extensive strong $k$-design. We also import this
guarantee, with its explicit tolerances, as \cref{lem:patch_protocol} in
\cref{app:iterated-gluing}. The bricks themselves (strong $k$-designs on a support of size at most
$3\xi=\Theta(\log(nk/\varepsilon))$ qubits) can be realized by one-dimensional random circuits, as stated in the imported lemma:

\begin{lemma}[1D random circuits are strong unitary $k$-designs, Lemma~3 of Ref.~\cite{SchusterEtAlStrong2025}]
There exists a universal constant $C_{\rm 1D}>0$ such that one-dimensional random circuits on $n$ qubits form strong
$\varepsilon$-approximate unitary $k$-designs in relative error whenever
\begin{equation*}
    d
    \geq
    C_{\rm 1D}
    \log^{7}(k)
    \left(
        nk+\log\frac{1}{\varepsilon}
    \right),
\end{equation*}
for any $\varepsilon \geq 2k^{2}/2^{n}$.

\label{lem:1D_strong_unitary_designs}
\end{lemma}

That is, 1D brickwork circuits generate $k$-designs at a depth that scales linearly with the number of qubits on which they act. We therefore apply them only to small blocks of size at most
$3\xi=\Theta(\log(nk/\varepsilon))$. By \cref{lem:1D_strong_unitary_designs}, each block can be implemented in depth $O(k\log^7(2k)\log(nk/\varepsilon))$, and since the blocks within one macro-layer have
disjoint supports, they can be implemented in parallel. The two brickwork macro-layers therefore have depth $O(\log n)$ for fixed $k$ and $\varepsilon$, and the resulting strong $k$-design meets the desired depth scaling we target for our final strong designs.

The remaining bottleneck is the pair of outer strong $2$-design shells required by the gluing lemma. In the construction of Ref.~\cite{SchusterEtAlStrong2025}, these shells are implemented using a blocked fast scrambling circuit architecture. The register is partitioned into patches, and the light cone of each patch doubles at every spreading step. Reaching the full system therefore requires $O(\log(n/\xi))$ such steps, each realized by a brickwork $2$-design of nonconstant depth. As a result, the complete strong $2$-design shell has depth $O(\log^2 n)$.

These shells are the only part of the construction whose depth exceeds $O(\log{n})$, and they are therefore responsible for the overall $O(\log^{2} n)$ scaling. The remainder of this work removes this bottleneck by constructing a strong $2$-design shell in depth $O(\log n)$.

\begin{figure}[t]
    \centering

    \begin{subfigure}[t]{0.32\textwidth}
        \centering
        \includegraphics[width=\linewidth]{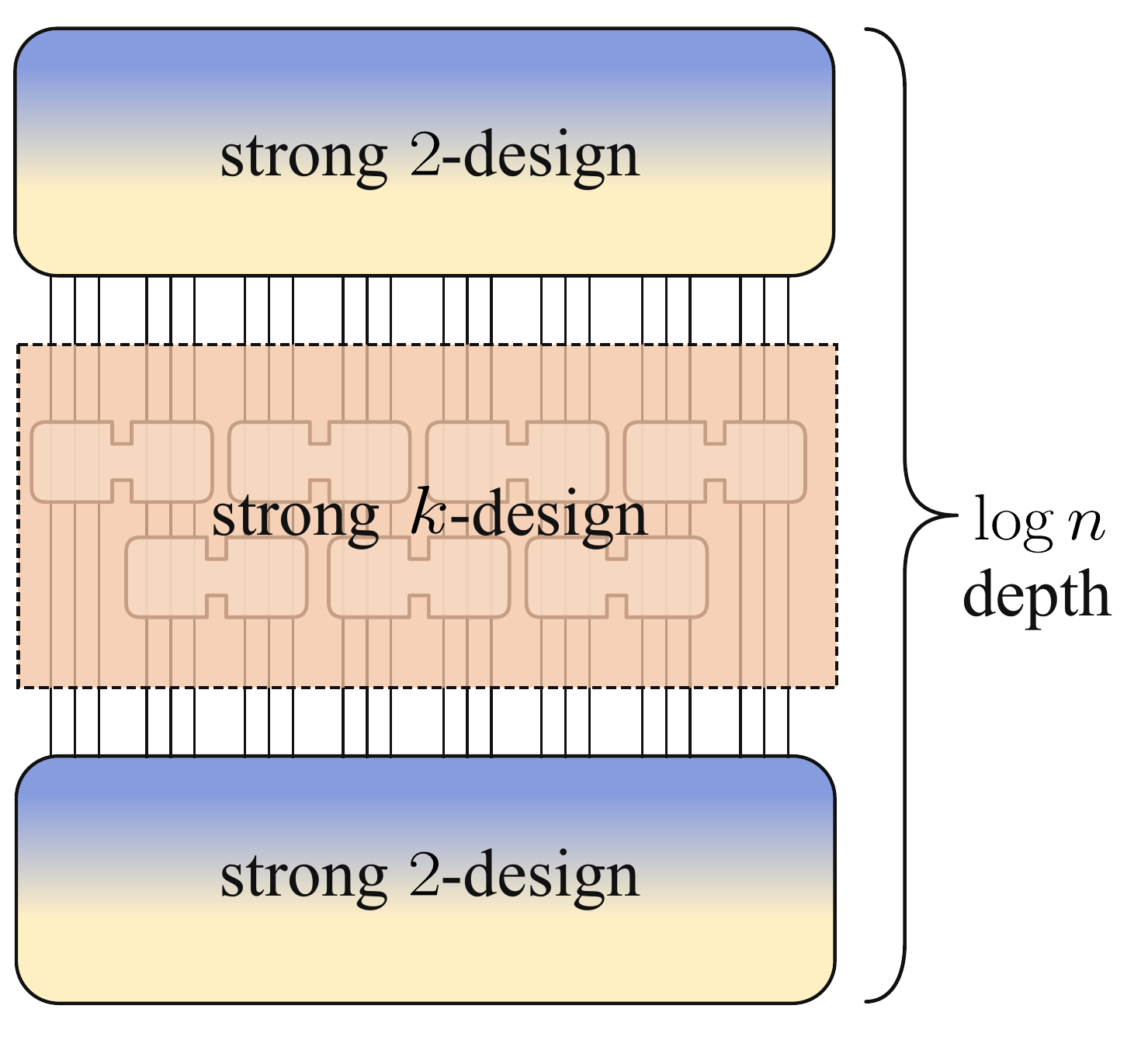}
        \caption{Strong $k$-design}
        \label{fig:strong-k}
    \end{subfigure}
    \hfill
    \begin{subfigure}[t]{0.32\textwidth}
        \centering
        \includegraphics[width=\linewidth]{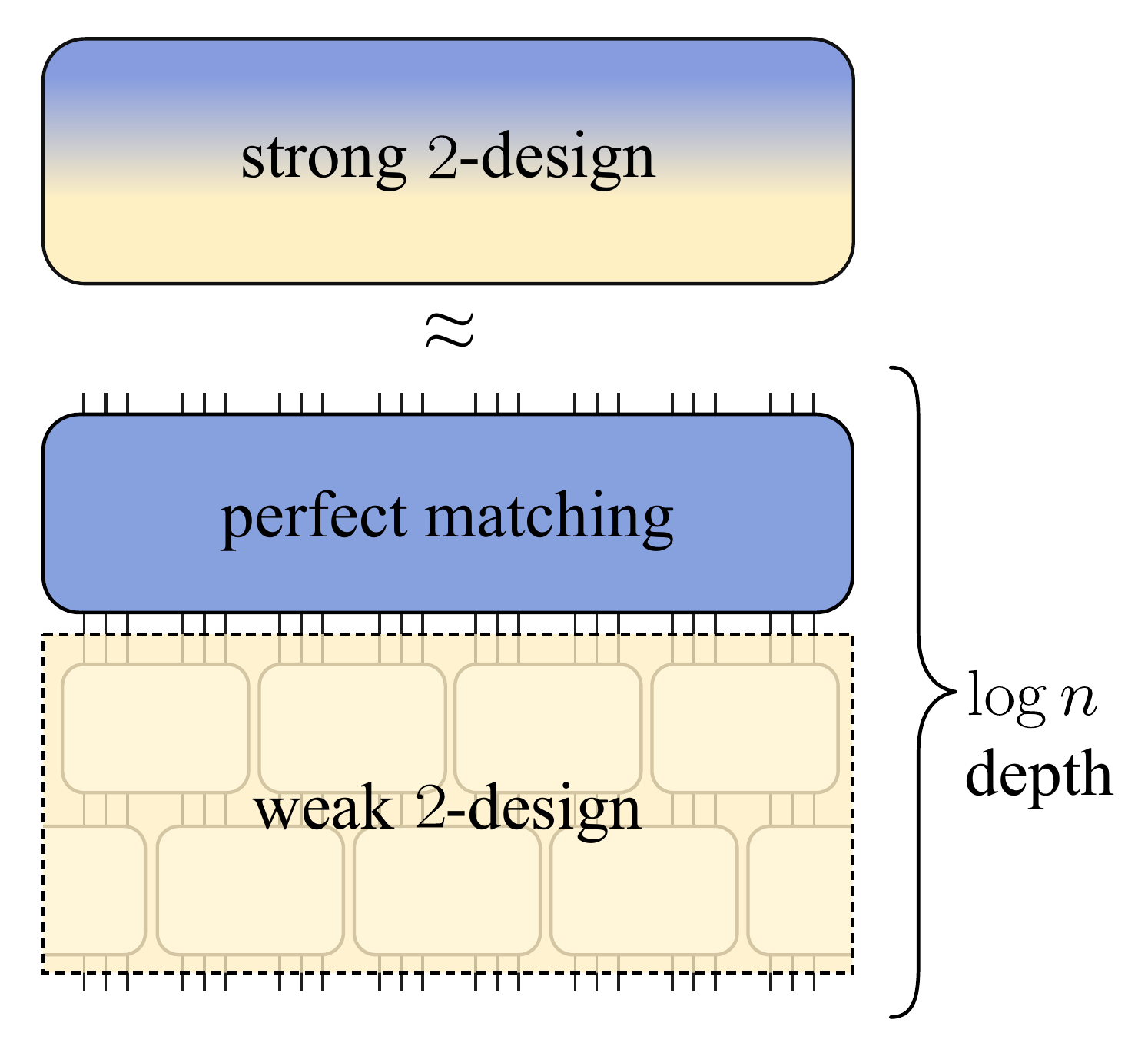}
        \caption{Approximate strong $2$-design}
        \label{fig:strong-2-layer}
    \end{subfigure}
    \hfill
    \begin{subfigure}[t]{0.32\textwidth}
        \centering
        \includegraphics[width=\linewidth]{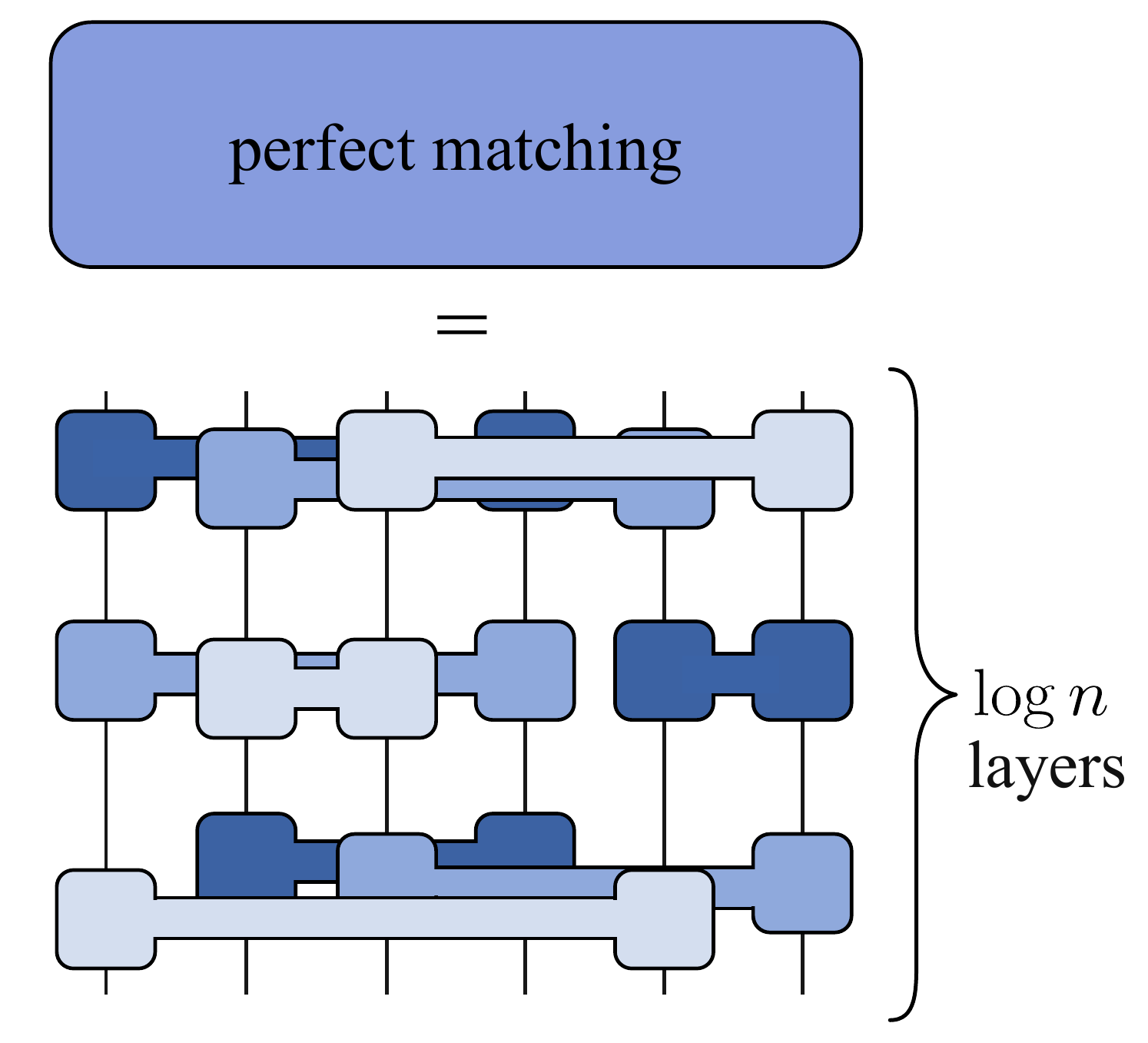}
        \caption{Perfect matching circuit}
        \label{fig:perfect-matching-layer}
    \end{subfigure}
    \caption{Architecture of our realization for approximate strong $k$-designs in measurable error. Each panel resolves a
component of the previous one. (a) The strong gluing \cref{lem:gluing_strong_random_unitaries} upgrades two brickwork macro-layers of local strong $k$-designs into an extensive strong $k$-design, after being sandwiched between two global strong 2-design shells
(\cref{thm:strong-k}). (b) The strong 2-design shell is realized as a perfect-matching circuit composed with an independent weak relative-error 2-design (\cref{thm:shell}). The former controls the mixed-sign query sector, while the latter controls the same-sign sector, both in depth $O(\log n)$ for constant error. (c)~Example of a perfect-matching
circuit: each layer pairs all $n$ qubits uniformly at
random and applies an independent Haar-random gate to every pair.}
    \label{fig:three-subfigures}
\end{figure}

\subsection{Composing the strong $2$-designs}
\label{sec:composing_2}
A strong $2$-design must be robust under every experiment that makes two queries, each drawn from $U$, $U^{T}$, $U^{*}$ or $U^{\dagger}$. Instead of checking each of the 16 possible options one by one, we exploit the fact that they
divide naturally into two classes, so controlling both classes controls all two-query experiments. We call $U$ and $U^{T}$ the positive (or forward-time) orientations and $U^{*}$ and $U^{\dagger}$ the negative (or time-reversal) ones. A 2-query experiment is then called same-sign if both of its queries come from the same orientation group, and mixed-sign (or mixed-queries) if it takes one from each.

Same-sign queries reduce to the ordinary two-copy moment channel. For a fixed word $\bm{s}\in\{1,T\}^{2}$, replacing a query by its transpose only exchanges the corresponding input and output wires in the Choi representation. Relabelling these wires does not affect the relative-error inequalities, and applying the comb preserves their completely positive order. Lemma~6 of Ref.~\cite{SchusterEtAlStrong2025} then bounds the full trace-norm measurable error by $2\varepsilon_{2}$. The same argument applies to words in $\{*,\dagger\}^{2}$, since the weak ensemble and the Haar measure are invariant under complex conjugation. The additional difficulty in the strong setting comes from mixed-sign sequences, which allow an adversary to combine forward and reverse evolution within the same experiment. The following imported result shows that composing an ensemble that controls the same-sign sector with one that controls the mixed-sign sector yields a measurable-error strong $2$-design, with the error measured in trace distance:
\begin{lemma}[Composition of the two
query sectors, Lemma~19 of Ref.~\cite{FolkertsmaEtAl2026}]
\label{lem:composition}
Let $\mu_{1}$ be an ensemble that is secure, with measurable error
$\varepsilon_{1}$, against every adversary querying one of $U,U^{T}$ together
with one of $U^{*},U^{\dagger}$, and let $\mu_{2}$ be a weak
$\varepsilon_{2}$-approximate relative-error $2$-design. Then the convolution
$\mu_{1}*\mu_{2}$ is a strong $2$-design with measurable error $\max\left(\varepsilon_{1},\,2\varepsilon_{2}\right)$.
\end{lemma}

We illustrate such composition in~\cref{fig:strong-2-layer}. Importantly, weak $2$-designs in relative error are available in logarithmic depth, following the weak gluing lemma in Ref.~\cite{SchusterHaferkampHuang2025}.~\cref{lem:composition} therefore reduces the construction of the strong $2$-design to a single missing ingredient: a logarithmic-depth ensemble robust under mixed-sign queries.

\subsection{From mixed queries to Markov Chains}
\label{sec:mixed-query-reduction}

By Lemma~\ref{lem:composition}, it remains to control the mixed-sign query sector. We now show that, under suitable symmetries, this quantum distinguishability problem reduces to the mixing of a classical Markov chain on Pauli strings.

Let $\mathcal{P}_{n}=\{I,X,Y,Z\}^{\otimes n}$, and
$\mathcal{P}_{n}^{\star}=\mathcal{P}_{n}\setminus\{I^{\otimes n}\}$, and set
$D=2^{n}$. For any $P\in\mathcal{P}_n^{\star}$, conjugation by a unitary spreads a Pauli operator over this basis,
$UPU^{\dagger}=\sum_{Q\in\mathcal{P}_{n}^{\star}}c_{Q}(U)\,Q$, with the squared
coefficients $|c_{Q}(U)|^{2}$ adding to one. Therefore, the squared coefficients define a
probability distribution over output Pauli strings:

\begin{definition}[Pauli transition distribution]
\label{def:pauli-dist}
For an ensemble $\mathcal{E}$ of unitaries on $n$ qubits and $P,Q\in\mathcal{P}_{n}^{\star}$, define
\begin{equation}
\label{eq:pauli-dist}
    p_{\mathcal{E}}(Q\,|\,P)
    \;:=\;\frac{1}{D^{2}}\,
    \E_{U\sim\mathcal{E}}
    \bigl|\operatorname{Tr}\bigl(QUPU^{\dagger}\bigr)\bigr|^{2}.
\end{equation}
Since $UPU^{\dagger}$ is unitary and traceless, $p_{\mathcal{E}}(\,\cdot\,|\,P)$
is a probability distribution on $\mathcal{P}_{n}^{\star}$ for every $P$
\end{definition}

In words, $p_{\mathcal{E}}(\,\cdot\,|\,P)$ is the probability distribution on
Pauli strings induced by the ensemble, with $p_{\mathcal{E}}(Q\,|\,P)$ the
probability of finding the evolved operator on the string $Q$: conjugation by a
random circuit turns operator spreading into a classical probability
distribution over Pauli strings. Under Haar-random conjugation, this weight is distributed uniformly over the nonidentity Pauli strings, so that $p_{\mu_{H}}(Q\,|\,P)=1/(4^{n}-1)=:\pi(Q)$ for all
$P,Q\in\mathcal{P}_{n}^{\star}$. Thus, $\pi(Q)$ is the Haar target distribution for the induced classical dynamics. The following result connects convergence to this distribution with robustness against mixed-sign queries, provided the ensemble carries a short list of symmetries:

\begin{proposition}[Mixed-query reduction, adapted from Proposition~2 of
Ref.~\cite{SchusterEtAlStrong2025}]
\label{pro:mixed_query_reduction}
Consider a unitary ensemble $\mathcal{E}$ that is invariant under
conjugation, transposition, and random single-qubit Pauli rotations at the
input and output circuit layers.
 The ensemble forms a strong approximate unitary $2$-design with measurable error in half trace distance
\begin{equation}
\operatorname{Err}_{2}(\mathcal{E}) \;=\; \max_{P\in\mathcal{P}_{n}^{\star}}
\mathrm{TVD}\bigl(p_{\mathcal{E}}(\,\cdot\,|P),\,p_{\mu_H}(\,\cdot\,|P)\bigr),
 \end{equation}
 in any quantum experiment that queries one of $U$ or $U^{T}$ and one of $U^{*}$ or $U^{\dagger}$.
\end{proposition}
Proposition~\ref{pro:mixed_query_reduction} reduces robustness against mixed-sign queries to a mixing problem of a classical Markov chain. It therefore remains to construct an all-to-all ensemble of depth $O(\log n)$ that satisfies the required symmetries and whose Pauli transition distribution approaches $\pi$ in total variation distance.

\subsection{A depth lower bound for measurable error strong designs}
Finally, we introduce a depth lower bound for strong designs, against which our construction should be compared. A simple light-cone argument gives a logarithmic depth lower bound for strong unitary $k$-designs under all-to-all connectivity, $\Omega(\log n)$:

\begin{proposition}[Depth lower
bounds for strong unitary designs. Adapted from Proposition~1 of Ref.~\cite{SchusterEtAlStrong2025}]
Let $k\geq 2$ and $\varepsilon<1/4$. Any all-to-all circuit ensemble on $n$ physical qubits
that forms a strong $\varepsilon$-approximate unitary $k$-design 
must have depth
\begin{equation}
\label{eq:lb_alltoall}
    d \;=\; \Omega\!\left(\log n+\log k\right),
\end{equation}
even if the implementation allows an arbitrary number of additional qubits.
\label{pro:lower_bound}
\end{proposition}

Among the three error notions considered, measurable error is the strongest one not excluded by the known lower bound. Proposition~3 of Ref.~\cite{SchusterEtAlStrong2025} shows that circuits built from independent local Haar-random gates require depth $\Omega(n)$ to form a strong $2$-design in relative error, irrespective of the circuit geometry, including all-to-all connectivity. Relative error therefore rules out logarithmic-depth constructions for this class of random circuits.

We consequently focus on measurable error. Our goal is to construct strong unitary $k$-designs in depth $O(\log n)$ for fixed $k$ and fixed accuracy. Together with Proposition~\ref{pro:lower_bound}, such a construction establishes the optimal asymptotic depth $\Theta(\log{n})$.

\section{Main results} \label{sec:Main_results}

The previous section reduced the construction of a logarithmic-depth strong $k$-design to a single task. By \cref{pro:mixed_query_reduction}, it is sufficient to find an ensemble $\mathcal{E}$ of depth $O(\log n)$ satisfying two properties. First, $\mathcal{E}$ must be invariant under conjugation, transposition, and random single-qubit Pauli rotations at the input and output circuit layers required by the mixed-query reduction. Second, for every nonidentity Pauli string, its Pauli transition distribution $p_{\mathcal{E}}(\cdot|P)$ must be within total variation distance $\eta$ of the uniform Haar distribution $\pi$.

We show that both requirements are met for the perfect-matching ensemble, illustrated in~\cref{fig:perfect-matching-layer}. In every layer, the qubits are paired uniformly at random via $n/2$ edges, and for every edge, a random two-qubit gate is applied. We now define the ensemble formally:

\begin{definition}[Perfect-matching ensemble]\label{def:pm-ensemble}
Let $n$ be even, and let $\mathcal{M}_{n}$ be the set of perfect matchings on
$[n]$.  A single layer of the perfect-matching ensemble $U_M\sim\mathcal{E}_{PM}^{(1)}$ is generated by sampling $M\sim\mathrm{Unif}(\mathcal{M}_{n})$ and applying
an independent Haar-random gate $U_{e}\in SU(4)$ to each edge 
$e\in M$. The resulting layer is
\begin{equation}\label{eq:pm-layer}
 U_{M}=\bigotimes_{e\in M}U_{e}.
\end{equation}
The depth-$T$ perfect-matching ensemble $\mathcal{E}_{\rm PM}^{(T)}$ is the distribution of circuits $R=U_{M_T}\cdots U_{M_1}$ where the $T$ layers are sampled independently.
\end{definition}

The symmetry check in~\cref{pro:mixed_query_reduction} follows directly from the Haar invariance of the two-qubit gates. Both complex conjugation and transposition preserve the local gate distribution; transposition additionally reverses the order of the layers, but this does not change the distribution of the whole circuit because the layers are independent and identically distributed. The same Haar invariance also handles Pauli rotations at the boundaries, since a tensor-product Pauli factorizes over the edges of the first or last matching and can be absorbed into the adjacent two-qubit gates. Thus, $\mathcal E_{\mathrm{PM}}^{(T)}$ satisfies the assumptions of~\cref{pro:mixed_query_reduction}.

Because a perfect matching exists only when $n$ is even, we formulate the results below for even system size. Our main technical result shows that, starting from any nonidentity Pauli string, the induced Pauli Markov chain approaches the uniform distribution $\pi$ after only logarithmically many layers:

\begin{restatable}[Uniform Pauli mixing for the perfect-matching ensemble]
    {theorem}{paulimixing}
\label{thm:mixing}
For the perfect-matching ensemble $\mathcal{E}_{PM}^{(T)}$, there exist
universal constants $C_{\rm mix}\geq 1$, $0<c_{\rm mix}\leq 1$, and
$n_{\rm mix}\in\mathbb{N}$ such that the following holds. For every even
integer $n\geq n_{\rm mix}$ and every accuracy parameter $e^{-c_{\rm mix}n}\leq\eta\leq 1/2$, if
\begin{equation}
\label{eq:t1}
    T \;\geq\;
    \left\lceil C_{\rm mix}\log\frac{n}{\eta}\right\rceil,
\end{equation}
then
\begin{equation}
\label{eq:t2}
    \max_{P\in\cP_n^\star}
    \TVD\!\left(
        p_{\cE_{\PM}^{(T)}}(\,\cdot\mid P),\,
        \pi
    \right)
    \;\leq\;
    \eta.
\end{equation}
\end{restatable}

The uniformity over $P$ is the essential point, as \cref{thm:mixing} controls the worst-case nonidentity input Pauli string, precisely as required by~\cref{pro:mixed_query_reduction}. The proof is given in Appendix~\ref{app:PauliMixing}, and its main ingredients are summarized in Section~\ref{sec:methods}.

We next combine this perfect-matching layer, which controls the mixed-sign queries in logarithmic depth, with an independent weak relative-error $2$-design that controls the same-sign sector. Such weak designs are also available in logarithmic depth through weak gluing on the brickwork construction result of Ref.~\cite{SchusterHaferkampHuang2025}. We denote the corresponding ensemble as $\nu_{n,\delta}$ (see
\cref{thm:shh-lowdepth} in Appendix~\ref{app:apendix_relativeerr_twodesign}). The composition of the two query sectors is \cref{lem:composition}, and yields a $2$-design in measurable error:

\begin{theorem}[Strong $2$-design shell]
\label{thm:shell}
Let $n\geq n_{\rm mix}$ be even, and let
$e^{-c_{\rm mix}n}\leq\eta\leq1/2$, with $0 <\delta\leq1$. Suppose that
\begin{equation}
    T
    \geq
    \left\lceil
        C_{\rm mix}\log\frac{n}{\eta}
    \right\rceil.
    \label{eq:shell-depth-condition}
\end{equation}
Let $V\sim\cE_{\PM}^{(T)}$, and let $W\sim\nu_{n,\delta}$ be sampled independently. Denote by $\cE_{\PM}^{(T)}*\nu_{n,\delta}$ the distribution of $U=VW$. Then $\cE_{\PM}^{(T)}*\nu_{n,\delta}$ is a strong approximate unitary $2$-design with measurable error
\begin{equation}
    \operatorname{Err}_{2}^{\rm full}
    \bigl(
        \cE_{\PM}^{(T)}*\nu_{n,\delta}
    \bigr)
    \leq
    \max\{2\eta,2\delta\}.
    \label{eq:shell-error}
\end{equation}
In particular, there exist universal constants $C_{\rm sh}\geq1$,
$0<c_{\rm sh}\leq1$, and $n_{\rm sh}\in\mathbb N$ such that, for
every even $n\geq n_{\rm sh}$ and every
$2e^{-c_{\rm sh}n}\leq\varepsilon_{2}\leq1$, the choices
$\eta=\delta=\varepsilon_{2}/2$ and
$T=\lceil C_{\rm mix}\log(2n/\varepsilon_{2})\rceil$ yield
measurable error at most $\varepsilon_{2}$ at total circuit depth
at most
\begin{equation}
    C_{\rm sh}\log\frac{n}{\varepsilon_{2}}.
    \label{eq:shell-total-depth}
\end{equation}
\end{theorem}

The proof is given in Appendix~\ref{app:proof-strong-two-shell}. \cref{thm:shell} supplies the logarithmic-depth global strong $2$-design shells required by the gluing construction. Sandwiching two brickwork layers of local strong $k$-designs between two independent strong 2-design shells allows to apply \cref{lem:gluing_strong_random_unitaries} iteratively, forming an extensive strong $k$-design (see the imported \cref{lem:patch_protocol}). This yields our main result:

\begin{theorem}[Strong $k$-design in logarithmic depth]\label{thm:strong-k}
There exist universal constants $0<c_0\leq 1$, $A,C\geq 1$, and $\xi_0\in\mathbb N$ such that the following holds. Let
\begin{equation}\label{eq:strong-k-parameter-range}
    n,k\in\mathbb N,
    \qquad
    n\ \text{is even},
    \qquad
    k\geq 2,
    \qquad
    0<\varepsilon\leq\frac12,
    \qquad
    \log\frac{nk}{\varepsilon}\leq c_0n.
\end{equation}
Define
\begin{equation}
    \xi := \left\lceil \frac{16}{3}\log_2\left(\frac{Ank^2}{\varepsilon}\right)\right\rceil+\xi_0,
    \qquad
    m:=\lfloor\frac{n}{\xi}\rfloor\in\mathbb{N}.
    \label{eq:strong-k-equal-patches}
\end{equation}
Then there exists an all-to-all circuit ensemble acting only on the $n$ physical qubits that forms a strong $\varepsilon$-approximate unitary $k$-design in measurable error and has depth
\begin{equation}\label{eq:strong-k-depth}
    d
    \leq
    C\bigl(1+k\log^7(2k)\bigr)
    \log\frac{nk}{\varepsilon}.
\end{equation}
For fixed $k$ and fixed measurable-error tolerance $\varepsilon < 1/4$, the upper bound is $O(\log n)$. The all-to-all two-qubit-gate model has a matching $\Omega(\log n)$ lower bound.
\end{theorem}

The constant $c_0$ is chosen small enough to ensure both that the patch size in \cref{eq:strong-k-equal-patches} satisfies $4\xi\leq n$ and that the required shell accuracy lies within the regime of \cref{thm:shell}. The circuit consists of two layers of local strong $k$-designs and two global strong $2$-design shells, with all local blocks and shells sampled independently. The factor $1+k\log^7(2k)$ in \cref{eq:strong-k-depth} is inherited from the available one-dimensional random-circuit construction of the local strong $k$-designs in \cref{lem:1D_strong_unitary_designs} and is not optimized here.
For fixed $k$ and fixed error, the upper bound is logarithmic in $n$ on the system sizes covered by \cref{thm:strong-k}; \cref{pro:lower_bound} gives the corresponding logarithmic lower bound, which establishes optimal asymptotic scaling $\Theta(\log n)$. The full proof is provided in \cref{app:proof-strong-k}.

\section{Methods and proof overview}\label{sec:methods}

The proof proceeds in four stages, each developed in a separate subsection.  In Section~\ref{sec:methods-reduction}, we reduce the robustness of the perfect-matching circuit against mixed-sign queries to a classical mixing problem for a Markov chain on Pauli strings. We then further simplify the chain onto the supports of the Pauli strings, viewed as nonempty subsets of $[n]$. In Section~\ref{sec:methods-coupling}, we analyze the resulting support chain using a grand coupling argument. This reduces the worst-case mixing problem over all $2^{n}-1$ nonempty initial supports to the coalescence of just $n+1$ coupled chains: the $n$ chains initialized from singletons of support $1$, and the chain initialized from full support. Section~\ref{sec:methods-coalescence} establishes this coalescence in two steps. First, we show that any singleton support grows to large size within $O(\log{n})$ steps with high probability. Second, we prove that, once a support is sufficiently large, it remains large and its discrepancy from the full-support chain contracts in logarithmic time. Finally, in Section~\ref{sec:methods-assembly}, we combine the perfect-matching circuit with a weak relative-error 2-design to obtain a strong 2-design in measurable error.
The gluing theorem then promotes this construction to a strong $k$-design in measurable error with the same logarithmic depth scaling. This section presents the main ideas of the argument, while the complete technical proofs are deferred to the corresponding appendices.
\subsection{From mixed queries to a classical Markov chain}
\label{sec:methods-reduction}

The perfect-matching ensemble $\mathcal{E}_{PM}^{(T)}$ from \cref{def:pm-ensemble} is invariant under conjugation, transposition, and random single-qubit Pauli rotations at the input and output circuit layers, required by \cref{pro:mixed_query_reduction} to hold. Controlling the mixed-sign sector therefore reduces to a classical Markov chain problem: we must show that, uniformly over all $P\in \mathcal{P}_n^*$, $p_{\mathcal{E}_{\rm PM}^{(T)}}(\,\cdot\,|\,P)$ approaches the Haar distribution $\pi$ in total variation distance after $T=O(\log(n))$ layers, even in the worst-case input string $P$.

Two structural properties, proved in Appendix~\ref{sec:The Markov Chains}, simplify this problem further. First, independent perfect-matching layers compose according to a classical Markov chain on Pauli strings. Second, after a single layer, the subsequent dynamics depend on the input Pauli only through its support: the local labels $X,Y,Z$ are completely randomized and retain no information about their initial values. The induced support update has then a particularly simple description. Fix a perfect-matching $M\sim \mathcal{M}_n$, and consider an edge $e\in M$. If the restriction of the current Pauli string to $e$ is $I\otimes I$, then both qubits remain outside the support. Otherwise, for an active
edge (one touching the support of $P$), Haar-random conjugation resamples the restriction uniformly from the $15$ nonidentity two-qubit Pauli strings. Among these, $9$ act nontrivially on both qubits, while $3$ act nontrivially only on either endpoint. Therefore, its two qubits join the new support together with probability $\tfrac{9}{15}=\tfrac35$, or one of them
alone with probability $\tfrac{3}{15}=\tfrac15$ each.  The support
$A_{t}=\mathrm{supp}(P_{t})$ is therefore itself a Markov chain on the nonempty subsets of $[n]$.

\begin{restatement}{Lemmas \ref{lem:pauli-chain-reduction},\ref{lem:stationarity},\ref{lem:tvd_equivalence} of Appendix~\ref{app:PauliMixing}, short form (reduction to a support Markov chain)}
Let $K_{n}$ denote the one-layer Pauli transition kernel, and let $\kappa_{n}$ denote the induced
kernel on nonempty supports $\mathcal{S}_{n}=\{A\subseteq[n]:A\neq\emptyset\}$.
Then for every $T\geq 1$ and every $P\in\mathcal{P}_{n}^{\star}$,
\begin{equation}\label{eq:support-reduction}
 p_{\mathcal{E}_{\rm PM}^{(T)}}(Q\,|\,P)=K_{n}^{T}(P,Q),
 \qquad
 \mathrm{TVD}\bigl(K_{n}^{T}(P,\cdot),\,\pi\bigr)
 =\mathrm{TVD}\bigl(\kappa_{n}^{T}(\mathrm{supp}(P),\cdot),\,\pi_{\rm supp}\bigr).
\end{equation}
The stationary distribution of the support chain is
$\pi_{\rm supp}(A)=3^{|A|}/(4^{n}-1)$. Equivalently, under $\pi_{\mathrm{supp}}$, the support size has the distribution of a
$\mathrm{Bin}(n,3/4)$ random variable conditioned to be nonzero.
\end{restatement}

The target stationary distribution 
has a simple interpretation: at equilibrium a random operator is typically supported on $3/4$ of the qubits. \cref{thm:mixing} is thus reduced to proving that the support chain reaches this equilibrium (Haar) profile from every nonempty initial support within $O(\log{n})$ layers.

\subsection{The grand coupling}\label{sec:methods-coupling}

Standard spectral-gap estimates cannot give a tight bound for Equation~\eqref{eq:support-reduction}~\cite{LevinPeresWilmer2017}. Converting a gap into total variation distance costs an exponential factor, so even a constant gap would only certify mixing at $T=O(n)$. We instead use a grand coupling argument for Markov chains~\cite{LevinPeresWilmer2017}, developed in detail in the Appendix~\ref{sec:Grand coupling between support chains}.

\begin{definition}[Grand coupling]\label{def:grand-coupling}
At each time $t$, sample a single matching $M_{t}\sim\mathcal{M}_n$. For every edge
$e=\{a,b\}\in M_{t}$, independently sample an outcome
\begin{equation}\label{eq:coupling-outcomes}
 Y_{t,e}=\{a,b\},\ \{a\},\ \{b\}
 \quad\text{with probabilities}\quad
 \tfrac35,\ \tfrac15,\ \tfrac15,
\end{equation}
respectively, shared by all initial supports. For a fixed realization, $w=\{M_t,Y_{t,e}\}_{0\leq t < T}$ of this shared randomness, each initial support $A_0$ evolves deterministically according to
\begin{equation}\label{eq:coupling-update}
 A_{t+1}=\Psi_{t}(A_{t})
 :=\bigcup_{\substack{e\in M_{t}\\ e\cap A_{t}\neq\emptyset}}Y_{t,e}.
\end{equation}
\end{definition}
Thus, each chain collects the sampled outcome of every active edge that touches it and ignores all inactive edges. Because every chain is driven by the same randomness, the coupled dynamics satisfy two deterministic properties, immediate
from \cref{eq:coupling-update}:

\begin{restatement}{Lemma \ref{lem:coupling_facts} of Appendix~\ref{sec:Grand coupling between support chains}, short form (coalescence and monotonicity)}
For every fixed realization $w$ of the shared randomness:
(a)~Coalescence: if $A_{t}=B_{t}$ then $A_{s}=B_{s}$ for all $s\geq t$.
(b)~Monotonicity: if $A_{t}\subseteq B_{t}$ then $A_{s}\subseteq B_{s}$ for all $s\geq t$.
\end{restatement}
Coalescence turns mixing into a collision problem: if we show that, apart from an exponentially small fraction of realizations $w$, all initial supports collide in $T=O(\log(n))$ time, such a collision will evolve identically for any further time (see \cref{fig:grand-coupling-overview}). Monotonicity then reduces the number of initial conditions that must be tracked. Let $A_{T}^{(i)}$ denote the chain initialized at singleton $(i)$, let $A_{T}^{(\mathrm{full})}$ denote the chain initialized at $([n])$, and let $A_{T}^{(S)}$ denote the chain initialized at an arbitrary nonempty support $(S)$. For every $i\in S$,  $A_{T}^{(i)}\subseteq A_{T}^{(S)}\subseteq A_{T}^{(\mathrm{full})}$ at all times $t$. We also write $A_t^{(\pi)}$ for a chain whose initial support is sampled from $\pi_{\rm supp}$. Then, we obtain:

\begin{restatement}{Lemmas \ref{lem:coupling_inequality}-\ref{lem:singleton_reduction} of Appendix~\ref{sec:Grand coupling between support chains}, short form (coupling inequality and reduction to singletons)}
For every $A\in\mathcal{S}_{n}$ and every $T\geq0$,
\begin{equation}\label{eq:coupling-inequality}
 \mathrm{TVD}\bigl(\kappa_{n}^{T}(A,\cdot),\,\pi_{\rm supp}\bigr)
 \;\leq\;\Pr\bigl[A_{T}\neq A_{T}^{(\pi)}\bigr]
 \;\leq\;\sum_{i=1}^{n}\Pr\bigl[A_{T}^{(i)}\neq A_{T}^{(\mathrm{full})}\bigr].
\end{equation}
\end{restatement}
The worst-case mixing problem over all $2^{n}-1$ nonempty initial supports is therefore reduced to $n+1$ coupled chains: the $n$ initial singletons and the full-support chain.  Once every singleton chain has coalesced with the full-support chain, monotonicity forces every chain initialized between them to follow the same trajectory, independently on its initial support. 

\begin{figure}[t]
    \centering

    \begin{subfigure}[t]{0.48\textwidth}
        \centering
        \includegraphics[width=\linewidth]{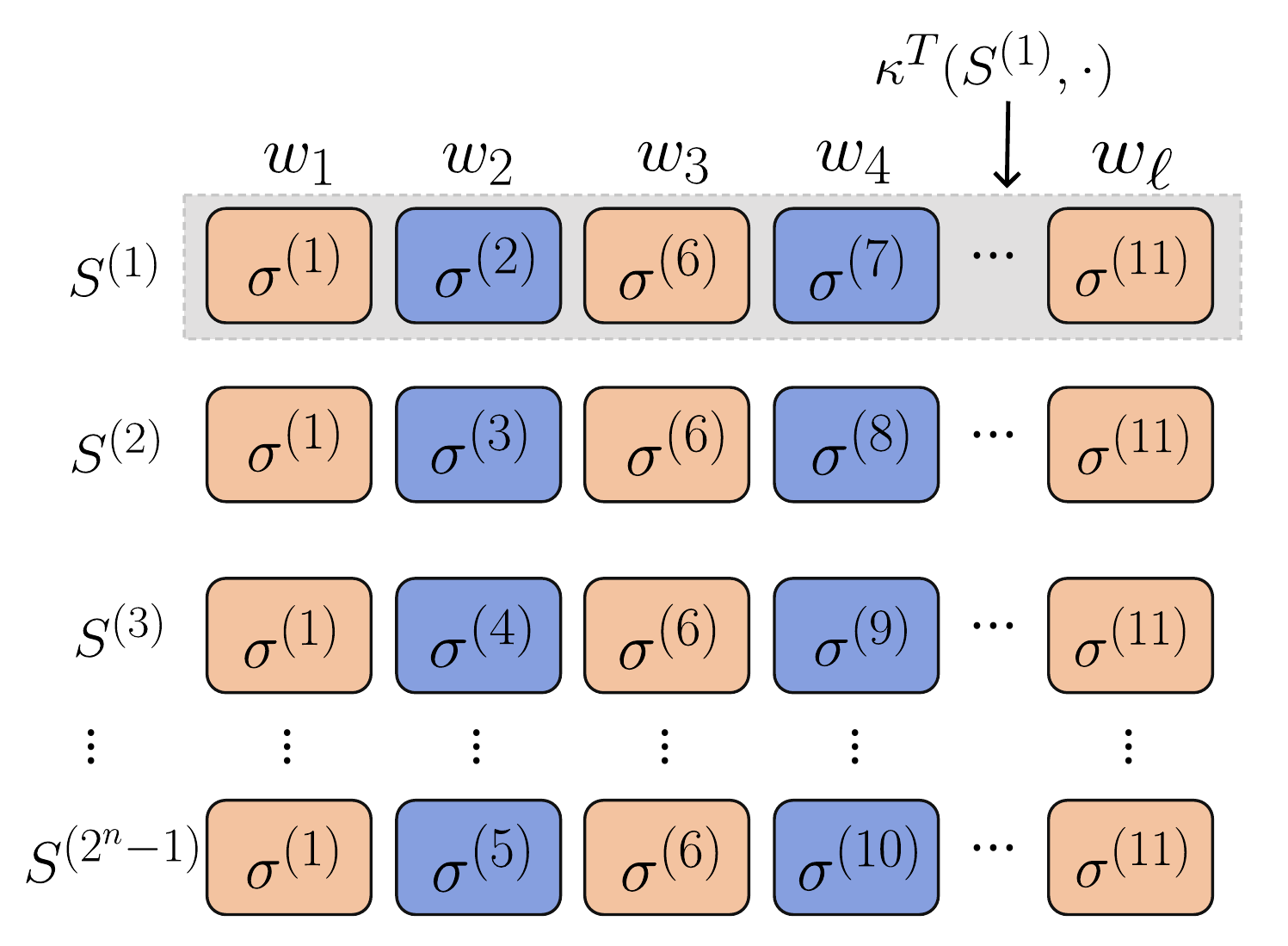}
        \caption{The grand coupling strategy}
        \label{fig:grand-coupling-overview-panel-a}
    \end{subfigure}
    \hfill
    \begin{subfigure}[t]{0.48\textwidth}
        \centering
        \includegraphics[width=\linewidth]{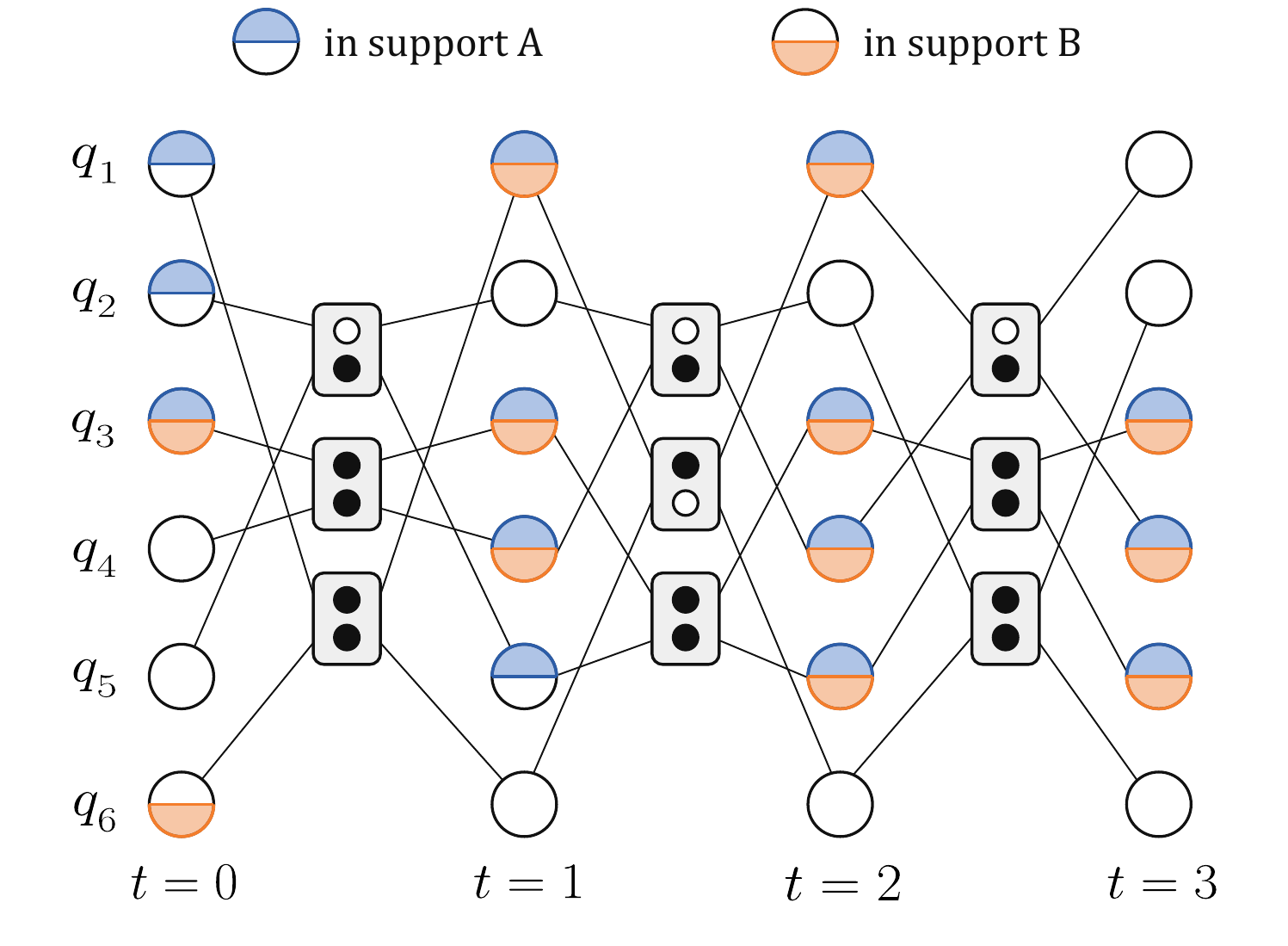}
        \caption{Exemplified deterministic process}
        \label{fig:grand-coupling-overview-panel-b}
    \end{subfigure}
    \caption{(a)~Schematic of the grand coupling. Each row is indexed by one of the $2^n-1$ supports $S^{(i)} \in \mathcal{S}_n$, while each column corresponds to a realization $w_j=\{M_t,(Y_{t,e})_{e\in M_t}\}_{0\leq t < T}$ of the shared randomness. The entry $\sigma^{(i)}$ is the support obtained by evolving the initial row-index support $S^{(i)}$ under the column-index deterministic process $w_j$. Instead of comparing the probability distributions of each support after going through $T$-steps of the Markov chain (rows), we fix the possible random deterministic instances $w$ of the chain, and study the deterministic evolution of all supports under such realization (columns). By showing that outside an exponentially small fraction of columns all the elements of the column are equal (orange columns), we may infer properties about the probability distributions in the rows. (b)~Example of a fixed realization $w_j=\{M_t,(Y_{t,e})_{e\in M_t}\}_{0\leq t < T}$ for $n=6$ and two initial supports $A=\{1,2,3\}$ (blue) and $B=\{3,6\}$ (orange). The fixed $Y_{t,e}$ schedule updates the supports at each layer: if an active edge is $\{a,b\}=:\{\bullet,\bullet\}$, it outputs support in both endpoints, while $\{a\}=:\{\bullet,\circ\}$ and $\{b\}=:\{\circ, \bullet\}$ only in one. The two chains evolve deterministically under the same process $w_j$, and coalesce at $t=2$. By the coalescence property, they remain identical at all subsequent times.}
    \label{fig:grand-coupling-overview}
\end{figure}

\subsection{Coalescence in logarithmic time}\label{sec:methods-coalescence}

We now show that, with probability at least $1-\eta$, all $n$ singleton chains coalesce with the full-support chain within $T=O(\log(n/\eta))$ layers. As illustrated in \cref{fig:PauliMarkovMixing}, the proof proceeds in two phases, each within $O(\log{(n)})$ layers. First, a chain initialized from singletons reaches support size at least $n/2$. This growth phase is analyzed in Appendix~\ref{sec:growth}. Second, once the support is large, it remains large with high probability, while its discrepancy from the full-support chain contracts geometrically. This coalescence phase is treated in Appendix~\ref{sec:uniformmix}.
\subsubsection{Phase 1: growth}
Let $H_{t}=|A_{t}|$ denote the support size at time $t$, and let $J_{t}$ be the number of edges of the sampled matching $M_t$ whose two endpoints both lie in $A_{t}$. The number of active edges is then $R_{t}=H_{t}-J_{t}$. Each active edge contributes at least one qubit to the next support, and with
probability $3/5$ to both. Consequently, conditional on the current support and matching, the exact one-step structure is
\begin{equation}\label{eq:overview-one-step}
 H_{t+1}=R_{t}+B_{t},\qquad B_{t}\sim\mathrm{Bin}\bigl(R_{t},\tfrac35\bigr).
\end{equation}
The single-step weight update from Eq.~\eqref{eq:overview-one-step} is therefore a probability distribution, which can at most change the support by a factor of two. Let $\mathcal{F}_t$ denote the information revealed before sampling the $t$-th layer, which we define precisely in the Appendix~\ref{sec:growth}. We first control the multiplicative drift of a single step by averaging over the random matchings:

\begin{restatement}{Lemma \ref{lem:drift} of Appendix~\ref{sec:growth-step}, short form (multiplicative drift)}
Conditional on $\mathcal{F}_{t}$, $\E[J_{t}\,|\,\mathcal{F}_{t}]=H_{t}(H_{t}-1)/(2(n-1))$. Consequently,
\begin{equation}\label{eq:overview-drift}
 \E[H_{t+1}\,|\,\mathcal{F}_{t}]
 =\tfrac85H_{t}\Bigl(1-\tfrac{H_{t}-1}{2(n-1)}\Bigr)
 \;\geq\;\tfrac65\,H_{t}
 \qquad\text{while }H_{t}\leq n/2.
\end{equation}
\end{restatement}

\begin{figure}[t]
    \centering

    \begin{subfigure}[t]{0.85\linewidth}
        \centering
        \includegraphics[width=\linewidth]{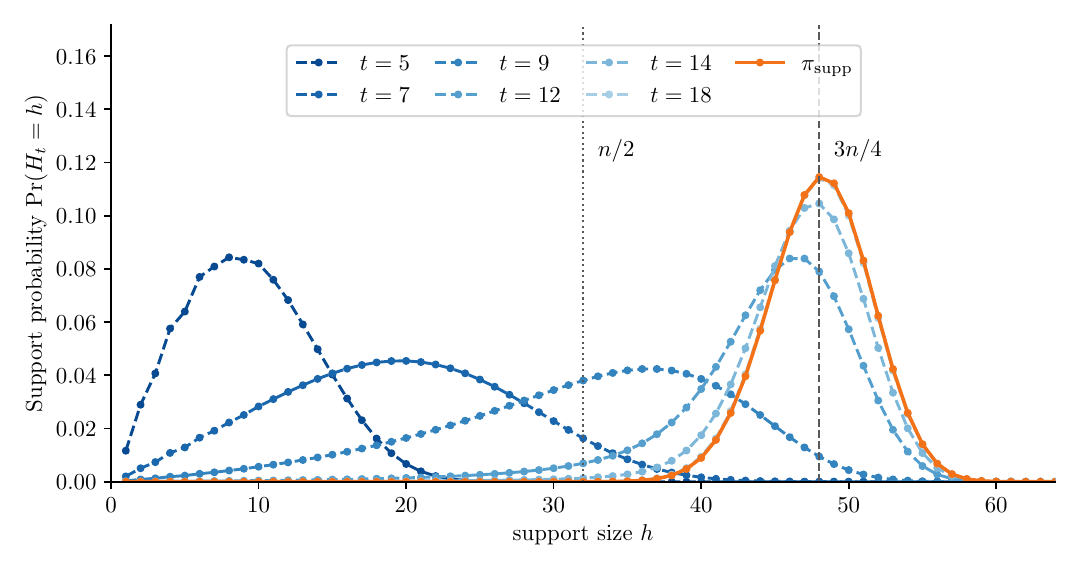}
        \caption{Singleton mixing towards the uniform distribution}
        \label{fig:PauliMarkovMixing-a}
    \end{subfigure}

    \vspace{0.5em}

    \begin{subfigure}[t]{0.42\linewidth}
        \centering
        \includegraphics[width=\linewidth]{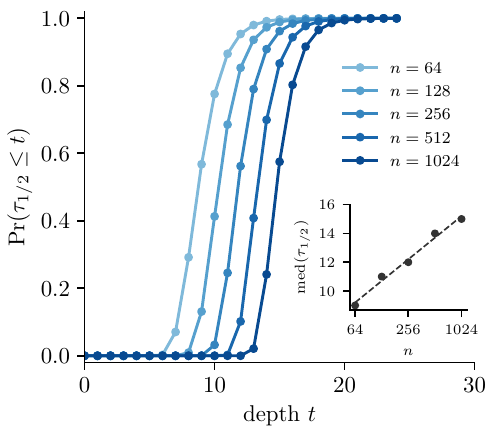}
        \caption{Phase 1: growth}
        \label{fig:PauliMarkovMixing-b}
    \end{subfigure}
    \hspace{5mm}
    \begin{subfigure}[t]{0.42\linewidth}
        \centering
        \includegraphics[width=\linewidth]{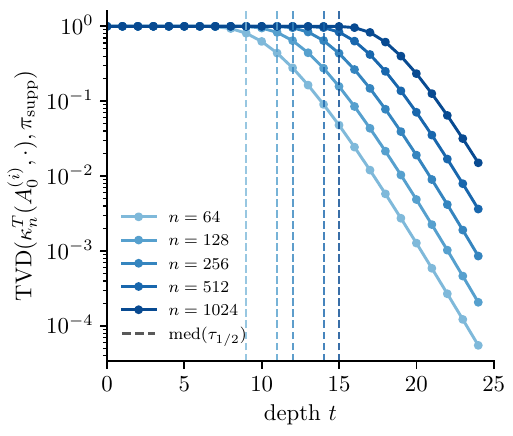}
        \caption{Phase 2: coalescence}
        \label{fig:PauliMarkovMixing-c}
    \end{subfigure}

    \caption{
    (a) Mixing of the support chain induced by the perfect-matching circuit, exemplified numerically for $n=64$. An initial singleton evolves towards the uniform support distribution $\pi_{\rm supp}$ after $O(\log n)$ steps. (b) Starting from a singleton of support size $H_0=1$, the support density grows at each step. The figure shows the cumulative probability  $Pr(\tau_{1/2}\leq t)$ of achieving the hitting time after $t$ steps in \cref{lem:hitting-linear}. After $O(\log n)$ steps, the cumulative probability contracts geometrically towards 1 with further steps. The inset, which has the x axis in $\log$ scale, shows the $O(\log n)$ scaling of the median hitting time $med(\tau_{1/2})$, approximately equal to the first $t$ when $Pr(\tau_{1/2}\leq t)>1/2$. 
    (c) After the growth phase, dense supports remain
    dense up to $e^{-\Omega(n)}$ corrections (\cref{lem:dense-persistence}) while discrepancies
    between coupled chains contract geometrically per layer (\cref{lem:discrepancy}). The error in total variational distance between the chain and the uniform distribution decays exponentially after the hitting time.
    }
    \label{fig:PauliMarkovMixing}
\end{figure}

In Appendix~\ref{sec:growth-step}, we identify the two sources of randomness that fluctuate around the drift. The first comes from the random matching $M_t$, which determines the number $R_t$ of active edges and is analyzed in Lemma~\ref{lem:exposure}. Conditional on the matching, the second source is the binomial randomness in the edge updates, controlled in Lemma~\ref{lem:growthstep}. We also account for the possibility that very small supports can temporarily fail to grow. For example, a singleton remains a singleton after one layer with probability $2/5$. Such stalling events are treated in Lemma~\ref{lem:escape}. A stopped-process argument combines these three into the following logarithmic hitting-time bound in Lemma \ref{lem:hitting-linear}:

\begin{restatement}{Lemma \ref{lem:hitting-linear} of Appendix~\ref{sec:growth-step}, short form (hitting half density)}
There exists a universal constant $C_{\rm hit}>0$ such that, for every chain initialized with $H_{0}=1$ and every $0<\delta\leq1/2$, the hitting time
$\tau_{1/2}=\inf\{t:H_{t}\geq n/2\}$ satisfies
\begin{equation}\label{eq:hitting}
 \Pr\Bigl[\tau_{1/2}>C_{\rm hit}\bigl(\log n+\log\tfrac1\delta+1\bigr)\Bigr]\leq\delta.
\end{equation}
\end{restatement}
Thus, any singleton grows with probability $1-\delta$ to size $\geq n/2$ in time $T=O(\log(n/\delta))$.

\subsubsection{Phase 2: persistence and contraction}

We next show that, once a support reaches size $\geq n/2$, it is unlikely to return to the sparse regime: 

\begin{restatement}{Lemma \ref{lem:dense-persistence} of Appendix~\ref{sec:uniformmix}, short form (dense persistence)}
On the event $H_{t}\geq n/2$,
$\ \Pr[H_{t+1}<n/2\,|\,\mathcal{F}_{t}]\leq2e^{-c_{\rm d}n}$ with the
explicit constant $c_{\rm d}=1/3200$.
\end{restatement}

Finally, we consider two coupled chains
$A_{t}\subseteq\widetilde{A}_{t}$. The discrepancy $d_{t}=|\widetilde{A}_{t}\setminus A_{t}|$ can persist only through a matching edge that is active for the upper chain $\widetilde A_t$ but inactive for the lower chain $A_t$. Once $A_t$ occupies at least half of the qubits, we show that such edges are sufficiently unlikely that the expected discrepancy contracts geometrically.

\begin{restatement}{Lemma \ref{lem:discrepancy} of Appendix~\ref{sec:uniformmix}, short form (gap contraction)}
If $|A_{t}|\geq n/2$ and $n\geq6$, then
$\ \E[d_{t+1}\,|\,A_{t},\widetilde{A}_{t}]\leq\varrho\,d_{t}$ with
$\varrho=24/25$.
\end{restatement}

Then, the proof of mixing combines growth, persistence, and contraction with appropriately chosen error budgets $\eta/16$ each. We first apply the hitting-time estimate to all $n$ singleton chains and take a union bound. We then control the probability that these chains return below half support during the contraction phase. This step introduces the restriction $\eta\geq e^{-\Omega(n)}$, since the per-layer failure probability is exponentially small in $n$. Finally, after the $O(\log(n/\eta))$ contraction
layers, Markov's inequality shows that every singleton chain has coalesced with the full-support chain, except with probability $O(\eta)$. The coupling bound in Eq.~\eqref{eq:coupling-inequality}, followed by the support reduction in Eq.~\eqref{eq:support-reduction}, then yields the desired uniform Pauli-mixing estimate:

\begin{restatement}{Theorem~\ref{thm:mixing} of Appendix~\ref{sec:uniformmix}
(uniform Pauli mixing)}
There exist universal constants $C_{\mathrm{mix}},c_{\mathrm{mix}}>0$ and $n_{\mathrm{mix}}$ such that the following holds. For every even $n\ge n_{\mathrm{mix}}$, every
$e^{-c_{\mathrm{mix}}n}\le\eta\le\tfrac12$, and every depth
$T\ge C_{\mathrm{mix}}\log\tfrac n\eta$,
\begin{equation}\label{eq:mixing}
  \max_{P\in\mathcal P_n^\star}\mathrm{TVD}\big(K_n^T(P,\cdot),\,\pi\big)\;\le\;\eta.
\end{equation}
\end{restatement}
Together with \cref{pro:mixed_query_reduction}, this establishes that the perfect-matching ensemble is robust against mixed-sign queries in measurable error.

\subsection{Assembling the strong scrambler}\label{sec:methods-assembly}

\cref{thm:mixing} and \cref{pro:mixed_query_reduction} together make the perfect-matching circuit robust in the mixed-sign sector. The remaining task is to assemble the entire strong $k$-design in measurable error. To do so, we reverse the reduction procedure we made in \cref{sec:composing_2} and \cref{sec:assembling}, which we illustrated in \cref{fig:three-subfigures}. 

The construction proceeds in two steps. First, Theorem~\ref{thm:mixing} and Proposition~\ref{pro:mixed_query_reduction} show that the perfect-matching ensemble controls the mixed-sign sector in logarithmic depth. We compose it with an independent weak relative-error $2$-design, which controls the same-sign sector. Applying Lemma~\ref{lem:composition} then yields the logarithmic-depth strong $2$-design shell of Theorem~\ref{thm:shell}. Its proof is given in Appendix~\ref{app:proof-strong-two-shell}.
Second, we place two such shells around the brickwork layers of local strong $k$-designs and apply the strong gluing lemma, \cref{lem:gluing_strong_random_unitaries}, iteratively. This produces the strong $k$-design in measurable error stated in \cref{thm:strong-k}, with total depth logarithmic in $n$ for fixed $k$ and fixed accuracy.

\section{Discussion}
\label{sec:discussion}
Our main theorems resolve the unitary-design part of the strong fast-scrambling conjecture introduced in Ref.~\cite{SchusterEtAlStrong2025}. For every fixed design order $k$ and fixed measurable-error, \cref{thm:strong-k} constructs strong approximate unitary $k$-designs on the original physical qubits of the system in depth $O(\log n)$. A light-cone argument gives the corresponding $\Omega(\log n)$ lower bound, which is saturated, and therefore proves the optimal depth $\Theta(\log n)$.

The result applies to fixed orientation query strategies, with arbitrary adaptive quantum operations and quantum memory between the queries. It therefore covers bounded-query oracle formulations of OTOCs and echo protocols~\cite{LarkinOvchinnikov1969,Kitaev2015,RobertsStanford2015,SwingleEtAl2016,GarttnerEtAl2017}, as well as Hayden--Preskill decoding~\cite{HaydenPreskill2007,YoshidaKitaev2017}, whenever they fit this comb model. The theorem does not by itself cover controlled queries, adaptive choices of query orientation, noisy oracles, or physical implementation.

The main technical result is uniform Pauli mixing for the perfect-matching circuit, which is then simplified to a classical Markov chain on Pauli supports. Random perfect matchings spread every nonidentity Pauli operator to macroscopic support, after which the discrepancy between coupled chains contracts toward the Haar-induced Pauli distribution. A grand-coupling argument makes this convergence uniform over the initial Pauli operator.
Combining this result with a weak relative-error design and the strong gluing construction yields strong designs of arbitrary fixed order.

The construction leaves open whether all of these additional ingredients are necessary. In particular, our proof does not determine whether the perfect-matching ensemble alone forms a strong approximate $2$-design in logarithmic depth. A direct higher-moment analysis might also avoid the present gluing architecture and improve the dependence on the design order. The factor $k\log^7(2k)$ in \cref{thm:strong-k} is inherited from strong gluing, and might leave margin for improvement.

The roles of the error model and the interaction geometry also deserve further study. Measurable error captures bounded-query operational tests while avoiding the linear-depth obstruction known for relative-error strong designs~\cite{SchusterEtAlStrong2025}. Intermediate notions of approximation may control broader classes of experiments without requiring full relative error. Likewise, extending the support-mixing argument to sparse expanders, small-world networks, and time-dependent interaction graphs could help identify which connectivity properties are sufficient for logarithmic-depth operator spreading and strong-design behavior.

Finally, the optimal depth of strong pseudorandom unitaries remains open.
The present theorem gives a logarithmic-depth information-theoretic strong-design construction, but it does not provide a keyed family or computational security against efficient distinguishers. A logarithmic-depth, system-only construction of strong pseudorandom unitaries would require a suitable keyed local primitive together with a security-composition theorem that remains valid under inverse, transpose, and complex-conjugate queries.

\section{AI disclosure}
The idea for this project—attacking the strong fast-scrambling conjecture by reducing the analysis of perfect-matching circuits to a Markov-chain mixing problem—was conceived at QIP 2026, and we, the human authors, have worked on the problem throughout the year. The conceptual development and the bulk of the mathematical work were carried out without assistance from LLMs. GPT-5.6 Pro was used only at the final stage to help fill gaps in the full proof of Pauli mixing.

\section{Acknowledgements}
We thank Antonio Acín, Fernando Brandão, Jens Eisert, Soumik Ghosh, Jonas Haferkamp, Hsin-Yuan Huang, Shivan Mittal, and Thomas Schuster for discussions.  
T.P.D. and J.B.R. have been supported by the Government of Spain (Severo Ochoa CEX2019-000910-S and FUNQIP), Fundació Cellex, Fundació Mir-Puig, Generalitat de Catalunya (CERCA program). T.P.D. has been supported by grant PREP2022-000452, funded by MCIN/AEI/10.13039/501100011033 and the European Social Fund Plus. J.B.R. has received funding from the “Secretaria d’Universitats i Recerca del Departament de Recerca i Universitats de la Generalitat de Catalunya” under grant FI-3 00096, as well as the European Social Fund Plus. A. A. M. and S.F.E.O. have been supported by the BMFTR (DAQC, MuniQC-Atoms, QuSol, Hybrid++, PasQuops), Clusters of Excellence (ML4Q, MATH+), the QuantERA, the Munich Quantum Valley, Berlin Quantum, the Quantum Flagship (Millenion, Pasquans2), the DFG (CRC 183, SPP 2514), the European Research Council (DebuQC), PraktiQOM. A.A.M. further acknowledges support from a 2025 Google PhD Fellowship.
\bibliography{ref}

\newpage

\appendix 
\resumetoc

\setcounter{secnumdepth}{2}
\setcounter{equation}{0}
\setcounter{figure}{0}
\setcounter{table}{0}
\setcounter{section}{0}
\renewcommand{\thetable}{A\arabic{table}}
\renewcommand{\theequation}{A\arabic{section}.\arabic{equation}}
\renewcommand{\thefigure}{A\arabic{figure}}
\renewcommand{\thesection}{A.\arabic{section}} 
\renewcommand{\thesubsection}{\thesection.\arabic{subsection}}
\renewcommand{\thesubsubsection}{\thesubsection.\arabic{subsubsection}}
\makeatletter
\renewcommand{\p@subsection}{}
\renewcommand{\p@subsubsection}{}
\makeatother
\begin{center}
\textbf{\large Appendix}
\end{center}
\tableofcontents 

\newpage
\section{Preliminaries}\label{app:preliminaries}

\subsection{Strong designs in measurable error and the comb formalism}\label{app:definitions}

\subsubsection{Moment channels. Additive and relative approximation errors}
Set \(d=2^n\).  For \(p,q\geq0\) with \(p+q=k\), we define the mixed-moment channel of an ensemble $\mathcal{E}$ as
\begin{equation}\label{eq:mixed-twirl}
 \Phi_{\cE}^{(p,q)}(X)
 :=\E_{U\sim\cE}\!\left[
 (U^{\otimes p}\otimes U^{*,\otimes q})\,X\,
 (U^{\dagger,\otimes p}\otimes U^{T,\otimes q})
 \right],
\end{equation}
acting on operators \(X\) on \((\mathbb{C}^{d})^{\otimes k}\). It describes a single parallel query round: \(p\) registers receive the sampled unitary \(U\) and the remaining \(q\) registers receive its complex conjugate \(U^{*}\).  For \(q=0\) it reduces to the usual \(k\)-th moment (twirling) channel, and an exact unitary \(k\)-design agrees with the Haar measure on it exactly. We can define the following two weak error notions, which quantify in different ways the error in the sector \((p,q)=(k,0)\):
 
\begin{definition}[Weak additive error]\label{def:weak-additive-error}
The \emph{weak additive error} $\varepsilon_a$ of \(\cE\) at order \(k\) is the diamond distance between the forward moment channels.
The ensemble \(\cE\) is a \emph{weak additive-error \(k\)-design} with error
\begin{equation}\label{eq:weak-additive}
 \varepsilon_a
 :=\bigl\|\Phi_{\cE}^{(k,0)}-\Phi_{\Haar}^{(k,0)}\bigr\|_{\diamond}
 =\sup_{\rho}\,
 \Bigl\|\Bigl(\bigl[\Phi_{\cE}^{(k,0)}-\Phi_{\Haar}^{(k,0)}\bigr]\otimes\mathbb{I}\Bigr)(\rho)\Bigr\|_{1},
\end{equation}
where the supremum runs over all density operators \(\rho\) on the query register together with an additional register.
\end{definition}
The most general such experiment prepares a state \(\rho\) on the \(k\) query copies and an arbitrarily entangled additional register, applies the sampled unitary to all \(k\) copies in parallel, and performs an arbitrary joint measurement.  By the Helstrom bound~\cite{Helstrom1976}, the optimal acceptance-probability gap between \(\cE\) and Haar over all such experiments equals one half of \cref{eq:weak-additive}.

However, the weak additive error only ensures security under parallel queries of the sampled unitary. Instead, the weak relative error represents the strongest notion of approximation error for unitary designs, ensuring adaptive security even to properties that cannot be efficiently measured:
 
\begin{definition}[Weak relative error]\label{def:weak-relative-error}
The ensemble \(\cE\) is a \emph{weak relative-error \(k\)-design} with error \(\varepsilon_r\geq0\) if
\begin{equation}\label{eq:weak-relative}
 (1-\varepsilon_r)\,\Phi_{\Haar}^{(k,0)}
 \;\preceq_{\rm CP}\;
 \Phi_{\cE}^{(k,0)}
 \;\preceq_{\rm CP}\;
 (1+\varepsilon_r)\,\Phi_{\Haar}^{(k,0)},
\end{equation}
where \(\Psi\preceq_{\rm CP}\Phi\) means that \(\Phi-\Psi\) is completely positive.  Equivalently, the Choi operators obey
\((1-\varepsilon_r)\,J\bigl(\Phi_{\Haar}^{(k,0)}\bigr)\leq J\bigl(\Phi_{\cE}^{(k,0)}\bigr)\leq(1+\varepsilon_r)\,J\bigl(\Phi_{\Haar}^{(k,0)}\bigr)\)
as operator inequalities.
\end{definition}
Both definitions are statements about $(k,0)$ queries, so these are the weak definitions of design. For completeness, if the conditions \cref{eq:weak-additive}, \cref{eq:weak-relative} are satisfied for any possible $(p,q)$ choice (with $p+q=k$), this represents the upgrade to their strong versions. The strong additive error is then given by
\begin{equation}
 \max_{p+q=k}\left\|\Phi_{\cE}^{(p,q)}-\Phi_{\Haar}^{(p,q)}\right\|_\diamond,
\end{equation}
while the strong relative error \(\varepsilon_r\) is
\begin{equation}
 (1-\delta)\Phi_{\Haar}^{(p,q)}\preceq_{\rm CP}
 \Phi_{\cE}^{(p,q)}\preceq_{\rm CP}
 (1+\delta)\Phi_{\Haar}^{(p,q)}
 \quad\text{for all }p+q=k.
\end{equation}

\subsubsection{Approximate strong designs in measurable error}
The relative error is a powerful condition, as it provides security even for events no physical experiment can efficiently resolve. For strong designs this strength is fatal, since relative error requires depth $\Omega(n)$ on any geometry~\cite{SchusterEtAlStrong2025}. The measurable error notion introduced in \cite{cui2025unitarydesignsnearlyoptimal} relaxes the relative error condition to adaptive security without the restriction of multiplicative error. In this appendix, we reformulate its definition from the point of view of quantum combs \cite{ChiribellaEtAl2009}.

Fix a choice of queries \(\bm s=(s_1,\ldots,s_{k})\in\{1,\dagger,T,*\}^k\).  A \(k\)-slot comb \(\cW\) has memory registers \(\mathsf{M}_0,\ldots,\mathsf{M}_k\), an initial state \(\rho_0\) on \(\mathsf{Q}\otimes \mathsf{M}_0\), interleaving channels with compatible input and output spaces, and a final measurement.  Writing
\[
 \mathcal V_{s_j}^{U}(X)=U^{s_j}X(U^{s_j})^\dagger,
\]
the output state is
\begin{equation}\label{eq:comb-output}
 \rho_{\cW,\bm s}^{U}
 =\cW_{k+1}\circ(\mathcal V_{s_k}^{U}\otimes\operatorname{id}_{\mathsf{M}_k})
 \circ\cW_{k}\circ\cdots\circ
 (\mathcal V_{s_1}^{U}\otimes\operatorname{id}_{\mathsf{M}_1})
 \circ\cW_1(\rho_0).
\end{equation}
The maps \(\cW_j\) may enlarge or discard private memory. By purification and deferred measurement this represents an arbitrary adaptive strategy for the given word $\bm s$~\cite{ChiribellaEtAl2009}.

\begin{definition}[Strong measurable error]
An ensemble \(\cE\) is a strong \(\varepsilon_m\)-approximate unitary \(k\)-design in measurable error if
\begin{equation}\label{eq:measurable-error}
 \operatorname{Err}_{k}^{\rm full}:=\sup_{\bm s,\cW}
 \left\|
 \E_{U\sim\cE}\rho_{\cW,\bm s}^{U}
 -\E_{U\sim\Haar}\rho_{\cW,\bm s}^{U}
 \right\|_1\leq\varepsilon_m.
\end{equation}
The supremum ranges over all orientation words $\bm s$ and all finite-memory $k$-slot strategies. The word is fixed before the interaction, while the strategy for each fixed word is fully adaptive.
\end{definition}

Two normalizations of \cref{eq:measurable-error} appear in the literature and in the imported statements we use throughout our work. We write $ \operatorname{Err}_{k}^{\rm full}$ for the full trace norm, as displayed
above, and  $\operatorname{Err}_{k} :=\frac{1}{2}  \operatorname{Err}_{k}^{\rm full}$ for the half trace distance, which is the operationally normalized
distinguishing advantage. We will make precise which one we are using in each case. All our final statements about measurable error are quoted in the full trace norm of Eq.~\eqref{eq:measurable-error}.

The strong measurable error is the principal error notion used in this work for our theorems. Relative error will also play an important role in the imported strong gluing lemmas.

\newpage
\section{Pauli mixing in perfect-matching circuits} \label{app:PauliMixing}
In this appendix we show that the ensemble of perfect-matching circuits of depth $T=O(\log(n/\eta))$, mixes the Pauli-strring chain second-moment Markov chain to total variation error $\eta$. First, in \cref{sec:The Markov Chains} we define the classical Markov chains associated to the perfect-matching ensemble, and reduce the Pauli strings chain to a support chain. In \cref{sec:Grand coupling between support chains}, we introduce the grand coupling strategy to resolve the support Markov chain. We then show in \cref{sec:growth} that singleton supports of size $1$ reach at least size $n/2$ within $O(\log n +\log (1/\delta))$ perfect-matching layers, for every failure probability $0<\delta \leq 1/2$. Finally, in \cref{sec:uniformmix} we show that supports size persist above $n/2$ and that the support distribution converges with high probability towards the uniform distribution.

Throughout the appendix, $n\geq 2$ is even. We write $\log=\log_2$, and $\operatorname{ln}$ is the natural logarithm.

\subsection{Markov chains}\label{sec:The Markov Chains}
In this section, we first show in \cref{sec:The Markov Chain and Pauli strings} how the action of
an ensemble over Pauli strings can be represented by a classical Markov chain. We then apply this representation to the perfect-matching ensemble in \cref{sec:Pauli_Mixing}, and characterize the corresponding Haar target distribution in \cref{sec:pauli_haar}. We then simplify the Pauli Markov chain into a Markov chain of supports in \cref{sec:supports}, which finally provides a single mixing condition on supports in \cref{eq:pauli_support_identity} which is used throughout the rest of this appendix.

\subsubsection{Markov chain on Pauli strings}
\label{sec:The Markov Chain and Pauli strings}
We consider the set of Pauli strings $\mathcal{P}_n=\{\mathbb{I},X,Y,Z\}^{\otimes n}$, and its non-trivial part $\mathcal{P}^{\star}_n=\mathcal{P}_n \backslash \{\mathbb{I}^{\otimes n}\}$. We just take the Hermitian representatives without extra phases, which satisfy $\operatorname{Tr}(PQ)=D\delta_{P,Q}$ with $D=2^n$. The support of a string $P \in \mathcal{P}_n^{\star}$ is $\operatorname{supp}(P)=\{i\in [n]:P_i \neq \mathbb{I}\}$, and its weight is $|\operatorname{supp}(P)|$. 

The conjugation of a non-trivial Pauli string $P$ by a unitary $U$ drawn from an ensemble $\mathcal{E}$ is 
\begin{equation}
    U P U^{\dagger} = \sum_{Q \in \mathcal{P}_n^{\star}} c_Q(U)\, Q,
    \qquad
    c_Q(U) = \frac{1}{D}\operatorname{Tr}\!\left(Q\, U P U^{\dagger}\right).
\end{equation}
Given that Pauli strings form an orthogonal basis of the operator space.
The identity coefficient vanishes because
$c_{\mathbb{I}^{\otimes n}}(U) = \operatorname{Tr}(U P U^{\dagger})/D
= \operatorname{Tr}(P)/D = 0$. The squared coefficients form a probability distribution with $\sum_Q |c_Q(U)|^2=1$. Averaging this action over $\mathcal{E}$, we define
\begin{equation}
\label{eq:probabilities-paulis}
    p_{\mathcal{E}}(Q \,|\, P) := \frac{1}{D^{2}}\, \mathbb{E}_{U \sim \mathcal{E}} \left| \operatorname{Tr}\!\left(Q\, U P U^{\dagger}\right) \right|^{2},
\end{equation}
as the probability that the string $P$, conjugated by a random circuit $U\sim \mathcal{E}$, lands on string $Q$. The unitary does not literally map a Pauli string to a single Pauli string. We instead refer to the induced classical distribution of squared Pauli coefficients. In particular, $p_{\mathcal{E}}(\cdot \,|\, P)$ is the probability distribution over all possible output strings $Q$ after sending $P$ through the ensemble $\mathcal{E}$. This is precisely the probability distribution that we need to analyze for the mixed-query criterion.

\subsubsection{The Pauli Markov chain on the perfect-matching ensemble}\label{sec:Pauli_Mixing}
We are interested in understanding the transition probabilities in \cref{eq:probabilities-paulis} for the perfect-matching ensemble of depth $T$, that is, when $\mathcal{E}=\mathcal{E}_{PM}^{(T)}$. For simplicity, we first study the case for a single layer $T=1$, and then we extend the construction to arbitrary depth $T$.

We now define the set of perfect matchings as
\begin{equation}
    \mathcal{M}_n
    :=
    \left\{
        M \subseteq \binom{[n]}{2}
        \;:\;
        \text{each vertex in $[n]$ is incident to exactly one edge in $M$}
    \right\}.
\end{equation}
A perfect-matching layer is obtained by randomly sampling a perfect-matching $M \sim \operatorname{Unif}(\mathcal{M}_n)$ on $n$ qubits, and then applying an independent Haar-$SU(4)$ gate to each of its $n/2$ edges. Conditioned on $M$, the induced second-moment transition kernel factorizes over the edges and takes a remarkably simple form. For an edge $\{i,j\} \in M$, let $P_{\{i,j\}}$ denote the restriction of $P$ to qubits $i$ and $j$. Then:
\begin{enumerate}
    \item If $P_{\{i,j\}} =\mathbb{I}\otimes\mathbb{I}$, it is left untouched.
    \item Otherwise, $P_{\{i,j\}}$ transitions uniformly at random to one of the $15$ non-trivial two-qubit Paulis, independently across the edges.
\end{enumerate}
The first rule follows immediately from unitary invariance of the identity,
$U(\mathbb{I}\otimes\mathbb{I})U^\dagger = \mathbb{I}\otimes\mathbb{I}$.
The second is Schur orthogonality for the conjugation action of
$SU(4)$ on traceless $4\times4$ matrices: the
averaged squared overlap of $U P_{\{i,j\}} U^{\dagger}$ with each of
the fifteen non-trivial labels equals exactly $\tfrac{1}{15}$, and the
factorization over edges is the independence of the gates.

Averaging over both the random matching and the Haar gates, one layer induces a classical Markov kernel $K_n$ on $\mathcal{P}_n^\star$. Its transition probabilities are
\begin{equation}
\label{eq:markov_kernel}
    K_n(P,Q) := p_{\mathcal{E}_{PM}^{(1)}}(Q \,|\, P)= \frac{1}{D^{2}}\, \mathbb{E}_{U \sim \mathcal{E}_{PM}^{(1)}} \left| \operatorname{Tr}\!\left(Q\, U P U^{\dagger}\right) \right|^{2}
    =
    \frac{1}{D^{2}}\, 
    \mathbb{E}_{M \sim \mathcal{M}_n}\;
    \mathbb{E}_{\{U_e\}_{e\in M} \sim \mathrm{Haar}(SU(4))}
    \left| \operatorname{Tr}\!\left(Q\, U_M P U_M^{\dagger}\right) \right|^{2},
\end{equation}
where the gates $U_e$ are sampled independently from the Haar measure on $SU(4)$, and $U_M=\bigotimes_{e\in M} U_e$. Consequently, for a fixed Pauli string $P$, the distribution after one perfect-matching layer is $K_n(P,\cdot)=\{K_n(P,Q)\}_{Q\in\mathcal{P}_n^{\star}}$. Independent layers correspond to successive applications of the same Markov kernel. This yields the following reduction.
\begin{lemma}[Classical Markov chain reduction of the perfect-matching ensemble]\label{lem:pauli-chain-reduction}
Let $\mathcal{E}_{PM}^{(T)}$ be the ensemble of $T$ independent perfect-matching layers. Then for all $P,Q \in \mathcal{P}_n^\star$,
\begin{equation}
    p_{\mathcal{E}_{PM}^{(T)}}(Q \,|\, P) = K_n^{T}(P,Q).
\end{equation}
\end{lemma}
\begin{proof}
Consider the composite circuit $LV$, where $V$ is drawn from an
arbitrary ensemble $\mathcal{E}$, and $L$ is an independent perfect-matching layer drawn from $\mathcal{E}_{PM}^{(1)}$. Conjugation by $V$ maps the input Pauli string $P$ to an operator that can be expanded in the Pauli basis as
\begin{equation}
\label{eq:Vconj}
    V P V^{\dagger} = \sum_{R \in \mathcal{P}_n^\star} a_R(V)\, R,
    \qquad
    a_R(V) = \frac{1}{D}\operatorname{Tr}\!\left(R\, V P V^{\dagger}\right),
\end{equation}
so the non-trivial input string $P$ is sent to the non-trivial intermediate string $R$ with probability $p_{\mathcal{E}}(R \,|\, P) = \mathbb{E}_{V \sim \mathcal{E}} |a_R(V)|^2$. The identity coefficient again vanishes, because $P$ is traceless. Next, the second layer $L$ acts by conjugation on the output of the first layer
\begin{equation}
\label{eq:LVconj}
    L\left(V P V^{\dagger}\right) L^{\dagger}
    = \sum_{R \in \mathcal{P}_n^\star} a_R(V)\, L R L^{\dagger}= \sum_{R \in \mathcal{P}_n^\star} a_R(V)\, \sum_{Q \in \mathcal{P}_n^\star} b_{Q,R}(L)\, Q = \sum_{Q \in \mathcal{P}_n^\star} \bigg ( \sum_{R \in \mathcal{P}_n^\star} a_R(V)\,  b_{Q,R}(L) \bigg )\, Q,
\end{equation}
where we used $L R L^{\dagger} = \sum_{Q \in \mathcal{P}_n^\star} b_{Q,R}(L)\, Q$ with $b_{Q,R}(L) = \frac{1}{D}\operatorname{Tr}\!\left(Q\, L R L^{\dagger}\right)$. Therefore, the action of the entire circuit $LV$ can be read as a single transition from the initial Pauli string $P$ to the final one $Q$ as
\begin{equation}
\label{eq:LV_whole}
    (LV)P(LV)^{\dagger}= \sum_{Q \in \mathcal{P}_n^\star} c_Q(LV)\, Q,
    \qquad
    c_Q(LV) = \frac{1}{D}\operatorname{Tr}\!\left(Q\, (LV) P (LV)^{\dagger}\right).
\end{equation}
Comparing \cref{eq:LVconj,eq:LV_whole}, we identify
\begin{equation}
    c_Q(LV) = \sum_{R \in \mathcal{P}_n^\star} b_{Q,R}(L)\, a_R(V).
\end{equation}
Using this identity and the independence of $L$ and $V$, the transition probability of the full circuit is
\begin{align}
\label{eq:R_terms_decomposition}
    p_{\mathcal{E}_{PM}^{(1)}\circ \;\mathcal{E}}(Q \,|\, P)
    &= \mathbb{E}_{L \sim \mathcal{E}_{PM}^{(1)}}\;\mathbb{E}_{V \sim \mathcal{E}}\, \big| c_Q(LV) \big|^2 \nonumber \\
    &= \sum_{R, R' \in \mathcal{P}_n^\star}
      \mathbb{E}_{V\sim \mathcal{E}}\!\left[ a_R(V)\, \overline{a_{R'}(V)} \right]
      \mathbb{E}_{L\sim \mathcal{E}_{PM}^{(1)}}\!\left[ b_{Q,R}(L)\, \overline{b_{Q,R'}(L)} \right] \nonumber \\
    &= \sum_{R \in \mathcal{P}_n^\star}
      \mathbb{E}_{V\sim \mathcal{E}} \big| a_R(V) \big|^2 \;
      \mathbb{E}_{L\sim \mathcal{E}_{PM}^{(1)}} \big| b_{Q,R}(L) \big|^2
    \;+\;
    \sum_{\substack{R, R' \in \mathcal{P}_n^\star \\ R \neq R'}}
      \mathbb{E}_{V\sim \mathcal{E}}\!\left[ a_R(V)\, \overline{a_{R'}(V)} \right]
      \mathbb{E}_{L\sim \mathcal{E}_{PM}^{(1)}}\!\left[ b_{Q,R}(L)\, \overline{b_{Q,R'}(L)} \right],
\end{align}
where in the last equality we separated the diagonal terms $R = R'$ from the cross terms $R \neq R'$. We now show that the cross terms vanish for every fixed matching, hence also after averaging over the matching. This means that $\mathbb{E}_{L}\!\left[ b_{Q,R}(L)\, \overline{b_{Q,R'}(L)} \right] = 0$ whenever $R \neq R'$, so only the diagonal part survives.

To show this, we fix a perfect-matching $M\in\mathcal{M}_n$ and recall that the associated perfect-matching layer is $L=\bigotimes_{e\in M}U_e$, with each $U_e \sim \mathrm{Haar}(SU(4))$ independent. Since $L$ factorizes, the intermediate and output Pauli strings can also be written accordingly $R=\bigotimes_{e\in M}R_e$, and $Q=\bigotimes_{e\in M}Q_e$, where $R_e$ and $Q_e$ are two qubit Pauli labels. The coefficients $b_{Q,R}(L)$ then factorize as
\begin{equation}
    b_{Q,R}(L)
    = \frac{1}{D}\operatorname{Tr}\!\left(Q\, L R L^{\dagger}\right)
    =\frac{1}{D}\operatorname{Tr}\!\left(\bigotimes_{e\in M}Q_e\, U_e R_e U_e^{\dagger}\right)
    = \prod_{e \in M}
      \underbrace{\frac{1}{4}\operatorname{Tr}\!\left(Q_e\, U_e R_e U_e^{\dagger}\right)}_{=:\ \beta_e(R_e)},
\end{equation}
so $b_{Q,R}(L)$ is a product of the edge numbers $\beta_e(R_e)$. The crossed terms in \cref{eq:R_terms_decomposition} then read
\begin{equation}
    b_{Q,R}(L)\, \overline{b_{Q,R'}(L)}=\prod_{e \in M} \beta_e(R_e) \overline{\beta_e(R'_e)}.
\end{equation}
The factor on edge $e$ involves only the gate $U_e$, and different edges involve different
independent gates, so, conditionally on the matching $M$, the expectation of the product is
the product of expectations; averaging over $M$ then gives the expectation over the layer,
\begin{align}
\mathbb{E}\!\left[\, b_{Q,R}(L)\, b_{Q,R'}(L) \,\middle|\, M \,\right]
&{}= \prod_{e \in M} \mathbb{E}_{U_e \sim \mathrm{Haar}(SU(4))}\!\left[\, \beta_e(R_e)\,
\beta_e(R'_e) \,\right], \nonumber \\[2pt]
\mathbb{E}_{L \sim \mathcal{E}^{(1)}_{\rm PM}}\!\left[\, b_{Q,R}\, b_{Q,R'} \,\right]
&= \mathbb{E}_{M \sim \mathrm{Unif}(\mathcal{M}_n)} \prod_{e \in M}
\mathbb{E}_{U_e \sim \mathrm{Haar}(SU(4))}\!\left[\, \beta_e(R_e)\, \beta_e(R'_e) \,\right].
\label{eq:layer_prod}
\end{align}
It therefore suffices to show that the conditional expectation on the first line vanishes for
every fixed matching $M$. As soon as one factor is zero, the whole product vanishes.

Since $R\neq R'$, they must disagree on at least one qubit. Since every qubit belongs to one edge of the matching, there is at least one edge $e^*$ such that $R_{e^*} \neq R'_{e^*}$. We distinguish two cases.
\begin{enumerate}
    
    \item Exactly one restriction is trivial. Suppose $R_{e^*}=\mathbb{I} \otimes \mathbb{I}$, while $R'_{e^*}$ is non-trivial (the argument is symmetric). For the trivial one $U_{e^*} (\mathbb{I} \otimes \mathbb{I})U_{e^*}^{\dagger}=\mathbb{I} \otimes \mathbb{I}$, so $\beta_{e^*}(R_{e^*})=\beta_{e^*}(\mathbb{I} \otimes \mathbb{I})=\frac{1}{4}\operatorname{Tr}(Q_{e^*})=\delta_{\{Q_{e^*}=\mathbb{I}\otimes \mathbb{I}\}}$. For the non-trivial one, $R_{e^*}'$ is traceless and conjugation keeps it traceless, so $\beta_{e^*}(R'_{e^*})=0$ whenever $Q_{e^*}=\mathbb{I}\otimes \mathbb{I}$. Therefore, the product $\beta_{e^*}(R_{e^*})\overline{\beta_{e^{*}}(R'_{e^*})}$ vanishes when $Q_{e^*}=\mathbb{I}\otimes\mathbb{I}$ and when $Q_{e^*}$ is non-trivial.

    \item Both restrictions are non-trivial and distinct. Choose a two-qubit Pauli $S$ that anticommutes with the non-trivial product $R_{e^*}R'_{e^*}$. Such an $S$ exists because distinct Hermitian representatives never differ by a phase, so $R_{e^\ast}R'_{e^\ast}$ is proportional to a non-trivial two-qubit Pauli, and every non-trivial Pauli anticommutes with half of the two-qubit Pauli group. Equivalently, $S$ commutes with exactly one of $R_{e^\ast}$ or $R'_{e^\ast}$, and anticommutes with the other. Because the Haar measure is invariant under the right substitution $U_{e^\ast} \to U_{e^\ast} S$, we get
    \begin{align}
    \mathbb{E}_{U_{e^\ast} \sim \mathrm{Haar}(SU(4))}\!\left[\beta_{e^\ast}(R_{e^*})\,\overline{\beta_{e^\ast}(R'_{e^*})}\right]
    &= \mathbb{E}_{U_{e^\ast} \sim \mathrm{Haar}(SU(4))}\!\left[
        \tfrac{1}{4}\operatorname{Tr}\!\left(Q_{e^\ast}\, U_{e^\ast} R_{e^*}\, U_{e^\ast}^{\dagger}\right)\,
        \overline{\tfrac{1}{4}\operatorname{Tr}\!\left(Q_{e^\ast}\, U_{e^\ast} R'_{e^*}\, U_{e^\ast}^{\dagger}\right)}
       \right] \nonumber\\
    &= \mathbb{E}_{U_{e^\ast} \sim \mathrm{Haar}(SU(4))}\!\left[
        \tfrac{1}{4}\operatorname{Tr}\!\left(Q_{e^\ast}\, (U_{e^\ast} S)\, R_{e^*}\, (U_{e^\ast}S)^{\dagger}\right)\,
        \overline{\tfrac{1}{4}\operatorname{Tr}\!\left(Q_{e^\ast}\, (U_{e^\ast} S)\, R'_{e^*}\, (U_{e^\ast}S)^{\dagger}\right)}
       \right] \nonumber\\
    &= \mathbb{E}_{U_{e^\ast} \sim \mathrm{Haar}(SU(4))}\!\left[
        \tfrac{1}{4}\operatorname{Tr}\!\left(Q_{e^\ast}\, U_{e^\ast} \left(S R_{e^*} S^{\dagger}\right) U_{e^\ast}^{\dagger}\right)\,
        \overline{\tfrac{1}{4}\operatorname{Tr}\!\left(Q_{e^\ast}\, U_{e^\ast} \left(S R'_{e^*} S^{\dagger}\right) U_{e^\ast}^{\dagger}\right)}
       \right] \nonumber\\
    &= -\mathbb{E}_{U_{e^\ast} \sim \mathrm{Haar}(SU(4))}\!\left[
        \tfrac{1}{4}\operatorname{Tr}\!\left(Q_{e^\ast}\, U_{e^\ast} R_{e^*}\, U_{e^\ast}^{\dagger}\right)\,
        \overline{\tfrac{1}{4}\operatorname{Tr}\!\left(Q_{e^\ast}\, U_{e^\ast} R'_{e^*}\, U_{e^\ast}^{\dagger}\right)}
       \right] \nonumber\\
    &= -\,\mathbb{E}_{U_{e^\ast} \sim \mathrm{Haar}(SU(4))}\!\left[\beta_{e^\ast}(R_{e^*})\,\overline{\beta_{e^\ast}(R'_{e^*})}\right]
    \;=\; 0.
\end{align}
where in the fourth line we used that either $SR_{e^\ast}S^{\dagger} = R_{e^\ast}$, $SR'_{e^\ast}S^{\dagger} = -R'_{e^\ast}$, or $SR_{e^\ast}S^{\dagger} = -R_{e^\ast}$, $SR'_{e^\ast}S^{\dagger} = R'_{e^\ast}$.
\end{enumerate}
In either case one factor of the product in Eq.~\eqref{eq:layer_prod} vanishes for every fixed matching $M$, so the conditional expectation vanishes and hence so does its average over $M$:
\begin{equation}
    \mathbb{E}_{L\sim \mathcal{E}_{PM}^{(1)}}\!\left[ b_{Q,R}(L)\, \overline{b_{Q,R'}(L)} \right]
    = \delta_{R,R'}\; \mathbb{E}_{L\sim \mathcal{E}_{PM}^{(1)}} \big| b_{Q,R}(L) \big|^2
    = \delta_{R,R'}\, K_n(R,Q),
\end{equation}
where the last equality is the definition of the kernel. This kills the cross terms in \cref{eq:R_terms_decomposition}, and only the diagonal terms survive
\begin{equation}\label{eq:onestep}
    p_{\mathcal{E}_{PM}^{(1)}\circ \;\mathcal{E}}(Q \,|\, P)
    = \sum_{R \in \mathcal{P}_n^\star} \mathbb{E}_V |a_R|^2\;\, \mathbb{E}_L |b_{Q,R}|^2
    = \sum_{R \in \mathcal{P}_n^\star} p_{\mathcal{E}}(R \,|\, P)\, K_n(R,Q).
\end{equation}
We conclude the proof by induction on the circuit depth. For $T=1$, the claim follows directly from the definition of the one-layer kernel, $p_{\mathcal{E}_{PM}^{(1)}}(Q\,|\,P) = K_n(P,Q)$. Now assume that $p_{\mathcal{E}_{PM}^{(T)}}(Q\,|\,P) = K_n^{T}(P,Q)$ holds for some $T \geq 1$. A depth-$(T+1)$ circuit can be written as the composition of $LV$, which consists of the first $T$ layers, and $L\sim\mathcal E_{\mathrm{PM}}$, which is an independent final layer. Applying the one-step composition rule in \cref{eq:onestep}, we obtain
\begin{equation}
    p_{\mathcal{E}_{PM}^{(T+1)}}(Q \,|\, P)
    = \sum_{R \in \mathcal{P}_n^\star} K_n^{T}(P,R)\, K_n(R,Q)
    = K_n^{T+1}(P,Q).
\end{equation}
This proves the lemma for every depth.
\end{proof}

\subsubsection{Pauli Haar target and stationarity}
\label{sec:pauli_haar}

We next consider the Haar measure, which sets the target distribution we want to reproduce with the perfect-matching ensemble. Haar-random conjugation spreads the Pauli weight of any non-identity string uniformly over all non-identity Pauli strings. Hence, for all $P$, $Q \in\mathcal{P}_n^{\star}$: 
\begin{equation}
\label{eq:haar_dstr}
p_{\mathrm{Haar}}(Q \,|\, P) = \frac{1}{4^n - 1}=:\pi(Q).
\end{equation}
Thus, Haar random conjugation induces the uniform distribution $\pi$, independently of the input Pauli string $P$. Accordingly, reproducing the Haar second moment amounts to showing that the $T$-step distribution $K_n^T(P, \cdot)$ converges to $\pi$. As a first consistency check, we verify that $\pi$ is an equilibrium of the dynamics.

\begin{lemma}[Stationarity of the uniform distribution]\label{lem:stationarity}
The uniform distribution $\pi$ on $\mathcal{P}_n^\star$ is stationary for the kernel $K_n$:
\begin{equation}
    \sum_{P \in \mathcal{P}_n^\star} \pi(P)\, K_n(P,Q) = \pi(Q)
    \qquad \text{for every } Q \in \mathcal{P}_n^\star.
\end{equation}
In other words, if the input string distribution is uniformly random, it remains uniformly random after any number of perfect-matching layers.
\end{lemma}

\begin{proof}
Since $K_n = \mathbb{E}_{M}\, K_M$ is an average of fixed-matching kernels by \cref{eq:markov_kernel}, it suffices to show that $\pi K_M = \pi$ for every perfect matching $M$.
 
Fix $M$ and decompose Pauli strings into their two-qubit edge labels, $P = \bigotimes_{e \in M} P_e$. We call an edge $e$ active if $P_e \neq \mathbb{I}\otimes\mathbb{I}$, and denote the set of active edges as $\mathrm{act}(P) \subseteq M$. A perfect matching layer preserves this activity pattern exactly: inactive edges remain equal to $\mathbb{I}\otimes\mathbb{I}$, whereas active edges are replaced by non-identity two-qubit labels and therefore remain active.
 
Now sample $P \sim \pi$ and condition on a fixed non-empty activity pattern $\mathrm{act}(P) = S$. There are $15^{|S|}$ Pauli strings with pattern $S$, all equally likely under $\pi$. Conditioned on $S$ the labels on the active edges are therefore independent and uniformly distributed over the $15$ non-trivial two-qubit Paulis. This is exactly the distribution from which the layer resamples them. In other words, one perfect-matching layer is a heat-bath move for $\pi$: it fixes the coarse statistic (the activity pattern) and redraws the remaining detail (the active labels) from its exact conditional distribution under $\pi$. Such a move leaves $\pi$ invariant: the output pattern is distributed as the input pattern, because it is unchanged, and conditionally on the pattern the output labels have the $\pi$-conditional law, because they were just resampled from it. Hence $\pi K_M = \pi$ for every $M$, and averaging over $M$ gives $\pi K_n = \pi$.
\end{proof}

\subsubsection{Reduction to a Markov chain on supports} 
\label{sec:supports}
The update rule of a perfect matching layer does not distinguish among Pauli labels $X$, $Y$, $Z$ of a string. On each edge, it depends only on whether the local Pauli is trivial or nontrivial. This symmetry leads to two structural observations that reduce the Pauli Markov chain to a Markov chain on subsets of $[n]$:
 \begin{enumerate}
     \item For a fixed matching $M$, the distribution of the output support depends only on the input support $\operatorname{supp}(P)$, not on the specific Pauli labels. The support therefore evolves autonomously as a Markov chain on the non-empty subsets $\mathcal{S}_n := \{A \subseteq [n] : A \neq \emptyset\}$. 
     This state space is closed under the dynamics: any non-empty input activates at least one edge, and every active edge remains nontrivial after the update. Hence, the support cannot reach $\emptyset$.
\item Every qubit belongs to some edge of the matching. After a single layer, every letter of an initial string has an associated active edge, and is updated into a uniform output in the labels $\{X,Y,Z\}$ in that active edge support. Hence, output labels carry no information about the past labels. 
 \end{enumerate}
The first observation will enable us to reduce the Pauli Markov chain into a support Markov chain. Later, observation 2 will provide us the dictionary between string distributions and support
distributions.

\subsubsection{The support Markov chain.}
We write $P_t$ for the random string after $t$ layers, so that $\mathbb{P}(P_T = Q) = K_n^{T}(P,Q)$ by \cref{lem:pauli-chain-reduction}, and let $A_t = \operatorname{supp}(P_t)$ be the induced support process, with $A_0 = \operatorname{supp}(P)$. We again call an edge of the sampled matching active if it touches the current support. On an active edge $\{a,b\}$, the fresh label is uniform among the $15$ non-trivial two-qubit Paulis, of which $9$ act non-trivially on both qubits and $3+3$ on exactly one. The active edge therefore branches to the new support
\begin{equation}\label{eq:edge_rule}
    \{a,b\} \text{ with probability } \tfrac{9}{15}=\tfrac{3}{5},
    \qquad
    \{a\} \text{ with probability } \tfrac{3}{15}=\tfrac{1}{5},
    \qquad
    \{b\} \text{ with probability } \tfrac{3}{15}=\tfrac{1}{5},
\end{equation}
independently across all active edges, while inactive edges contribute nothing. Therefore, the new support $A_{t+1}$ is the union of the active edge contributions. Averaging over the matching and the label outcomes, one layer defines a Markov kernel $\kappa_n$ on $\mathcal{S}_n$, in direct analogy with the Pauli kernel $K_n$. Formally, $\kappa_n$ is obtained from $K_n$ by grouping the output strings by their support,
\begin{equation}\label{eq:kappa_def}
    \kappa_n(A, B)
    := \sum_{\substack{Q \in \mathcal{P}_n^\star \\ \operatorname{supp}(Q) = B}} K_n(P, Q),
    \qquad \text{for any } P \in \mathcal{P}_n^\star \text{ with } \operatorname{supp}(P) = A.
\end{equation}
The right-hand side does not depend on the chosen representative $P$, because by observation 1 the whole distribution $K_n(P,\cdot)$ is a function of $\operatorname{supp}(P)$ alone. The same representative-independence holds after any number of layers:
\begin{equation}\label{eq:kappa_T_classsum}
    \kappa_n^{T}(A,B) \;=\; \sum_{\substack{Q \in \mathcal{P}_n^\star\\ \operatorname{supp}(Q)=B}} K_n^{T}(P,Q),
    \qquad \text{for any } P \text{ with } \operatorname{supp}(P)=A.
\end{equation}
Indeed, the case $T = 1$ is \cref{eq:kappa_def}, and for the inductive step we split off the last layer and group the intermediate string $R$ by its support $C$:
\begin{equation}
    \sum_{\operatorname{supp}(Q)=B} K_n^{T}(P,Q)
    = \sum_{C \in \mathcal{S}_n}\; \sum_{\operatorname{supp}(R)=C} K_n^{T-1}(P,R)
      \underbrace{\sum_{\operatorname{supp}(Q)=B} K_n(R,Q)}_{=\;\kappa_n(C,B) \text{ by } \cref{eq:kappa_def}}
    = \sum_{C \in \mathcal{S}_n} \kappa_n^{T-1}(A,C)\,\kappa_n(C,B)
    = \kappa_n^{T}(A,B).
\end{equation}
The underbraced sum can be pulled out of the $R$-sum precisely because it takes the same value for every $R$ with $\operatorname{supp}(R) = C$, and the second equality applies the inductive hypothesis. Consequently, the support after $T$ layers is     distributed as
\begin{equation}\label{eq:support_T_law}
    \mathbb{P}\big(A_T = B\big)
    = \sum_{\substack{Q \in \mathcal{P}_n^\star\\ \operatorname{supp}(Q) = B}} \mathbb{P}(P_T = Q)
    = \sum_{\substack{Q \in \mathcal{P}_n^\star\\ \operatorname{supp}(Q) = B}} K_n^{T}(P,Q)
    = \kappa_n^{T}(A_0, B),
\end{equation}
regardless of which string with support $A_0$ the chain started from.
 
\subsubsection{The stationary support distribution.} The stationary distribution on the supports is
\begin{equation}
 \label{eq:Harr_supp}
    \pi_{\mathrm{supp}}(A) := \sum_{\substack{Q \in \mathcal{P}_n^\star\\ \operatorname{supp}(Q) = A}} \pi(Q)= \frac{3^{|A|}}{4^n-1},
\end{equation}
since there are $3^{|A|}$ strings with support $A$. Then, $\pi_{\mathrm{supp}}$ is stationary for $\kappa_n$:
\begin{equation}
\label{eq:stationarity_pi_supp}
    \sum_{A \in \mathcal{S}_n} \pi_{\mathrm{supp}}(A)\, \kappa_n(A,B)
    = \sum_{P \in \mathcal{P}_n^\star} \pi(P) \!\!\sum_{\substack{Q\,:\, \operatorname{supp}(Q)=B}}\!\! K_n(P,Q)
    = \!\!\sum_{\substack{Q\,:\,\operatorname{supp}(Q)=B}}\!\! \pi(Q)
    = \pi_{\mathrm{supp}}(B),
\end{equation}
where the first equality follows by grouping input Pauli strings according to their supports and using \cref{eq:kappa_def,eq:Harr_supp}. The second equality applies \cref{lem:stationarity}, and the final equality uses \cref{eq:Harr_supp}.
 
\subsubsection{TVD equivalence between Markov chains.}
The two chains introduced above are not merely analogous: for $T\ge1$ the Pauli chain carries no information beyond what the support chain already carries, and the two sit at exactly the same distance from their respective equilibrium:

\begin{lemma}[TVD equivalence between Markov chains]
\label{lem:tvd_equivalence}
Let $T\ge1$, let $P\in\mathcal{P}_n^\star$ with
$\operatorname{supp}(P)=A_0$, and let $Q\in\mathcal{P}_n^\star$. Then
the Pauli chain and its Haar target both factorize into a support
distribution times a uniform letter distribution,
\begin{equation}\label{eq:dictionary}
    K_n^{T}(P,Q)
    = \kappa_n^{T}\big(A_0, \operatorname{supp}(Q)\big)\; 3^{-|\operatorname{supp}(Q)|},
    \qquad
    \pi(Q)
    = \pi_{\mathrm{supp}}\big(\operatorname{supp}(Q)\big)\; 3^{-|\operatorname{supp}(Q)|},
\end{equation}
and consequently the two chains are at the same total variation
distance from equilibrium,
\begin{equation}\label{eq:tvd_supports}
    \mathrm{TVD}\big(K_n^{T}(P,\cdot),\,\pi\big)
    \;=\;
    \mathrm{TVD}\big(\kappa_n^{T}(A_0,\cdot),\,\pi_{\mathrm{supp}}\big).
\end{equation}
\end{lemma}

\begin{proof}
Taking now observation 2, after just one layer, the letters of $P_T$
are, conditionally on the support, independent and uniform. After
$T\geq 1$ layers, multiplying \cref{eq:support_T_law} by this common
conditional distribution (and also using it in \cref{eq:Harr_supp} for
the Haar distribution) gives \cref{eq:dictionary}. Both
distributions are the same uniform letter distribution applied to
different support distributions, so their total variation distance is
carried entirely by the supports:
\begin{align*}
    \,\mathrm{TVD}\big(K_n^{T}(P,\cdot),\,\pi\big)
    &= \frac{1}{2}\sum_{Q\in\mathcal{P}_n^\star} \big| K_n^{T}(P,Q) - \pi(Q)\big|
    = \frac{1}{2}\sum_{B\in\mathcal{S}_n}\; \sum_{\substack{Q\,:\,\operatorname{supp}(Q)=B}} 3^{-|B|}\,
      \big|\kappa_n^{T}(A_0,B) - \pi_{\mathrm{supp}}(B)\big| \\
    &= \frac{1}{2}\sum_{B\in\mathcal{S}_n} \big|\kappa_n^{T}(A_0,B) - \pi_{\mathrm{supp}}(B)\big|
    \;=\; \,\mathrm{TVD}\big(\kappa_n^{T}(A_0,\cdot),\,\pi_{\mathrm{supp}}\big).
\end{align*}
Here we used $\mathrm{TVD}(\mu,\nu):=\frac{1}{2}\sum_x |\mu(x)-\nu(x)|$,
and in the second equality the fact that there are exactly $3^{|B|}$
strings with support $B$, which cancels the factor $3^{-|B|}$
carried by every string of that support.
\end{proof}

\Cref{lem:tvd_equivalence} simplifies the rest of the derivation: the
TVD condition of the Pauli Markov chain collapses into a support Markov
chain for $T\geq 1$, and $P$ enters only through
$A_0=\operatorname{supp}(P)$, so the worst case over the $4^n-1$ Pauli
strings is a worst case over the $2^n-1$ nonempty supports. For a
target accuracy $\eta\in(0,1]$, the mixed-query condition becomes:
\begin{equation}\label{eq:pauli_support_identity}
    \max_{P \in \mathcal{P}_n^\star}\;
    \mathrm{TVD}\Big( p_{\mathcal{E}_{PM}^{(T)}}(\,\cdot \mid P),\; p_{\mathrm{Haar}}(\,\cdot \mid P) \Big)
    \;=\;
    \max_{P \in \mathcal{P}_n^\star}\;
    \mathrm{TVD}\Big( K_n^{T}(P,\cdot),\; \pi \Big)
    \;=\;
    \max_{A \in \mathcal{S}_n}\;
    \mathrm{TVD}\Big( \kappa_n^{T}(A,\cdot),\; \pi_{\mathrm{supp}} \Big)\leq \eta.
\end{equation}
The rest of the appendix forgets Pauli strings, and only requires controlling the supports.

\newpage

\subsection{Grand coupling between support chains} \label{sec:Grand coupling between support chains}
\Cref{eq:pauli_support_identity} is a worst-case condition over the $2^n-1$ non-empty supports. A standard approach in reversible Markov chains is to attack the problem via the spectral gap $\Delta$ of the kernel~\cite{LevinPeresWilmer2017}. In our case, such an approach bounds the TVD as
\begin{equation}
    \max_{A\in\mathcal{S}_n}\; \mathrm{TVD}\big(\kappa_n^{T}(A,\cdot),\, \pi_{\mathrm{supp}}\big)
    \;\le\; e^{-T\,\Delta}\, \sqrt{\tfrac{1}{\pi_{\min}}},
\end{equation}
where $\pi_{\min} = \min_{A}\pi_{\mathrm{supp}}(A) = 3/(4^n-1)$. Even in the best scenario where $\Delta = \Omega(1)$, the resulting bound only becomes small once $T = O(n)$, far larger than the $O(\log n)$ depth we seek. We must therefore take a different route.

\subsubsection{The grand coupling}\label{sec:grandcoupling}
We now approach the Markov chain on supports via the strategy of grand coupling~\cite{LevinPeresWilmer2017}. The idea is the following: instead of studying the output probability distribution of a support $A_0$ through the Markov chain, the grand coupling fixes instead a randomized instance $w$ and studies the deterministic evolution of all supports under $w$ (see \cref{fig:grand-coupling-overview-panel-a}). Afterwards, we will infer information about the probability distributions from this engineered proof device.
\\
\\
In the support Markov chain associated to our ensemble of interest $\mathcal{E}_{PM}^{(T)}$, for each time step $t$ we sample one uniform perfect-matching $M_t \sim \mathcal{M}_n$, and for every edge $e = \{a,b\} \in M_t$ we sample and fix a random update
\begin{equation}\label{eq:shared_variables}
    Y_{t,e} \;=\;
    \{a,b\} \text{ with probability } \tfrac{3}{5},
    \qquad
    \{a\} \text{ with probability } \tfrac{1}{5},
    \qquad
    \{b\} \text{ with probability } \tfrac{1}{5},
\end{equation}
independently across edges and time steps. Every chain, regardless of the initial support, is then updated with these same shared global variables $Y_{t,e}$. Writing $\Psi_t$ for the update function of supports,
\begin{equation}\label{eq:coupled_update}
    A_{t+1} \;=\; \Psi_t(A_t)
    \;:=\; \bigcup_{\substack{e \in M_t \\ e \cap A_t \neq \emptyset}} Y_{t,e}.
\end{equation}
Therefore, each chain collects the outcome $Y_{t,e}$ of every edge that touches its current support, and ignores the rest. The whole deterministic evolution after T perfect matching layers is then fixed by $w=\{M_t,(Y_{t,e})_{e\in M_t}\}_{0\leq t < T}$. Because each layer can implement $|\mathcal{M}_n|=(n-1)!!$ different perfect matchings, and each of the $n/2$ edges of a matching can have the three possible outcomes $Y_{t,e}\in \{\{a\},\{b\},\{a,b\}\}$, after $T$ layers, the number of different realizations for $w$ is $((n-1)!! \;3^{n/2})^T$, which is the number of columns in \cref{fig:grand-coupling-overview-panel-a}.

Sharing the randomness via $w$ buys two structural facts, both immediate from \cref{eq:coupled_update}:

\begin{lemma}[Absorption and monotonicity]\label{lem:coupling_facts}
Under the grand coupling fixed realization $w$, the deterministic path for all times $t$:
\begin{enumerate}
    \item[(a)] if $A_t = B_t$, then $A_s = B_s$ for all $s \ge t$;
    \item[(b)] if $A_t \subseteq B_t$, then $A_s \subseteq B_s$ for all $s \ge t$.
\end{enumerate}
\end{lemma}
\begin{proof}
The map $\Psi_t$ in \cref{eq:coupled_update} depends only on its input set and the shared variables $w$, so equal inputs give equal outputs. Iterating proves (a). For (b), if $A_t \subseteq B_t$, then every edge active for $A_t$ is active for $B_t$ and contributes the same $Y_{t,e}$ to both unions, while $B_t$ may activate additional edges. Hence $\Psi_t(A_t) \subseteq \Psi_t(B_t)$. Iterating proves (b).
\end{proof}
Therefore, fact (a) says that after two supports collide, they permanently receive identical updates forever (see \cref{fig:grand-coupling-overview-panel-b}). Fact (b) says the coupling is monotone: inclusions between supports, once true, are preserved for all time.

The two facts together suggest a powerful strategy. First, fact (a) suggests that we can reduce the mixing problem into a collision problem, by showing that all supports collide in $\log n$ time for most of the realizations $w$: if all columns in \cref{fig:grand-coupling-overview-panel-a} are equal, so will be the rows. Secondly, instead of tracking the whole $2^n-1$ non-trivial supports, fact (b) suggests nesting relations between them, so we only have to track a few of them. We elaborate on these two steps in the following.

\newpage
\subsubsection{From mixing to collisions}

We can now convert information from the columns of \cref{fig:grand-coupling-overview-panel-a} (fixed realizations $w$) back to its rows (probability distributions). To this end, we introduce the equilibrium chain $A^{(\pi)}_t$, whose initial support is drawn from $\pi_{supp}$ (that is, $A^{(\pi)}_0 = B$ with probability $\pi_{\mathrm{supp}}(B) = 3^{|B|}/(4^n-1)$) and which thereafter consumes the shared variables $Y_{t,e}$ like every other chain. This chain is a pure proof device, but it has a convenient property: since $\pi_{\mathrm{supp}}$ is stationary for $\kappa_n$, the chain $A^{(\pi)}_t$ is distributed exactly as $\pi_{\mathrm{supp}}$ at every time. Then, comparing the output distribution of $A$ against the target $\pi_{\mathrm{supp}}$ is equivalent to comparing the two chains for $A$ and $A^{(\pi)}$ under the same randomness, as we capture in the following lemma:
\begin{lemma}[Coupling inequality]\label{lem:coupling_inequality}
For any initial support $A \in \mathcal{S}_n$ and any time $T$,
\begin{equation}\label{eq:coupling_inequality}
    \mathrm{TVD}\big( \kappa_n^{T}(A,\cdot),\; \pi_{\mathrm{supp}} \big)
    \;\le\;
    \mathbb{P}\big( A_T \neq A^{(\pi)}_T \big).
\end{equation}
\end{lemma}
\begin{proof}
Each chain, viewed on its own, follows the correct transition law on every set of supports $\mathcal{B} \subseteq \mathcal{S}_n$,
\begin{equation}\label{eq:marginals}
    \kappa_n^{T}(A,\mathcal{B}) = \mathbb{P}(A_T \in \mathcal{B}),
    \qquad
    \pi_{\mathrm{supp}}(\mathcal{B}) = \mathbb{P}\big(A^{(\pi)}_T \in \mathcal{B}\big).
\end{equation}
Taking now the variational form of the TVD, that is $\mathrm{TVD}(\mu,\nu) = \max_{\mathcal{B} \subseteq \mathcal{S}_n}
\big[\mu(\mathcal{B}) - \nu(\mathcal{B})\big]$, together with
\cref{eq:marginals},
\begin{equation}\label{eq:variational}
    \mathrm{TVD}\big( \kappa_n^{T}(A,\cdot),\; \pi_{\mathrm{supp}} \big)
    = \max_{\mathcal{B} \subseteq \mathcal{S}_n}
      \Big[\, \mathbb{P}(A_T \in \mathcal{B})
            - \mathbb{P}\big(A^{(\pi)}_T \in \mathcal{B}\big) \Big].
\end{equation}
Fixing $\mathcal{B}$, and decomposing each probability in \cref{eq:variational}
over the collision event $\{A_T = A^{(\pi)}_T\}$ and its complement,
\begin{align}
\label{eq:transition_1}
    \mathbb{P}(A_T \in \mathcal{B})
    &= \mathbb{P}\big(A_T \in \mathcal{B},\, A_T = A^{(\pi)}_T\big)
     + \mathbb{P}\big(A_T \in \mathcal{B},\, A_T \neq A^{(\pi)}_T\big), \\
    \label{eq:transition_2}
    \mathbb{P}\big(A^{(\pi)}_T \in \mathcal{B}\big)
    &= \mathbb{P}\big(A^{(\pi)}_T \in \mathcal{B},\, A_T = A^{(\pi)}_T\big)
     + \mathbb{P}\big(A^{(\pi)}_T \in \mathcal{B},\, A_T \neq A^{(\pi)}_T\big). 
\end{align}
On the collision event, the two chains occupy the same support, so $\{A_T \in \mathcal{B}\} \cap \{A_T = A^{(\pi)}_T\}
    \;=\;
    \{A^{(\pi)}_T \in \mathcal{B}\} \cap \{A_T = A^{(\pi)}_T\}$, and therefore $ \mathbb{P}\big(A_T \in \mathcal{B},\, A_T = A^{(\pi)}_T\big) = \mathbb{P}\big(A^{(\pi)}_T \in \mathcal{B},\, A_T = A^{(\pi)}_T\big)$. Subtracting \cref{eq:transition_2} from \cref{eq:transition_1}:
    
\begin{align}\label{eq:after_cancel}
    \mathbb{P}(A_T \in \mathcal{B})
    - \mathbb{P}\big(A^{(\pi)}_T \in \mathcal{B}\big)
    &=
    \mathbb{P}\big(A_T \in \mathcal{B},\, A_T \neq A^{(\pi)}_T\big)
    - \mathbb{P}\big(A^{(\pi)}_T \in \mathcal{B},\, A_T \neq A^{(\pi)}_T\big) \nonumber\\
    &\leq \mathbb{P}\big(A_T \in \mathcal{B},\, A_T \neq A^{(\pi)}_T\big) \nonumber\\ 
    &\leq
    \mathbb{P}\big(A_T \neq A^{(\pi)}_T\big),
\end{align}
where in the second step we dropped the negative term, and in the last one we discarded the constraint $A_T\in\mathcal{B}$ in the remaining term. The final bound no longer depends on $\mathcal{B}$. Inserting it in \cref{eq:variational} already saturates the maximum, which proves
\cref{eq:coupling_inequality}.
\end{proof}
The right-hand side of \cref{eq:coupling_inequality} has a nice interpretation in the table of \cref{fig:grand-coupling-overview-panel-a}. Write $A_T(w\mid A)$ for the entry of the table: the support reached at time $T$ starting from support $A$ under a fixed realization $w = \{M_t, Y_{t,e}\}_{t \le T}$ of the shared variables. Two independent pieces of randomness enter: which column $w$ we are in, and which partner row $B$ the equilibrium chain starts from. Unpacking,
\begin{equation}\label{eq:coupling_unpacked}
    \mathbb{P}\big( A_T \neq A^{(\pi)}_T \big)
    \;=\;
    \sum_{w}\,\sum_{B \in \mathcal{S}_n}\;
    \underbrace{\mathbb{P}(w)}_{\substack{\text{prob.\ of}\\ \text{the column}}}\;
    \underbrace{\pi_{\mathrm{supp}}(B)}_{\substack{\text{weight of}\\ \text{the partner}}}\;
    \underbrace{\mathbf{1}\big[\, A_T(w \mid A) \neq A_T(w \mid B) \,\big]}_{\text{do the two cells differ?}}
\end{equation}
In words: pick a random column $w$ with probability $\mathbb{P}(w)$, pick a $\pi_{\mathrm{supp}}$-weighted partner row, check whether
rows $A$ and $B$ disagree in that column, and compute this probability.

Therefore, \cref{eq:coupling_unpacked} tells us what to do: if we manage to show rapid coalescence in the columns, \Cref{lem:coupling_facts}(a) says that such coalescence, once it occurs, is permanent. The mixing problem is then reduced to a collision task: show that, outside an exponentially small fraction of columns, the entire column collapses to one trajectory within $T = O(\log n)$ layers.

\subsubsection{The sandwich: two chains to rule them all}
At first sight, the worst case over all initial supports still forces us to control exponentially many chains. We now exploit the fact (b) about monotonicity in \cref{lem:coupling_facts}, which allows us to reduce the problem to only controlling $n+1$ chains.

Now run the grand coupling from the full support $[n]$ and from every singleton support $\{i\}$, and write $A^{(\mathrm{full})}_t$ and $A^{(i)}_t$ for the resulting chains. We also take any initial support $S \in \mathcal{S}_n$ and pick any qubit $i \in S$. Since $\{i\} \subseteq S \subseteq [n]$ at time zero, \Cref{lem:coupling_facts}(b) propagates the inclusions for all time:
\begin{equation}\label{eq:sandwich_1}
    A^{(i)}_t \;\subseteq\; A^{(S)}_t \;\subseteq\; A^{(\mathrm{full})}_t
    \qquad \text{for every } t \ge 0.
\end{equation}
Consequently, on the event that every singleton chain has caught the full chain, all chains are equal: if $A^{(i)}_T = A^{(\mathrm{full})}_T$ for all $i \in [n]$, then any chain sandwiched as in \cref{eq:sandwich_1} equals both. This applies in particular to the stationary chain $A^{(\pi)}_t$: its initial support is non-empty, so it contains some singleton, and it is sandwiched like every other chain. Combining with the coupling inequality \cref{eq:coupling_inequality} and a union bound over the singletons:

\begin{lemma}[Reduction to singleton coalescence]\label{lem:singleton_reduction}
Under the grand coupling, for every $T \ge 0$,
\begin{equation}\label{eq:singleton_reduction}
    \max_{A \in \mathcal{S}_n}\; \mathrm{TVD}\big( \kappa_n^{T}(A,\cdot),\; \pi_{\mathrm{supp}} \big)
    \;\le\;
    \sum_{i=1}^{n} \mathbb{P}\big( A^{(i)}_T \neq A^{(\mathrm{full})}_T \big).
\end{equation}
\end{lemma}

\begin{proof}
Fix an initial support $A \in \mathcal{S}_n$ and run, alongside $A_t$ and the
stationary chain $A^{(\pi)}_t$, the chains $A^{(\mathrm{full})}_t$ and
$A^{(i)}_t$ for $i \in [n]$, all under the same shared variables. Define the event where all singleton chains have coalesced with the full-support chain:
\begin{equation}\label{eq:catchup_event}
    \mathcal{C}_T \;:=\; \bigcap_{i=1}^{n}
    \big\{ A^{(i)}_T = A^{(\mathrm{full})}_T \big\}.
\end{equation}
Since $A$ is non-empty, pick any $i \in A$, so $\{i\} \subseteq A \subseteq [n]$ at time zero. By
\Cref{lem:coupling_facts}(b), at any given $T$
\begin{equation}
    A^{(i)}_T \;\subseteq\; A_T \;\subseteq\; A^{(\mathrm{full})}_T .
\end{equation}
On $\mathcal{C}_T$ the two ends of this sandwich are equal, which forces all chains to be equal:
$A_T=A_T^{\pi}=A_T^{(i)} = A^{(\mathrm{full})}_T$. Therefore
$A_T = A^{(\mathrm{full})}_T = A^{(\pi)}_T$ on $\mathcal{C}_T$, i.e.
\begin{equation}\label{eq:disagreement_implies_miss}
    \big\{ A_T \neq A^{(\pi)}_T \big\} \;\subseteq\; \mathcal{C}_T^{\,c}
    \;=\; \bigcup_{i=1}^{n} \big\{ A^{(i)}_T \neq A^{(\mathrm{full})}_T \big\}.
\end{equation}
Combining the coupling inequality from \cref{lem:coupling_inequality} with
\cref{eq:disagreement_implies_miss} and the union bound,
\begin{equation}
    \mathrm{TVD}\big( \kappa_n^{T}(A,\cdot),\; \pi_{\mathrm{supp}} \big)
    \;\le\;
    \mathbb{P}\big( A_T \neq A^{(\pi)}_T \big)
    \;\le\;
    \mathbb{P}\big( \mathcal{C}_T^{\,c} \big)
    \;\le\;
    \sum_{i=1}^{n} \mathbb{P}\big( A^{(i)}_T \neq A^{(\mathrm{full})}_T \big).
\end{equation}
The right-hand side does not depend on $A$, so the bound survives the maximum over $A \in \mathcal{S}_n$, which proves \cref{eq:singleton_reduction}.
\end{proof}

As a conclusion, by only controlling the coalescence between the singleton chains and the full-support chain, \Cref{lem:singleton_reduction} bounds the TVD on any of the exponentially-many input supports $A$.

\subsubsection{The gap set} \label{sec:gapset}
For $i\in[n]$, we now define the discrepancy between the two chains $A^{(\mathrm{full})}_t$ and $A_t^{(i)}$ as the gap set:
\begin{equation}
    G_t^{(i)}= A^{(\mathrm{full})}_t \backslash A_t^{(i)}.
\end{equation}
In particular, the two chains collide ($A_t^{(i)}=A^{(\mathrm{full})}_t$) when $G_t=\emptyset$. We can infer the update of the gap set from \cref{eq:coupled_update}. Because $A_t^{(i)} \subseteq A^{(\mathrm{full})}_t$, every active edge of the small chain is also active on the full chain, and contributes the same $Y_{t,e}$ to both. The only elements in the updated gap set can arise from edges that are for the full chain, but inactive for the small chain:
\begin{equation}
    G_{t+1}^{(i)}
=
\bigcup_{\substack{
e\cap A_t^{(i)}=\emptyset\\
e\cap A_t^{(\mathrm{full})}\neq\emptyset
}}
Y_{t,e}.
\end{equation}
We now want to understand this contraction of the gap set towards $\emptyset$. To this end, we check what happens for a qubit in $G_t^{(i)}$, when getting paired with another: 
\begin{itemize}
    \item If the paired qubit is in $A_t^{(i)}$, that edge becomes active for both chains, as both receive the same $Y_{t,e}$ and that qubit in $G_t^{(i)}$ is not included in $G_{t+1}^{(i)}$. This helps contract the gap set.
    \item If the paired qubit is outside $A_t^{(i)}$, but inside $A_t^{(\mathrm{full})}$ (hence in $G_t^{(i)}$), then those two qubits are still present in $G_{t+1}^{(i)}$ with probability $3/5$ ($Y_{t,e}=\{a,b\}$), helping to keep the gap size equal. Only one of the qubits is present in $G_{t+1}^{(i)}$ with probability $2/5$ ($Y_{t,e}\in \{\{a\},\{b\}\}$), helping the gap set to contract. 
    \item If the paired qubit is outside $A_t^{(\mathrm{full})}$, then the edge is active for the full chain and inactive for the small one, so both qubits are present in $G_{t+1}^{(i)}$ with probability $3/5$ ($Y_{t,e}=\{a,b\}$), helping the gap set to expand. Only one of them is present with probability $2/5$ ($Y_{t,e}\in\{\{a\},\{b\}\}$), leaving the gap set size equal.
\end{itemize}
Therefore, each element in the gap set dies with probability roughly equal to $|A_t^{(i)}|/n$, and survives or expands with probability roughly equal to $1-|A_t^{(i)}|/n$. Therefore, if the small chain covers at least half the qubits ($|A_t^{(i)}|\geq n/2$), this suggests the geometric contraction $|G_t^{(i)}| \lesssim c^t$, with $c<1$, so $G_T^{(i)}=\emptyset$ for $T=O(\log n)$. However, if $|A_t^{(i)}|<n/2$, we have no geometrical contraction. Note the pleasant asymmetry: only the singleton chains matter, because $A_t^{(\mathrm{full})} \supseteq A_t^{(i)}$ is satisfied for free.
\\
\\
This suggests a clean three-step plan. Because the singletons $A_0^{(i)}$ start at the worst part of the chain:
\begin{enumerate}
    \item \textbf{Growth:} Grow the initial singletons $A_0^{(i)}$ to size $|A_T^{(i)}|\geq n/2$ with exponentially high probability for $T=O(\log n)$.
    \item \textbf{Persistence:} Show that supports $\geq n/2$ persist above $n/2$ with exponentially high probability.
    \item \textbf{Contraction:} Show geometrical contraction of the gap set towards $\emptyset$ for supports $\geq n/2$.
\end{enumerate}
The rest of the appendix is devoted to proving rigorously this three-step roadmap for all of the singletons, so that each summand in \cref{eq:singleton_reduction} is at most $\eta/n$ after $T = O(\log \tfrac{n}{\eta})$ steps, and the TVD is bounded by $\eta$.

\newpage

\subsection{Growth of singletons}
\label{sec:growth}

The previous \cref{sec:Grand coupling between support chains} ended with a program to attack \cref{lem:singleton_reduction}: grow singletons to size $\geq n/2$, show large supports of size $\geq n/2$ stay large, and finally contract them geometrically. In this section we carry out the first step, and show that every singleton chain $A^{(i)}_t$ reaches half density, $|A_T^{(i)}|\ge n/2$ after $T=O(\log n)$ perfect matching layers, except with an exponentially small failure probability that we are free to tune.
\\
\\
To this end, we begin by defining the support weight at step $t$ as $H_t:=|A_t|$. For a fixed realization $w$ of the grand coupling, we now define the filtration at time $t$ as the information contained in its first $t$ layers, which formally is captured by the $\sigma$-algebra:
\begin{equation}\label{eq:filtration}
  \mathcal F_t \;:=\; \sigma\Big(\big(M_s,\,(Y_{s,e})_{e\in M_s}\big)\;:\;0\le s<t\Big).
\end{equation}
Conditioning on $\mathcal F_t$ freezes the past of every coupled chain simultaneously. The remaining randomness is in the layers $\{M_s,Y_{s,e}\}$ with $s\geq t$.

\subsubsection{Two sources of randomness and the expected drift}
Conditioned on $\mathcal{F}_t$, a support chain lands on $A_t$, with size $H_t$. We now want to understand what happens to the updated support size $H_{t+1}$ after one more perfect matching layer. There are two sources of stochasticity involved in this whole section:
\begin{enumerate}
   \item \textbf{The active edges.} For the fresh perfect matching layer, a perfect matching $M_t$ is sampled from $\mathcal{M}_n$. This defines the stochastic variable $R_t$: the number of active edges, those that have some endpoint in $A_t$. Out of the $R_t$ active edges, we define $J_t$ as the number of active edges that are internal to $A_t$, that is, with both endpoints in $A_t$. Every internal edge holds two support qubits, and every other active edge exactly one, so counting support qubits,
    \begin{equation}
    \label{eq:active-edges}
        R_t \;=\; H_t - J_t .
    \end{equation}
    The randomness of the active edges can therefore be captured in the internal-edge count $J_t$.
    \item \textbf{Edge update.} Once the perfect matching $M_t$ is fixed, so is the number of active edges $R_t$ and internal edges $J_t$. Each active edge keeps both of its endpoints with probability $\tfrac{3}{5}$ and a single one with probability $\tfrac{2}{5}$, independently across edges. Indeed, by the update rule \cref{eq:coupled_update}, the new support $A_{t+1}$ is the disjoint union of the outcomes $Y_{t,e}$ over the active edges, so each active edge contributes one guaranteed qubit, plus a second one with probability $\tfrac{3}{5}$. The stochastic updated size is therefore a binomial distribution on top of the active-edge count, so conditional on $(\mathcal{F}_t, M_t)$:
    \begin{equation}
    \label{eq:one-step}
        H_{t+1} \;=\; R_t + B_t, \qquad B_t\sim\mathrm{Bin}\!\left(R_t , \tfrac{3}{5}\right),
    \end{equation}
    where $B_t$ counts the active edges that keep both
endpoints.
\end{enumerate}
Two deterministic bounds come for free from this structure: every internal edge uses two support qubits, so $J_t \leq \lfloor H_t/2 \rfloor$ and hence $R_t \geq \lceil H_t/2 \rceil \geq 1$; and every active edge outputs one or two qubits, so $R_t \leq H_{t+1} \leq 2R_t$. Together,
\begin{equation}
\label{eq:sandwich}
    \lceil H_t / 2 \rceil \;\leq\; H_{t+1} \;\leq\; 2 H_t .
\end{equation}
In particular $A_{t+1} \neq \emptyset$: one layer can at most double or halve the support, which can stall but never dies.

We now capture the expected growth of the support size, or the one-step drift, in the following lemma:

\begin{lemma}[Expected one-step drift]
\label{lem:drift}
Condition on $\mathcal{F}_t$, so that the support $A_t$ and its size $H_t \geq 1$ are fixed. The expected value of internal edges is
\begin{equation}
\label{eq:internal-edges}
    \mathbb{E}\!\left[ J_t \,\middle|\, \mathcal{F}_t \right] \;=\; \frac{H_t (H_t - 1)}{2(n-1)} .
\end{equation}
Consequently, if $H_t \leq n/2$,
\begin{equation}
\label{eq:drift-heuristic}
    \mathbb{E}\!\left[ H_{t+1} \,\middle|\, \mathcal{F}_t \right] \;\geq\; \tfrac{6}{5}\, H_t .
\end{equation}
\end{lemma}
 
\begin{proof}
For \cref{eq:internal-edges}, a qubit in the support has equal chance $1/(n-1)$ to be paired with any other qubit. Of those, $H_t-1$ are also in the support, so every qubit gets an internal partner with probability $(H_t-1)/(n-1)$. Summing over all support qubits gives $H_t(H_t-1)/(n-1)$, but an internal edge contains two support qubits, so this count sees every internal edge twice. Dividing by 2 yields \cref{eq:internal-edges}.
\\
\\
For the drift, we first take the expected value for a fixed matching in \cref{eq:one-step}:
\begin{equation} \label{eq:match_update}
    \mathbb{E}\!\left[ H_{t+1} \,\middle|\, \mathcal{F}_t, M_t \right]
    \;=\; R_t + \mathbb{E}\!\left[ B_t \,\middle|\, \mathcal{F}_t, M_t \right]
    \;=\; R_t + \tfrac{3}{5}\, R_t
    \;=\; \tfrac{8}{5}\, R_t
    \;=\; \tfrac{8}{5} \left( H_t - J_t \right),
\end{equation}
where we used $\mathbb{E}[B_t]=\frac{3}{5}R_t$ for $B_t \sim \mathrm{Bin}\!\left(R_t , \tfrac{3}{5}\right)$ in the second step, and \cref{eq:active-edges} in the last one. Then, averaging \cref{eq:match_update} over the matchings yields
\begin{align}
    \mathbb{E}\!\left[ H_{t+1} \,\middle|\, \mathcal{F}_t \right]
    \;=\; \mathbb{E}\Big[\, \mathbb{E}\!\left[ H_{t+1} \,\middle|\, \mathcal{F}_t, M_t \right] \,\Big|\, \mathcal{F}_t \Big]
    \;=\; \tfrac{8}{5} \left( H_t - \mathbb{E}\!\left[ J_t \,\middle|\, \mathcal{F}_t \right] \right)
    \;=\; \tfrac{8}{5}\, H_t \left( 1 - \frac{H_t - 1}{2(n-1)} \right), \label{eq:drift}
\end{align}
where we inserted the already-proven \cref{eq:internal-edges} in the last step. Finally, if $H_t \leq n/2$, the drag of the internal edges is small, $\frac{H_t - 1}{2(n-1)} \leq \frac{n/2 - 1}{2(n-1)} < \frac{1}{4}$, and \cref{eq:drift-heuristic} follows.
\end{proof}
The previous \cref{lem:drift} has a clean reading: the support size has a positive multiplicative drift when $H_t\leq n/2$. This suggests that a logarithmic number of layers should carry even a singleton above $n/2$. To turn this into a guarantee, we must control additionally the two sources of stochasticity around the drift: the number active edges may be unlucky (too few, with many internal edges), and the edge updates may be unlucky too (too many active edges update into a single endpoint).

There is, however, one more danger to account for: at small sizes the support can get stuck. A singleton, for instance, simply stays at size one with probability $\tfrac{2}{5}$. The picture to attack is therefore: once the support is large, it grows steadily after controlling the stochasticity sources. While it is small, it may stall for a few layers. In the rest of the section we show that the stalls are short and the growth is reliable with controlled stochasticity sources, so that every singleton reaches half density after logarithmically many layers.

\begin{figure}[t]
    \centering
    \includegraphics[width=0.85\linewidth]{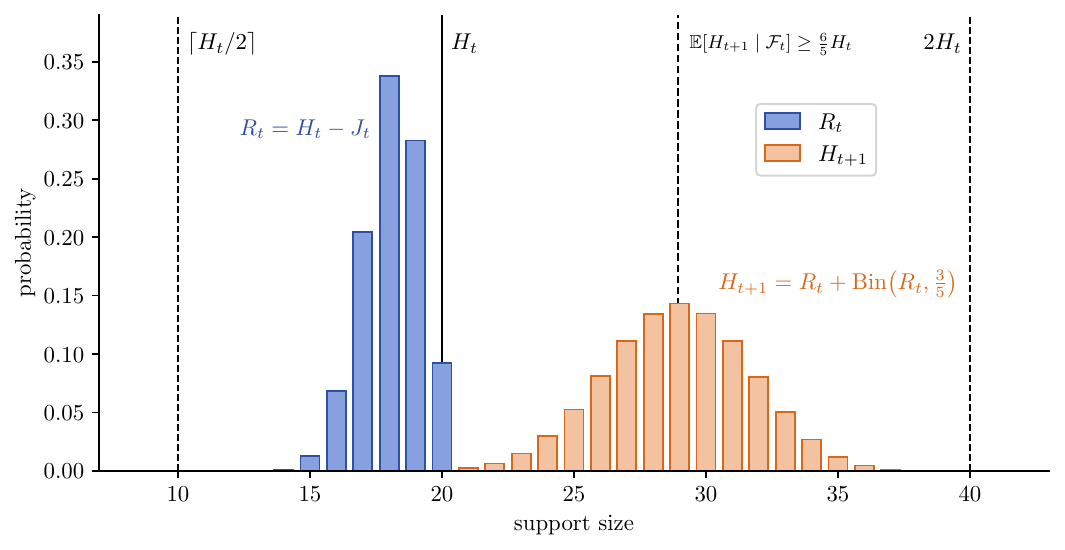}
    \caption{Exemplified one step growth of a support of size $H_t=20$ for $n=100$, conditioned on $\mathcal{F}_t$. Sampling a matching induces a random variable of active edges $R_t$ (blue). Each of the resulting active edges outputs a support of size $1$ with probability $2/5$, or output support of size $2$ with probability $3/5$, resulting in an updated $H_{t+1}$ distribution from \cref{eq:one-step} (orange). The resulting distribution is contained in \cref{eq:sandwich}.}
    \label{fig:placeholder}
\end{figure}

\newpage

\subsubsection{Controlling stochasticity: One step of growth with high probability}\label{sec:growth-step}
We now control the fluctuations of both sources of stochasticity around the drift: the active edges and the edge updates.
\\
\\
We begin with the stochasticity of active edges $R_t$, which we control with the variable $J_t=H_t-R_t$ via \cref{eq:active-edges}. We want to show that $J_t$ concentrates around its mean. We cannot use a standard Chernoff bound for $J_t$, as internal edges are not formed independently and are correlated (if $a$ is matched to $b$, then $a$ cannot be matched to $c$). Instead, we reveal the matching one edge at a time and apply the bounded-difference martingale estimate of Azuma~\cite{Azuma1967} to the Doob martingale of $J_t$ along the revelation.
\begin{lemma}[Matching exposure inequality]\label{lem:exposure}
Let $A_t\subseteq[n]$ with $|A_t|=H_t\ge 1$, let $M_t$ be a uniform perfect-matching of $[n]$, and let $J_t$ be the number of internal edges of $M_t$ with both endpoints in $A_t$. Then, for every $u\ge0$,
\begin{equation}
  \mathbb P\big(J_t-\mathbb EJ_t\ge u\big)\;\le\;\exp\!\Big(-\frac{u^2}{2H_t}\Big).
\end{equation}
\end{lemma}

\begin{proof}
We  reveal the matching in stages: while some qubit of $A_t$ is unmatched, pick the lowest labelled qubit and reveal its partner. Each stage matches at least one qubit of $A_t$, so after at most $H_t$ stages all of $A_t$ is matched and $J_t$ is determined. We continue with zero-contribution stages, if needed, until exactly $H_t$ stages have been counted. Let $Z_j:=\mathbb E[J_t\mid\text{first $j$ revelations}]$ be the associated Doob martingale, so that $Z_0 = \mathbb E[J_t]$ (nothing revealed yet) and $Z_{H_t} = J_t$ (all of $A_t$ matched, so $J_t$ is determined).

We claim $|Z_j-Z_{j-1}|\le 1$. To prove it, we fix a stage $j$, and let $v\in A_t$ be the qubit whose partner is being revealed. We compare two possible partners $w\ne w'$, both currently unmatched. Conditional on $\{v,w\}\in M$ (and on the past revelations), the qubit $w'$ is matched to some unmatched qubit $z$. The switch
\[
  \{v,w\},\{w',z\}\;\longleftrightarrow\;\{v,w'\},\{w,z\}
\]
is a bijection between the conditional matching ensembles given the two outcomes. The bijection alters only these two edges, so it can only change $J_t$ by at most one. To see this we take $a=[v \in A_t]$, $b=[w \in A_t]$, $c=[w'\in A_t]$, $d=[z \in A_t]$, and consider that an edge is internal iff both its bits are $1$:
\begin{equation}
    \Delta J_t=(ac+bd)-(ab+cd)=(a-d)(c-b)
\end{equation}
Therefore $|\Delta J_t|\leq 1$. Any zero-contribution stages have increment zero. Hence the
conditional expectations of $J_t$ given any two possible $j$-th revelations differ by at
most one, so the Doob increment satisfies $|Z_j-Z_{j-1}| \leq 1$.

Finally, we apply the Azuma-Hoeffding inequality to martingales with increments bounded by $c_j=1$, and using the fact that there are at most $H_t$ stages, gives
\begin{equation}
  \mathbb P\big(Z_{H_t} - Z_0 \ge u\big)
  \;\le\; \exp\!\Big(-\frac{u^2}{2\sum_j c_j^2}\Big)
  \;=\; \exp\!\Big(-\frac{u^2}{2H_t}\Big).
\end{equation}
\end{proof}
\Cref{lem:exposure} controls the first source of stochasticity, claiming it is exponentially concentrated around the drift. We now incorporate the second source of stochasticity, which is the edge update in \cref{eq:one-step}. Conditional on $(\mathcal F_t, M_t)$, the number $B_t$ of
active edges that keep both endpoints is a sum of $R_t$ independent Bernoulli
variables, so we can apply a Chernoff bound:

\begin{lemma}[One growth step]\label{lem:growthstep}
Condition on $\mathcal F_t$, so that the support $A_t$ and its size $H_t$ are fixed.
If $H_t \le n/2$, then
\[
  \mathbb P\Big(H_{t+1}<\tfrac{21}{20}H_t\;\Big|\;\mathcal F_t\Big)\;\le\;2\,e^{-H_t/800}.
\]
\end{lemma}
\begin{proof}
We condition on $\mathcal F_t$, and first bound an error on the stochasticity from the active edges. By \cref{eq:internal-edges} in \cref{lem:drift}, when $H_t\le n/2$, we have that $\frac{H_t - 1}{2(n-1)} \leq \frac{n/2 - 1}{2(n-1)} = \frac{n-2}{4(n-1)} < \frac{1}{4}$, and therefore $\mathbb E[J_t\mid\mathcal F_t]\le H_t/4$. \Cref{lem:exposure} with conservative margin $u=H_t/20$ gives
\begin{equation}
\label{eq:error_1}
  \mathbb P\Big(J_t>\tfrac{3H_t}{10}\;\Big|\;\mathcal F_t\Big)
  \;\le\;\mathbb P\Big(J_t-\mathbb E[J_t\mid\mathcal F_t]>\tfrac {H_t}{20}\;\Big|\;\mathcal F_t\Big)
  \;\le\;e^{-H_t/800},
\end{equation}
Hence, except with probability $e^{-H_t/800}$, we have $J_t \le \tfrac{3H_t}{10}$,
and therefore $R_t = H_t - J_t \ge \tfrac{7H_t}{10}$.
\\
\\
We can now control the stochasticity arising from the edges update. For that, we also condition on $M_t$ and suppose $R_t\ge\tfrac{7H_t}{10}$. If $H_{t+1}<\tfrac{21}{20}H_t$, then using \cref{eq:one-step}
\begin{equation}
  B_t\;=\;H_{t+1} - R_t\;<\;\tfrac{21}{20}H_t-\tfrac{7}{10}H_t\;=\;\tfrac{7}{20}H_t\;\le\;\tfrac {R_t}{2}\;=\;\tfrac56\,\mathbb E\big[B_t\,\big|\,\mathcal F_t,M_t\big],
\end{equation}
where we used $\mathbb{E}[B_t]=\frac{3}{5}R_t$ for $B_t \sim \mathrm{Bin}\!\left(R_t , \tfrac{3}{5}\right)$ in the last step. Therefore, failing to grow forces $B_t$ to take a value $1/6$ below its mean, so we now bound the probability of such an event. For $B \sim \mathrm{Bin}(m,p)$ and any $0<\varepsilon<1$, the multiplicative Chernoff
bound for the lower tail is
\begin{equation}\label{eq:chernoff}
  \mathbb P\big(B \le (1-\varepsilon)\, mp\big) \;\le\; e^{-\varepsilon^2 mp/2}.
\end{equation}
Applying \cref{eq:chernoff} to $B_t$, with $m = R_t$, $p = \tfrac35$ and the margin
$\varepsilon=\tfrac16$ identified above, gives
\begin{equation}\label{eq:error_2}
  \mathbb P\Big(B_t<\tfrac{R_t}{2}\;\Big|\;\mathcal F_t,M_t\Big)
  \;\le\;e^{-R_t/120}\;\le\;e^{-7H_t/1200}\;\le\;e^{-H_t/800}.
\end{equation}
Therefore, on the good event of active edges $R_t \ge \tfrac{7H_t}{10}$, the
failure $\{H_{t+1}<\tfrac{21}{20}H_t\}$ implies $\{B_t<\tfrac{R_t}{2}\}$, and the probability of failing to grow is bounded by \cref{eq:error_2}. Finally, growth failure requires one of the two bad events bounded in \cref{eq:error_1} and \cref{eq:error_2}. Adding their probabilities proves the lemma.
\end{proof}
\Cref{lem:growthstep} ensures multiplication by $\tfrac{21}{20}$ per layer, with a
failure probability exponentially small in the current size $H_t$. To convert
multiplication into additive drift we pass to logarithms, yielding the following corollary:
\begin{corollary}[Logarithmic drift above the threshold]\label{cor:drift}
Define the universal constants
\begin{equation}\label{eq:constants}
  \mu \;:=\; \tfrac12\log\tfrac{21}{20},
  \qquad
  h_\ast \;:=\; \Bigg\lceil\,800\,\ln\frac{4\big(1+\log\frac{21}{20}\big)}
  {\log\frac{21}{20}}\,\Bigg\rceil .
\end{equation}
Then, whenever $h_\ast\le H_t<n/2$,
\[
  \mathbb E\big[\log H_{t+1}-\log H_t\,\big|\,\mathcal F_t\big]\;\ge\;\mu .
\]
\end{corollary}

\begin{proof}
We separate the growth event $\{H_{t+1}\ge\tfrac{21}{20}H_t\}$ from its complement,
and write $q$ for the failure probability of the growth event in
\Cref{lem:growthstep}. On the growth event, the log increment is at least
$\log\tfrac{21}{20}$, while on the complementary event, \cref{eq:sandwich}
still guarantees $\log H_{t+1}-\log H_t\ge\log\tfrac12=-1$. Hence
\begin{equation}\label{eq:drift-split}
  \mathbb E\big[\log H_{t+1}-\log H_t\,\big|\,\mathcal F_t\big]
  \;\ge\;(1-q)\log\tfrac{21}{20}\;-\;q .
\end{equation}
The failure probability $q$ reduces the drift. We require that this loss is at most
half of the ideal increment $\log\tfrac{21}{20}$, so that we always keep a positive
constant drift of at least $\mu$. Setting $\mu$ equal to the right-hand side of
\cref{eq:drift-split} and rearranging, this holds precisely when
\begin{equation}\label{eq:qstar}
  q\;\le\;\frac{\log\frac{21}{20}}{2\big(1+\log\frac{21}{20}\big)}.
\end{equation}
By \cref{lem:growthstep}, the failure probability satisfies $q\leq 2e^{-H_t/800}$, which decreases with the size. Solving for
$2e^{-H_t/800}$ to be at most the right-hand side of \cref{eq:qstar} shows that the
requirement is met exactly when $H_t\ge 800\ln\frac{4(1+\log\frac{21}{20})}{\log\frac{21}{20}}$,
that is, whenever $H_t\ge h_\ast$. Inserting \cref{eq:qstar} into
\cref{eq:drift-split} yields
\begin{equation}
  \mathbb E\big[\log H_{t+1}-\log H_t\,\big|\,\mathcal F_t\big]
  \;\ge\;(1-q)\log\tfrac{21}{20}-q\;\ge\;\tfrac12\log\tfrac{21}{20}\;=\;\mu. \qedhere
\end{equation}
\end{proof}

\subsubsection{Escaping from small supports}\label{sec:growth-escape}
After controlling the drift and both sources of stochasticity around it, 
the resulting \Cref{cor:drift} does not apply below $h_\ast$. This includes the singleton chains of size one. We then require a method to escape the small support, and grow in very short time above the size $h_*$ with overwhelming probability. We prove this in the following lemma, which states that any stalls below the threshold end in a time with a geometric tail:
\begin{lemma}[Small supports escape quickly]\label{lem:escape}
Define the universal constants
\begin{equation}\label{eq:escape-constants}
  L_0:=\lceil\log h_\ast\rceil,
  \qquad
  b_\ast:=\big(\tfrac{3}{10}\big)^{h_\ast},
  \qquad
  p_\ast:=b_\ast^{\,L_0},
  \qquad
  n_g:=4h_\ast+2 .
\end{equation}
Let $n\ge n_g$ be even, and condition on $\mathcal F_t$ with $1\le H_t<h_\ast$. Then
the escape time $\tau_\ast:=\inf\{s\ge0:H_{t+s}\ge h_\ast\}$ satisfies
\begin{equation}\label{eq:geom}
  \mathbb P\big(\tau_\ast>mL_0\,\big|\,\mathcal F_t\big)\;\le\;(1-p_\ast)^m
  \qquad\text{for every }m\ge0.
\end{equation}
\end{lemma}

\begin{proof}
For $H_t<h_\ast$, we consider the best-case event with no internal edges, and that all active edges have output on both endpoints:
\begin{equation}
  \{J_t=0\}\;\cap\;\{Y_{t,e}=e\ \text{for every active edge }e\}.
\end{equation}
The $H_t$ qubits on the support have in total $H_t$ active edges, and each of these edges outputs both endpoints. Then $H_{t+1}=2H_t$. Although this event is unlikely, we ask for the probability of such an event:
\begin{itemize}
  \item For the event $\{J_t=0\}$, we expose the partners of the support qubits one at a time, and ask every time that the newly revealed partner lands outside the support. Suppose the first $k-1$ exposed qubits all got outside partners. These pairs occupy $2(k-1)\le 2(h_\ast-1)$
  qubits, so for the next support qubit $k$, there remain at least
  \begin{align}
      n-2(h_\ast-1)-1 \;\ge\; 2h_\ast+3
  \end{align}
  candidate partners (the $-1$ excludes the qubit itself, and the inequality uses
  $n\ge n_g=4h_\ast+2$). Among these candidates, at most $h_\ast-1$ are support qubits.
  Hence each exposure lands outside with conditional probability at least
  \begin{align}
      1-\frac{h_\ast-1}{2h_\ast+3}\;\ge\;\frac12 ,
  \end{align}
  and chaining the at most $H_t\le h_\ast$ exposures with the chain rule,
  \begin{align}
      \mathbb P\big(J_t=0\,\big|\,\mathcal F_t\big)
      \;=\;\prod_{k=1}^{H_t}
      \mathbb P\big(k\text{-th partner outside}\,\big|\,\text{first }k-1\text{ outside},\,\mathcal F_t\big)
      \;\ge\;\Big(\frac12\Big)^{H_t}
      \;\ge\;2^{-h_\ast}.
  \end{align}
  \item For the event $\{Y_{t,e}=e\ \text{for every active edge }e\}$, conditional on $\{J_t=0\}$, there are $H_t$ active edges (one for each qubit in $A_t$). Their updates are independent, and each keeps both of its endpoints
  with probability $\tfrac35$, so
  \begin{align}
      \mathbb P\big(Y_{t,e}=e\ \text{for every active edge}\,\big|\,J_t=0,\,\mathcal F_t,M_t\big)
      \;=\;\Big(\frac35\Big)^{H_t}
      \;\ge\;\Big(\frac35\Big)^{h_\ast}.
  \end{align}
\end{itemize}
Multiplying the two factors, the support doubles uniformly for all $H_t<h_\ast$, with probability at least
\begin{equation}\label{eq:doubling}
  \mathbb P\big(H_{t+1}=2H_t\,\big|\,\mathcal F_t\big)
  \;\ge\;2^{-h_\ast}\Big(\frac35\Big)^{h_\ast}\;=
  \Big(\frac{3}{10}\Big)^{h_\ast}=
  \;b_\ast.
\end{equation}
Now we start from $H_t\ge1$, and ask the probability of consecutive best-case scenarios for growing above $h_{\ast}$. We can apply \cref{eq:doubling} iteratively before hitting $h_{\ast}$, so at most $\lceil\log h_\ast\rceil=L_0$ doublings are needed, as $2^{L_0} \geq h_*$.
\\
\\
The escape attempt above is a block of constant depth $L_0$ with a constant success probability of at least $p_\ast$, valid from any size below $h_\ast$ and any past. If the chain is still below $h_\ast$ after $m$ blocks, the next block fails to escape with probability at most
$1-p_\ast$, so for an escape time $\tau_\ast$:
\begin{equation}
  \mathbb P\big(\tau_\ast>(m+1)L_0\,\big|\,\mathcal F_t\big)
  \;\le\;(1-p_\ast)\,\mathbb P\big(\tau_\ast>mL_0\,\big|\,\mathcal F_t\big).
\end{equation}
Hence, iterating $m$ times contracts the survival probability geometrically to
$(1-p_\ast)^m$, which is \cref{eq:geom}.
\end{proof}

Since $(M_t, (Y_{t,e}))$ are i.i.d. across $t$ and independent of $\mathcal{F}_t$, every estimate conditional on $\mathcal{F}_t$ (\cref{lem:growthstep,cor:drift,lem:escape}) holds for the stopping times. We will use this fact in \cref{sec:growth-hitting}.

\subsubsection{From drift to a hitting time}\label{sec:growth-hitting}

We now control two regimes. On the one hand, above $h_\ast$ the size grows steadily and each layer pushes $\log H_t$ up by at least $\mu$ on average, and never moves it by more than one (\cref{cor:drift}). On the other hand, below $h_\ast$, stalls end in a time that contracts geometrically (\cref{lem:escape}). Hence, it remains to combine the two regimes into a bound on the total time to reach $n/2$. The subtlety is that the chain may cross $h_{*}$ more than once, so that the total time decomposes into the steps spent climbing above $h_\ast$ plus the sum of the durations of all the stalls.

However, the problem is that the durations of the stalls are not independent (earlier stalls influence the starting conditions of later ones), so stall times cannot be directly summed. We attack this problem and provide the final characterization of the total time to hit $n/2$ support in the following lemma:
\begin{lemma}[Hitting half support]\label{lem:hitting-linear}
There exists a universal constant $C_{\mathrm{hit}}>0$ such that, for every even $n\ge n_{\mathrm g}$ and every chain with $H_0=1$, the hitting time
\begin{equation}
  \tau_{1/2}=\inf\{t\ge0:\,H_t\ge n/2\}
\end{equation}
satisfies for every $0<\delta\le1/2$,
\begin{equation}\label{eq:hit-delta}
  \mathbb P\Big(\tau_{1/2}>C_{\mathrm{hit}}\Big(\log n+\log\tfrac1\delta+1\Big)\Big)
  \;\le\;\delta.
\end{equation}
\end{lemma}

\begin{proof}
\emph{Step 1: the embedded ladder.}
We start from $H_0<h_*$, and consider the possibility that the chain moves across $h_*$ multiple times. For this case, we want to study the chain only when it's in the region $\geq h_*$, so keep track of these steps. We define the first time the chain is ever above $h*$ as
\begin{equation}
  \sigma_0=\inf\{t\ge0:\,H_t\ge h_\ast\},
\end{equation}
which is almost surely finite by \cref{lem:escape}. Taking another time step $\sigma_j+1$, makes the support to either remain above $h_*$ (in that case we set $\sigma_{j+1} := \sigma_j + 1$) or drop below and stall again (in that case we do not record it, only when it comes again above $h_*$). Once the process concludes at $H_{\sigma_j}\ge n/2$, we just set $\sigma_{j+1}=\sigma_j$.
Hence, we have the complete list $\{\sigma_0,\sigma_1,\sigma_2,\ldots\}$ of recorded times at which $H_t \geq h_*$, concluding when $H_t\geq n/2$.
\\
\\
The recording times $\sigma_j$ are random. This matters because we want to use \cref{cor:drift,lem:escape}, which are statements conditional on the past. Applying them at a random time is only legitimate if that time is determined by the past alone. Random times with this property are called stopping times. Formally, $T$ is a stopping time if, for every deterministic $m$, the event $\{T\le m\}$ is decidable from the first $m$ layers, i.e.\ $\{T\le m\}\in\mathcal F_m$. The requirement simply forbids looking into the future.
\\
\\
The recording times pass this test. First, $\sigma_0$ is a first hitting time, and
\begin{equation}
  \{\sigma_0\le m\}=\big\{\max_{0\le t\le m}H_t\ge h_\ast\big\}\in\mathcal F_m .
\end{equation}
Next, let $\tau$ be a finite stopping time. Then $\tau+1$ clearly is one as well. Moreover, the first return
\begin{equation}
  \rho=\inf\{t\ge\tau+1:\,H_t\ge h_\ast\}
\end{equation}
is a stopping time, because for every deterministic $m$,
\begin{equation}
  \{\rho\le m\}=\bigcup_{s=0}^{m-1}\Big(\{\tau=s\}\cap
  \big\{\max_{s+1\le t\le m}H_t\ge h_\ast\big\}\Big)\in\mathcal F_m ,
\end{equation}
and it is finite because stalls end in finite time (\cref{lem:escape}). Every
$\sigma_{j+1}$ is built from $\sigma_j$ as either $\sigma_j+1$, such a first return,
or $\sigma_j$ itself, with the choice among the three cases decided by
the information available at those moments. Applying this recursively shows that every
$\sigma_j$ is a finite stopping time.
\\
\\
\medskip\noindent\emph{Step 2: the ladder climbs (drift and bounded increments).}
Define $Z_j:=\log H_{\sigma_j}$. We now show that $(Z_j)_j$ is well behaved by showing the two properties that Azuma's inequality needs in the upcoming Step 3: the recorded walk climbs on average, and it never jumps far.

\smallskip\noindent\emph{(a) Drift.}
On $\{H_{\sigma_j}<n/2\}$, \Cref{cor:drift} applies at the stopping time $\sigma_j$ and controls the very next layer: $\mathbb E[\log H_{\sigma_j+1}-\log H_{\sigma_j}\mid\mathcal F_{\sigma_j}]\ge\mu$. If that layer stalls, the next recorded value is the first return instead of the stall. But the stall lies below $h_\ast$ while the return lies at or above it, so recording the return only replaces the value by a larger one. Hence in both cases $Z_{j+1}\ge\log H_{\sigma_j+1}$ in every realization, and taking conditional expectations preserves the drift: \begin{equation}\label{eq:embedded-drift}
  \mathbb E\big[Z_{j+1}-Z_j\,\big|\,\mathcal F_{\sigma_j}\big]\;\ge\;\mu
  \quad\text{on }\{H_{\sigma_j}<n/2\}.
\end{equation}

\smallskip\noindent\emph{(b) Bounded increments:}
If the ladder is frozen, nothing moves. If no stall happens, the two recorded values
are just one layer apart, and one layer can at most double or halve the size because of
\cref{eq:sandwich}. The only worry is a stall: could the chain come back far from
where it left? It cannot, because a stall begins and ends around the threshold. To
fall below $h_\ast$, the chain must have been standing below $2h_\ast$ (one layer at
most halves), and to climb back above $h_\ast$ from below it, one layer at most
doubles, so it returns below $2h_\ast$. Hence both $Z_j$ and $Z_{j+1}$ lie in
$[\log h_*,\log h_*+1),$
an interval of length one, so
\begin{equation}
|Z_{j+1}-Z_j|\le 1.
\end{equation}
The recorded walk is therefore well-behaved, and gains at least $\mu$ per step on average and moves at most
one per step: stalls neither hide progress nor create jumps.
\\
\\
\medskip\noindent\emph{Step 3: few recorded times suffice.}
In step one we arranged the series of recorded times $\{\sigma_0,\sigma_1,\sigma_2,\ldots\}$. In this step we ask: how many recorded times do we need to reach $\geq n/2$? This is captured by
\begin{equation}
  \widehat\tau:=\inf\{j\ge0:\,H_{\sigma_j}\ge n/2\},
  \qquad
  \mathcal G_j=\mathcal F_{\sigma_{j\wedge\widehat\tau}}.
\end{equation}
where $\mathcal G_j$ is the information available at the $j$-th recorded moment. Azuma's inequality requires a martingale with a given number of increments,
whereas our walk stops at the random index $\widehat\tau$. We therefore pad the walk
past the hit with harmless deterministic increments:
\begin{equation}
  X_j=\begin{cases} Z_{j+1}-Z_j, & j<\widehat\tau,\\ \mu, & j\ge\widehat\tau.\end{cases}
\end{equation}
The padding will never affect the conclusion on the failure event
$\{\widehat\tau>N\}$, which is the only one we analyze; on this event, all increments are real.
By \cref{eq:embedded-drift}, each increment carries a drift,
\begin{equation}
\label{eq:Drift_mu}
  m_j:=\mathbb E[X_j\mid\mathcal G_j]\;\ge\;\mu,
\end{equation}
before the hit and trivial after it. Moreover $X_j$ is $\mathcal G_{j+1}$-measurable, $|X_j|\le1$ by Step 2(b), and hence $|X_j-m_j|\le2$.
Subtracting the promise from the performance,
\begin{equation}
  M_N=\sum_{j=0}^{N-1}(X_j-m_j),
\end{equation}
measures the accumulated deviation of the actual climb relative to its promised
average. By construction it is a martingale with increments bounded by two, so we can control it via Azuma's inequality.
\\
\\
On $\{\widehat\tau>N\}$ no artificial increment has been used, and
\begin{equation}
  \sum_{j=0}^{N-1}X_j=Z_N-Z_0<\log(n/2)-\log h_\ast\le\log n,
\end{equation}
where we telescoped the first sum, and used $Z_0 = \log\; H_{\sigma_0} \geq \log\; h_*$, 
and $Z_N < \log(n/2)$ 
on $\{\widehat\tau>N\}$.
Additionally, on $\{\widehat\tau>N\}$, by \cref{eq:Drift_mu},
\begin{equation}
  M_N \;=\; \sum_{j=0}^{N-1}X_j \;-\; \sum_{j=0}^{N-1}m_j
  \;\le\; \sum_{j=0}^{N-1}X_j \;-\; \mu N
  \;<\; \log n \;-\; \mu N,
\end{equation}
For $u\ge0$, we can choose
\begin{equation}\label{eq:N-choice}
  N=\big\lceil C_1(\log n+u+1)\big\rceil,\qquad C_1\ge2/\mu.
\end{equation}
We define $S=\log n+u+1$. Then $\mu N \ge \mu C_1 S \ge 2S \ge 2\log n$, so $\log n-\mu N\le-\mu N/2$. Therefore
\begin{equation}
  \{\widehat\tau>N\}\subseteq\{M_N\le-\mu N/2\}.
\end{equation}
Applying the Azuma-Hoeffding inequality on the martingale $(M_j)$ with increments bounded by two gives
\begin{equation}\label{eq:tauhat-tail}
  \mathbb P\big(\widehat\tau>N\big)\le\exp\!\Big(-\frac{\mu^2N}{32}\Big)\le e^{-\kappa_1S},
  \qquad \kappa_1=\frac{\mu^2C_1}{32}>0.
\end{equation}
Consequently, only $N=O(\log n+u)$ recorded times $\{\sigma_0,\dots,\sigma_N\}$ are required to reach half density, except with probability $e^{-\kappa_1 S}$. It remains to show that the physical time elapsed between consecutive entries (the stalls) is also under control.
\\
\\
\medskip\noindent\emph{Step 4: stalls are cheap.}
It remains to convert the stall steps into physical time. Between $\sigma_{j-1}$ and $\sigma_{j}$, there is the possibility of a stall. If so, we define $\theta_j$
as the moment the chain enters the stall, and $W_j = \sigma_j - \theta_j$ is the duration of such a stall. Therefore, we are interested in controlling the total stall time $\sum_j W_j$. To this end, we initialize at $W_0=\sigma_0$ and
$\theta_0=0$. For $j\ge1$, define
\begin{equation}
  \theta_j=\sigma_{j-1}+\mathbb I_{\{j\le\widehat\tau\}},\qquad W_j=\sigma_j-\theta_j.
\end{equation}
These definitions deserve three quick checks before we use them. First, each $\theta_j$ is a stopping time, since $\{j \le \widehat{\tau}\}= \left\{H_{\sigma_{j-1}} < \frac{n}{2}\right\} \in \mathcal{F}_{\sigma_{j-1}}$. Note also that after the hit all $\sigma_j$ coincide, $\theta_j=\sigma_j$, and $W_j=0$:
stalls simply stop occurring. Second, write $\mathcal F_{\theta_j}$ for
the information available when the $j$-th stall begins. Since the times $(\theta_j)$ are
nondecreasing, these sigma-algebras are nested. Third, by the time stall $j$ begins, all
earlier stall durations are already known facts: for $i<j$, both $\theta_i$ and
$\sigma_i$ are at most $\theta_j$, so $W_i=\sigma_i-\theta_i$ is
$\mathcal{F}_{\theta_j}$-measurable.
\\
\\
We now want to show that each stall is short, given its entire past. At the
moment $\theta_j$, exactly one of two things is true. Either the size is already at or
above $h_\ast$ (in which case $W_j=0$ and there is nothing
to bound), or the size is below $h_\ast$, in which case $W_j$ is precisely the return
time of \cref{lem:escape} started at $\theta_j$. In both cases,
applying that lemma at the stopping time $\theta_j$ gives, uniformly over the complete
history,
\begin{equation}\label{eq:W-tail}
  \mathbb P\big(W_j>mL_0\,\big|\,\mathcal{F}_{\theta_j}\big)\le(1-p_\ast)^m,\qquad m\ge0.
\end{equation}
We now choose
\begin{equation}\label{eq:lambda-B}
  0<\lambda<L_0^{-1}\ln\frac1{1-p_\ast},
  \qquad
  B=\frac{e^{\lambda L_0}}{1-e^{\lambda L_0}(1-p_\ast)}.
\end{equation}
We also partition the values of $W_j$ into the disjoint events
\begin{equation}
  \{0\le W_j\le L_0\},\qquad \{mL_0<W_j\le(m+1)L_0\},\quad m\ge1.
\end{equation}
The first event has conditional probability at most $1=(1-p_\ast)^0$. For $m\ge1$, the
corresponding event is contained in $\{W_j>mL_0\}$ and therefore has conditional
probability at most $(1-p_\ast)^m$ by \cref{eq:W-tail}. Consequently,
\begin{equation}\label{eq:cond-B}
  \mathbb E\big[e^{\lambda W_j}\,\big|\,\mathcal{F}_{\theta_j}\big]
  \le\sum_{m=0}^{\infty}e^{\lambda(m+1)L_0}(1-p_\ast)^m=B.
\end{equation}
Although the stall durations are not independent, the moment bound \cref{eq:cond-B}
holds conditionally on everything before each stall, uniformly over that past, and
this is enough to bound their sum. Indeed, $\sum_{j=0}^{N-1}W_j$ is
$\mathcal F_{\theta_N}$-measurable, so it factors out of the conditional expectation
given $\mathcal F_{\theta_N}$, and \cref{eq:cond-B} bounds the remaining factor:
\begin{equation}
  \mathbb E\,e^{\lambda\sum_{j=0}^{N}W_j}
  \;=\;\mathbb E\Big[e^{\lambda\sum_{j=0}^{N-1}W_j}\,
        \mathbb E\big(e^{\lambda W_N}\,\big|\,\mathcal F_{\theta_N}\big)\Big]
  \;\le\;B\,\mathbb E\,e^{\lambda\sum_{j=0}^{N-1}W_j}.
\end{equation}
Iterating the same argument for $N-1,N-2,\dots,0$ gives
\begin{equation}\label{eq:bound_markov}
  \mathbb E\exp\Big(\lambda\sum_{j=0}^{N}W_j\Big)\;\le\;B^{\,N+1};
\end{equation}
We now convert this moment bound into a tail bound. Choose $a>\lambda^{-1}\ln B$ and
set $\kappa_2=\lambda a-\ln B>0$. Since exponentiating is monotone, the events
$\{\sum_j W_j>a(N+1)\}$ and $\{e^{\lambda\sum_j W_j}>e^{\lambda a(N+1)}\}$ coincide,
so Markov's inequality applied to $e^{\lambda\sum_j W_j}$, together with
\cref{eq:bound_markov}, yields
\begin{equation}\label{eq:W-sum-tail}
  \mathbb P\Big(\sum_{j=0}^{N}W_j>a(N+1)\Big)
  \;\le\;B^{N+1}\,e^{-\lambda a(N+1)}
  \;=\;e^{-\kappa_2(N+1)}\;\le\;e^{-\kappa_2C_1S}.
\end{equation}
Consequently, the total stalled time $\sum_j W_j$ stays below a constant budget of
$a$ layers per recorded time, except with probability $e^{-\kappa_2 C_1 S}$: stalls
add only a constant factor to the length of the journey.
\\
\\
\medskip\noindent\emph{Step 5: telescope and tune.}
We now know that with high probability, $\widehat\tau\leq N$ suffices, by \cref{eq:tauhat-tail}, and the total stalled time is small \cref{eq:W-sum-tail}. It remains to assemble them together into a bound on the total physical hitting time $\tau_{1/2}$.
\\
\\
The total hitting time equals the time to get out of the first stall, the number of recorded times above $h_*$, and the additional stalls after the first one. If $\widehat\tau\le N$, then
\begin{equation}
\tau_{1/2}=\sigma_{\widehat\tau}=W_0+\widehat\tau+\sum_{j=1}^{\widehat\tau}W_j
  \le N+\sum_{j=0}^{N}W_j.
\end{equation}
$\tau_{1/2}=\sigma_{\widehat{\tau}}$ because the chain is strictly below $\frac{n}{2}$ before $\sigma_{\widehat{\tau}}$. Stall values are below $h_*,$ and earlier recorded values are below $\frac{n}{2}$ by the definition of $\widehat{\tau}$. We now intersect the two good events of few recorded times $<N$ (\cref{eq:tauhat-tail}) and small total stall time below $a(N+1)$ (\cref{eq:W-sum-tail}):
\begin{equation}
  \tau_{1/2}\le N+a(N+1)\le C_{\mathrm{tail}}S,
  \qquad C_{\mathrm{tail}}=(1+a)(C_1+1)+a,
\end{equation}
where $N\le C_1S+1$ and $S\ge1$. Therefore,
\begin{equation}\label{eq:pre-tail}
  \mathbb P\big(\tau_{1/2}>C_{\mathrm{tail}}S\big)
  \le e^{-\kappa_1S}+e^{-\kappa_2C_1S}\le 2e^{-c_{\mathrm{tail}}S},
\end{equation}
with
\begin{equation}
  c_{\mathrm{tail}}=\min\{\kappa_1,\;\kappa_2C_1\}>0.
\end{equation}
For $0<\delta\le1/2$, setting
\begin{equation}
  u=\frac{\ln2}{c_{\mathrm{tail}}}\Big(1+\log\frac1\delta\Big)
\end{equation}
yields
\begin{equation}
  2e^{-c_{\mathrm{tail}}u}=2e^{-(\ln2)(1+\log_2(1/\delta))}=\delta.
\end{equation}
Since $\log(1/\delta)\ge1$,
\begin{equation}
  u+1\le\Big(1+\frac{2\ln2}{c_{\mathrm{tail}}}\Big)\log\frac1\delta.
\end{equation}
Thus \cref{eq:hit-delta} follows with
\begin{equation}
C_{\mathrm{hit}}=C_{\mathrm{tail}}\Big(1+\frac{2\ln2}{c_{\mathrm{tail}}}\Big).
\end{equation}
\end{proof}
The previous lemma finally proves that singletons hit support size $n/2$ in logarithmic time with high probability.
\newpage
\subsection{Uniform mixing}\label{sec:uniformmix}
After proving that small supports reach support size $n/2$ with overwhelming probability in logarithmic time, we now proceed to the rest of the roadmap we established at the end of \cref{sec:Grand coupling between support chains}: Persistence and Contraction.
\subsubsection{Dense persistence}
We want to show that once a support has size $n/2$, it stays above $n/2$ with overwhelming probability. To do so, we use the drift and the control of both sources of stochasticity we already derived in \cref{sec:growth}.
\begin{lemma}[Dense persistence]\label{lem:dense-persistence}
With the explicit constant $c_{\mathrm d}=1/3200$, on the event $H_t\ge n/2$,
\begin{equation}\label{eq:persist}
  \mathbb P\big(H_{t+1}<n/2\,\big|\,\mathcal F_t\big)\;\le\;2\,e^{-c_{\mathrm d}n}.
\end{equation}
\end{lemma}

\begin{proof}
Condition on $\mathcal F_t$, so the support $A_t$ is fixed, with $H_t\ge n/2$. Using the equality \cref{eq:active-edges} in \cref{eq:internal-edges},
\begin{equation}\label{eq:ER-dense}
  \mathbb E[R_t\mid\mathcal F_t]\;=\;H_t-\frac{H_t(H_t-1)}{2(n-1)}\;\ge\;\frac{3n}{8},
\end{equation}
where the last inequality uses concavity in $H_t\in[n/2,n]$, with values $\ge 3n/8$ at
both endpoints ($n/2$ at $H_t=n$, and $>\tfrac{3n}{8}$ at $H_t=n/2$).

Taking $R_t<7n/20$, then $J_t-\mathbb EJ_t=\mathbb ER_t-R_t>n/40$. Since $H_t\le n$,
\Cref{lem:exposure} with $u=n/40$ gives 
\begin{equation}\label{eq:Rt-tail}
  \mathbb P\big(R_t<7n/20\,\big|\,\mathcal F_t\big)\;\le\;e^{-n/3200}.
\end{equation}
Using the relation  $B_t \;=\; H_{t+1}-R_t$ from \cref{eq:one-step}. On the event $R_t\ge7n/20$, if  $H_{t+1}<n/2$, then

\begin{equation}
  B_t\;<\;n/2-R_t\;\le\;3R_t/7\;=\;\tfrac57\,\mathbb E\big[B_t\,\big|\,\mathcal F_t,M_t\big],
\end{equation}
where the last equality uses the intermediate relation $\mathbb{E}\!\left[ B_t \,\middle|\, \mathcal{F}_t, M_t \right]\;=\; \tfrac{3}{5}\, R_t$ in \cref{eq:match_update}. Applying the Chernoff bound in \cref{eq:chernoff} with relative deviation $\varepsilon=2/7$ gives
\begin{equation}\label{eq:Bt-tail-dense}
  \mathbb P\big(B_t<3R_t/7\,\big|\,\mathcal F_t,M_t\big)
  \;\le\;e^{-6R_t/245}\;\le\;e^{-3n/350}.
\end{equation}
Since $3/350>1/3200$, adding \cref{eq:Rt-tail} and \cref{eq:Bt-tail-dense}
proves \cref{eq:persist}.
\end{proof}

\subsubsection{Contraction}
Before attacking the final contraction part, we prove a lemma that we will need to contract the gap set:
\begin{lemma}[Gap set contraction]\label{lem:discrepancy}
Couple $A_t\subseteq \widetilde{A}_t$ by the grand coupling and write
\begin{equation}
  a=|A_t|,\qquad d_t=|\widetilde{A}_t\setminus A_t|,\qquad z=n-a-d_t.
\end{equation}
Conditional on $A_t,\widetilde{A}_t$,
\begin{equation}\label{eq:contraction}
  \mathbb E[d_{t+1}\mid A_t,\widetilde{A}_t]
  \;=\;\frac{8}{5(n-1)}\left(\binom{d_t}{2}+d_tz\right)
  \;\le\;\frac{8(n-a)}{5(n-1)}\,d_t.
\end{equation}
If $a\ge n/2$ and $n\ge6$, the final contracting coefficient is at most
\begin{equation}\label{eq:rho}
  \varrho:=\frac{24}{25}<1.
\end{equation}
\end{lemma}

\begin{proof}
An output discrepancy occurs precisely on a matching edge that is inactive for $A_t$
but active for $\widetilde{A}_t$. Such an edge has either two endpoints in $\widetilde{A}_t\setminus A_t$, or
one endpoint there and one among the $z$ vertices outside $\widetilde{A}_t$. The expected number of such edges is
\begin{equation}
  \frac{\binom{d_t}{2}+d_tz}{n-1},
\end{equation}
because every pair is a
matching edge with probability $1/(n-1)$. On each such edge the $A_t$ chain outputs $\mathbb I\mathbb I$ and the upper chain
outputs a uniform nonidentity two-qubit Pauli, whose expected support size is $\frac{6\cdot1+9\cdot2}{15}=\frac85$. Edges active for both chains receive the same sampled label and create no
discrepancy. This proves the first equality of the lemma. Furthermore,
\begin{equation}
  \binom{d_t}{2}+d_tz\;\le\;d_t(d_t+z)\;=\;d_t(n-a),
\end{equation}
proving the inequality. If $a\ge n/2$, then for $n\ge6$
\begin{equation}
  \frac{8(n-a)}{5(n-1)}\;\le\;\frac{4n}{5(n-1)}\;\le\;\frac{24}{25}.
\end{equation}
\end{proof}

\subsubsection{Uniform coalescence in logarithmic time}
Finally, we put growth, persistence and contraction together into a final theorem, which proves Pauli mixing for the perfect-matching ensemble:

\paulimixing*

\begin{proof}
The proof runs the whole three-step program announced at the end of \cref{sec:Grand coupling between support chains}: grow every singleton chain to half density, ensure that no chain leaves the dense region, and merge each singleton chain with the full-support chain while both are dense.
Each of the three can fail, and we give each a failure budget of $\eta/16$. On the good event, all chains coalesce and the coupling inequality converts coalescence into the claimed total-variation bound.
 
\medskip\noindent\emph{Step 0: fix the constants.}
Let $\varrho=24/25$ be the contraction factor of \cref{lem:discrepancy}, and choose
\begin{equation}
  0<c_{\mathrm{mix}}<c_{\mathrm d}.
\end{equation}
Define the final depth constant
\begin{equation}\label{eq:C0}
  C_0:=4C_{\mathrm{hit}}+\frac{4}{\log(1/\varrho)}+3,
\end{equation}
whose two main terms will pay for the growth stage and the contraction stage, respectively. Finally, choose an even $n_{\mathrm{mix}}\ge\max\{n_{\mathrm g},6\}$ large enough that, for every $n\ge n_{\mathrm{mix}}$,
\begin{equation}\label{eq:nmix-cond}
  2nC_0\Big(\log n+\frac{c_{\mathrm{mix}}}{\ln2}n+1\Big)e^{-c_{\mathrm d}n}
  \;\le\;\frac{e^{-c_{\mathrm{mix}}n}}{16}.
\end{equation}
Such a threshold exists because $c_{\mathrm d}-c_{\mathrm{mix}}>0$, so the exponential on the right dominates the polynomial factor on the left; this display is exactly the persistence budget of Step~5, prepared in advance.
 
\medskip\noindent\emph{Step 1: the coupled chains and their discrepancies.}
Run the grand coupling simultaneously from the full support $A_0^{(\mathrm{full})}=[n]$ and from every singleton $A_0^{(i)}=\{i\}$. For each $i$, record when the singleton chain reaches half density, and how far it still is from the full chain:
\begin{equation}
  \tau_i:=\inf\{t\geq 0:\,|A_t^{(i)}|\ge n/2\},
  \qquad
  D_t^{(i)}=|A_t^{(\mathrm{full})}\setminus A_t^{(i)}|.
\end{equation}
Monotonicity of the coupling gives $A_t^{(i)}\subseteq A_t^{(\mathrm{full})}$ for all $t$, so $D_t^{(i)}$ is precisely the number of qubits on which the two chains still disagree: chain $i$ has coalesced with the full chain exactly when $D_t^{(i)}=0$.
 
\medskip\noindent\emph{Step 2: grow every singleton.}
Set
\begin{equation}\label{eq:tg}
  t_{\mathrm g}:=\Big\lceil C_{\mathrm{hit}}\Big(\log n+\log\frac{16n}\eta+1\Big)\Big\rceil.
\end{equation}
Applying \cref{lem:hitting-linear} to each singleton with failure probability $\delta=\eta/(16n) \in (0,1/2]$ and taking a union bound over the $n$ singletons,
\begin{equation}\label{eq:growth-fail}
  \mathbb P\big(\exists i:\,\tau_i>t_{\mathrm g}\big)\;\le\;\frac\eta{16}.
\end{equation}
This spends the first budget: except with probability $\eta/16$, all singleton chains are dense by time $t_{\mathrm g}$.
 
\medskip\noindent\emph{Step 3: the timetable.}
After growing, we will let the discrepancies contract for $s$ further layers, with
\begin{equation}\label{eq:s-T0}
  s=\Bigg\lceil\frac{\log(16n^2/\eta)}{\log(1/\varrho)}\Bigg\rceil,
  \qquad
  T_0=t_{\mathrm g}+s;
\end{equation}
$s$ is chosen so that $n\varrho^{\,s}\le\eta/(16n)$, i.e.\ so that a discrepancy of size at most $n$, shrunk by $\varrho$ per layer, falls below the second budget. The total time is logarithmic: writing $L=\log(n/\eta)$, we have $L\ge\log12>3$ (as $n\ge6$, $\eta\le1/2$) and $\log n\le L$, so
\begin{equation}
  \log n+\log\frac{16n}\eta+1\;\le\;2L+5\;\le\;4L,
  \qquad
  \log\frac{16n^2}\eta\;=\;4+L+\log n\;\le\;4L,
\end{equation}
and the ceiling inequality $\lceil x\rceil\le x+1$ gives
$t_{\mathrm g}\le(4C_{\mathrm{hit}}+1)L$ and $s\le(4/\log(1/\varrho)+1)L$. Therefore
\begin{equation}\label{eq:T0-bound}
  T_0+1\;\le\;C_0\log\frac n\eta.
\end{equation}
 
\medskip\noindent\emph{Step 4: merge.}
\Cref{lem:discrepancy} contracts the expected discrepancy by $\varrho$ per layer, but only while the singleton chain is dense. The subtlety is that we may not condition on the chain staying dense, as that would condition on the future. The clean device is to weight by the indicator of ``dense so far'', which only looks at the past: for $0\le r\le t_{\mathrm g}$ and $t\ge r$, define
\begin{equation}
  E_{i,r}=\{\tau_i=r\},
  \qquad
  G_{i,r,t}=E_{i,r}\cap\bigcap_{u=r}^{t}\{|A_u^{(i)}|\ge n/2\},
\end{equation}
the event that chain $i$ became dense at time $r$ and has remained dense through time $t$. This event is $\mathcal F_t$-measurable, and it updates by one intersection per layer, $G_{i,r,t+1}=G_{i,r,t}\cap\{|A_{t+1}^{(i)}|\ge n/2\}$. Hence, for deterministic $t\ge r$,
\begin{align}
  \mathbb E\big[D_{t+1}^{(i)}\mathbf 1_{G_{i,r,t+1}}\,\big|\,\mathcal F_t\big]
  &=\mathbf 1_{G_{i,r,t}}\,
    \mathbb E\big[D_{t+1}^{(i)}\mathbf 1_{\{|A_{t+1}^{(i)}|\ge n/2\}}\,\big|\,\mathcal F_t\big]\\
  &\le\mathbf 1_{G_{i,r,t}}\,
    \mathbb E\big[D_{t+1}^{(i)}\,\big|\,\mathcal F_t\big]\\
  &\le\varrho\,D_t^{(i)}\mathbf 1_{G_{i,r,t}},\label{eq:contract-step}
\end{align}
where the first step pulls out the past-measurable indicator, the second discards the remaining indicator, and the last is \cref{lem:discrepancy}: on $G_{i,r,t}$ the lower support has size at least $n/2$, and the fresh layer depends on the past only through the current supports. Iterating \cref{eq:contract-step} over the deterministic times $r,r+1,\dots,T_0-1$, and bounding the initial discrepancy crudely by $n$,
\begin{equation}\label{eq:contract-iter}
  \mathbb E\big[D_{T_0}^{(i)}\mathbf 1_{G_{i,r,T_0}}\big]
  \;\le\;\varrho^{\,T_0-r}\,\mathbb E\big[D_r^{(i)}\mathbf 1_{E_{i,r}}\big]
  \;\le\;n\,\varrho^{\,T_0-r}\,\mathbb P(E_{i,r}).
\end{equation}
Since $D_{T_0}^{(i)}$ is integer valued, $\{D_{T_0}^{(i)}>0\}$ means $D_{T_0}^{(i)}\ge1$, so Markov's inequality applied on each of the disjoint events $E_{i,r}$, summed over $r\le t_{\mathrm g}$ and using $T_0-r\ge s$, gives
\begin{align}
  \mathbb P\big(D_{T_0}^{(i)}>0,\ \tau_i\le t_{\mathrm g},\
  |A_u^{(i)}|\ge n/2\ \text{for }\tau_i\le u\le T_0\big)
  \;\le\;n\varrho^{\,s}\;\le\;\frac\eta{16n}.\label{eq:contract-fail}
\end{align}
A union bound over the $n$ chains spends the second budget, $\eta/16$: except with this probability, every chain that grew on schedule and stayed dense has merged with the full chain by time $T_0$.
 
\medskip\noindent\emph{Step 5: persist.}
The previous step is conditional on staying dense, so we must bound the probability that any chain ever falls back below $n/2$. Decompose by the exit time: for $0\le t<T_0$, let $F_{i,t}$ be the event that $\tau_i\le\min\{t,t_{\mathrm g}\}$, chain $i$ stays dense from $\tau_i$ throughout $t$, and falls below $n/2$ at $t+1$. The pre-exit part is $\mathcal F_t$-measurable, so \Cref{lem:dense-persistence} applies to the last layer and gives
\begin{equation}
  \mathbb P(F_{i,t})\;\le\;2e^{-c_{\mathrm d}n}.
\end{equation}
A union bound over the $nT_0$ pairs $(i,t)$ yields
\begin{equation}\label{eq:exit-fail}
  \mathbb P\big(\exists i\ \text{with a dense exit by }T_0\big)\;\le\;2nT_0e^{-c_{\mathrm d}n}.
\end{equation}
Since $\eta\ge e^{-c_{\mathrm{mix}}n}$, \cref{eq:T0-bound} implies
$T_0\le C_0\big(\log n+\tfrac{c_{\mathrm{mix}}}{\ln2}n+1\big)$, and the choice of
$n_{\mathrm{mix}}$ in \cref{eq:nmix-cond} bounds \cref{eq:exit-fail} by $\eta/16$:
the third budget.
 
\medskip\noindent\emph{Step 6: coalescence and mixing.}
Combining the three failure bounds, with probability at least $1-3\eta/16$ every singleton chain equals the full-support chain at time $T_0$. All other initial supports are then dragged along for free: for any nonempty $S$, choose $i\in S$. Monotonicity gives
\begin{equation}
  A_{T_0}^{(i)}\subseteq A_{T_0}^{(S)}\subseteq A_{T_0}^{(\mathrm{full})},
\end{equation}
and since the two ends of the sandwich coincide, all support chains agree at $T_0$.
 
Finally, couple the chain started from an arbitrary fixed support $A$ to a chain whose initial support is drawn from $\pi_\mathrm{supp}$. The latter remains distributed as $\pi_\mathrm{supp}$ forever by \cref{eq:stationarity_pi_supp}. The coupling inequality from \cref{lem:coupling_inequality} converts coalescence into total variation:
\begin{equation}\label{eq:final-TV}
  \mathrm{TVD}\big(\kappa_n^{T_0+1}(A,\cdot),\,\pi_\mathrm{supp}\big)\;\le\;\frac{3\eta}{16}\;<\;\eta.
\end{equation}
This holds for any $T\ge T_0+1$, as running further shared layers can only keep coalesced chains together. Finally, we conclude the theorem by taking $C_{\mathrm{mix}}\ge C_0$ and applying \cref{eq:pauli_support_identity} to convert the TVD from the support Markov chain to the Pauli Markov chain.
\end{proof}

\newpage
\section{Proofs of the main results}
\label{app:main-results-proofs}

This section assembles the ingredients proved in the preceding appendices and
provides the complete proofs of the results stated in
\cref{sec:Main_results}. The statements imported from the cited works remain
inputs; what is proved here is the complete deduction of the main results from
those inputs and from the Pauli-mixing analysis of this manuscript.

Throughout this appendix, let
\begin{equation}
\Sigma_+ := \{1, T\}, \qquad \Sigma_- := \{*, \dagger\}, \qquad \Sigma := \Sigma_+ \cup \Sigma_-,
\label{eq:sigma}
\end{equation}
where all transposes and complex conjugates are taken in the fixed basis used to define the oracle resources. For an ensemble $\mu$ and a fixed orientation word $\boldsymbol{s} = (s_1, \dots, s_r) \in \Sigma^r$, with $r \in \{1,2\}$, define the half-trace-distance measurable error in that word by
\begin{equation}
\operatorname{Err}_{\boldsymbol{s}}(\mu) := \sup_{\mathcal{W}} \frac{1}{2} \bigl\| \mathbb{E}_{U \sim \mu}\, \rho^{U}_{\mathcal{W},\boldsymbol{s}} - \mathbb{E}_{H \sim \mathrm{Haar}}\, \rho^{H}_{\mathcal{W},\boldsymbol{s}} \bigr\|_1,
\label{eq:err-word}
\end{equation}
where the supremum is over all finite-memory $r$-slot quantum combs (Appendix~A.1.1); the orientation word is fixed before the interaction, whereas the comb may use arbitrary quantum memory, intermediate channels, measurements, and classical feed-forward.

For an ensemble $\mu$, let $\operatorname{Err}_{2}(\mu)$ denote the worst
distinguishing advantage in half trace distance over all fixed orientation
words of length one or two. For general query number $k$, let
$\operatorname{Err}_{k}^{\rm full}(\mu)$ denote the corresponding supremum
of the full trace-norm difference of the averaged output states. Thus
\begin{equation}
    \operatorname{Err}_{2}^{\rm full}(\mu)
    =
    2\operatorname{Err}_{2}(\mu).
    \label{eq:appendix-error-conversion}
\end{equation}

\subsection{Weak relative-error 2-designs}
\label{app:apendix_relativeerr_twodesign}
In this section we detail the imported construction of a weak relative-error $k$-design in logarithmic depth. The key ingredient is the weak gluing theorem:
\begin{theorem}[Weak gluing, adapted Theorem~1 from Ref.~\cite{SchusterHaferkampHuang2025}]
\label{thm:gluing-small-designs}
Given any approximation error $\varepsilon \leq 1$, suppose that each small random unitary in the two-layer brickwork ensemble $\mathcal{E}$ is drawn from an $\varepsilon/n$-approximate unitary $k$-design on $2\xi$ qubits with circuit depth $d$. Then $\mathcal{E}$ forms an $\varepsilon$-approximate unitary $k$-design on $n$ qubits with depth $2d$, provided that the local patch size satisfies
\begin{equation}
    \xi \geq \log_2\left(\frac{nk^2}{\varepsilon}\right).
\end{equation}
\end{theorem}
Consequently, we consider a two-layer brickwork of weak relative-error $k$-designs, with each brick supported on $2\xi$ qubits and overlapping its neighbours on patches of $\xi =\Omega( \log_2\left(\frac{nk^2}{\varepsilon}\right))$ qubits (see the weak 2-design in \cref{fig:strong-2-layer}). The previous \cref{thm:gluing-small-designs} allows us to glue the bricks iteratively, producing a weak relative-error design on the full $n$ qubits system. 

In turn, each of the bricks can be realized through a 1D brickwork random circuit, which is known to form weak relative-error $k$-designs in depth linear in the support size \cite{chen2024incompressibilityspectralgapsrandom}. Because the weak relative-error $k$-designs required for weak gluing are supported on $2\xi=O( \log_2\left(\frac{nk^2}{\varepsilon}\right))$ qubits, and the bricks run in parallel, the two-layer brickwork realizes a weak relative-error design in logarithmic depth:  
\begin{corollary}[Low-depth relative-error designs, adapted Corollary~1 from Ref.~\cite{SchusterHaferkampHuang2025}]
\label{thm:shh-lowdepth}
Random quantum circuits over $n$ qubits form $\varepsilon$-approximate unitary $k$-designs in relative error, in circuit depth
\begin{equation}
d=O\!\bigl(k\operatorname{polylog}(k)\,\log(n/\varepsilon)\bigr),
\label{eq:shh-depth}
\end{equation}
for 1D circuits without additional qubits.
\end{corollary}
Consequently, for fixed $k=2$, let $\nu_{n,\delta}$ denote the resulting weak relative-error $2$-design with error $\delta$. There is a universal constant $C_w$ such that
\begin{equation}
    d_w \leq C_w\log(\frac{n}{\delta})
\end{equation}
Because every elementary gate of the one-dimensional realization is an independent Haar-random gate, and the complex conjugate of a Haar-random gate is again Haar random, $\nu_{n,\delta}$ is invariant under complex conjugation.

\newpage

\subsection{Proof of \cref{thm:shell}}
\label{app:proof-strong-two-shell}

\begin{proof}[Proof of Theorem 2]
 In \cref{lem:composition}, $\mu_{1}*\mu_{2}$ denotes the distribution of $U=VW$ with
$V\sim\mu_{1}$ and $W\sim\mu_{2}$ sampled independently. All errors refer to a fixed type of query (same-sign or mixed-sign). The mixed query error $\varepsilon_1$ and the resulting strong-design error use the full trace norm, whereas $\varepsilon_2$ denotes the weak relative error.

We apply \cref{lem:composition} with $\mu_{2}=\nu_{n,\delta}$ and
$\mu_{1}=\cE_{\PM}^{(T)}$, sampled independently. By
\cref{thm:shh-lowdepth}, $\nu_{n,\delta}$ is a weak $2$-design with
relative error $\delta$, and depth $d_{w}\leq C_{w}\log(n/\delta)$.
By \cref{thm:mixing}, for any
$T\geq\lceil C_{\rm mix}\log(n/\eta)\rceil$, the perfect-matching
ensemble $\cE_{\PM}^{(T)}$ has Pauli transition distribution
$\eta$-close in total variation to the uniform distribution $\pi$ for
every nonidentity input Pauli. Since $\cE_{\PM}^{(T)}$ is invariant
under conjugation, transposition, and single-qubit Pauli rotations at
the input and output layers, \cref{pro:mixed_query_reduction} applies and
shows that it is robust under mixed queries with measurable error
at most $\eta$ in half trace distance, which by the conversion
\eqref{eq:appendix-error-conversion} is at most $2\eta$ in full
trace norm. Hence \cref{lem:composition} applies with mixed-query
error $\varepsilon_{1}=2\eta$ and weak-design relative error
$\delta$, and yields that $\cE_{\PM}^{(T)}*\nu_{n,\delta}$ is a
strong $2$-design with measurable error $\max\{2\eta,2\delta\}$ in
full trace norm.

We briefly record why the two sector errors combine as a
maximum, which is the content of the imported \cref{lem:composition}. The
orientations of $U=VW$ factor as
\begin{equation}
    U=VW,
    \qquad
    U^{T}=W^{T}V^{T},
    \qquad
    U^{*}=V^{*}W^{*},
    \qquad
    U^{\dagger}=W^{\dagger}V^{\dagger},
    \label{eq:orientation-factorization}
\end{equation}
so in every fixed two-query word, each oracle call to an
orientation of $U$ contains exactly one call to the same orientation
of $V$ and one of $W$. For a mixed-sign query, condition on $W=w$: the
fixed oriented gates are absorbed into the neighbouring comb
channels, so the conditional experiment is again an admissible comb
querying $V$ alone, and the mixed-query security of $\mu_{1}$ replaces
$V$ by a Haar unitary $H$ at cost $\varepsilon_{1}$. Since $Hw$ is
exactly Haar-distributed for every fixed $w$, no second error occurs.
For a same-sign word, condition on $V=v$ and absorb the fixed oriented gates into the neighboring comb channels; the same-sign branch of the imported composition argument \cref{lem:composition} (Lemma~19 of Ref.~\cite{FolkertsmaEtAl2026}) replaces $W$ by a Haar unitary $H$ at cost at most $2\delta$ in full trace norm, with the negative branch following from complex-conjugation invariance, and left Haar invariance makes $vH$ Haar-distributed for every fixed $v$, so no further error occurs.
One-query words are included by appending a second query to a maximally mixed register and discarding its output. Since each fixed word
requires replacing only one of the two factors, the errors combine as
$\max\{\varepsilon_{1},2\delta\}$.

For the final claim, set $\eta=\delta=\varepsilon_{2}/2$ and
$T=\lceil C_{\rm mix}\log(2n/\varepsilon_{2})\rceil$. Both parameters
are admissible as long as $c_{\rm sh}\leq c_{\rm mix}$
(so that $\eta=\varepsilon_{2}/2\geq e^{-c_{\rm sh}n}\geq
e^{-c_{\rm mix}n}$) and $n_{\rm sh}\geq n_{\rm mix}$ is chosen large
enough to include all size thresholds of
\cref{thm:shh-lowdepth}. Note that $\delta=\varepsilon_{2}/2>0$, as
required, and the error bound becomes
$\max\{\varepsilon_{2},\varepsilon_{2}\}=\varepsilon_{2}$. The total
depth is $T$ plus the depth of $\nu_{n,\varepsilon_{2}/2}$, at
most $\lceil C_{\rm mix}\log(2n/\varepsilon_{2})\rceil
+C_{w}\log(2n/\varepsilon_{2})$, and is therefore at most
$C_{\rm sh}\log(n/\varepsilon_{2})$ for a universal constant
$C_{\rm sh}$.
\end{proof}

\subsection{Strong gluing of the scrambled two-layer ensemble}
\label{app:iterated-gluing}

Recall from \cref{sec:assembling} and \cref{fig:strong-k} the scrambled two-layer ensemble of \cite{SchusterEtAlStrong2025}: the $n$ qubits are partitioned into $m$ patches of $\xi$ qubits (except of the last patch, which has $\xi + \zeta$ qubits), two brickwork layers of independent local strong $k$-designs act on overlapping unions of adjacent patches, and the whole is sandwiched between two independent global strong $2$-designs. We import the guarantee that this architecture produces a strong $k$-design on all $n$ qubits, which follows by iteratively gluing together the strong $k$-designs in the two-brickwork layer via the imported \cref{lem:gluing_strong_random_unitaries}:

\begin{lemma}[Gluing protocol for the scrambled two-layer ensemble \footnote{The version in the references assumes the condition $n$ being exactly divisible by $\xi$. We relax that assumption with $m = \lfloor n/\xi \rfloor$ patches, and absorbing the remainder $\zeta = n - m\xi$ into the last patch, which
then has size $\xi + \zeta < 2\xi$. The strong gluing \cref{lem:gluing_strong_random_unitaries} is still applied a total amount of $m$ times, with each step still having registers $a$, $b$, $c$ of size $\geq \xi$, so the final result is unaffected.}, adapted Theorem 5 of Ref.~\cite{SchusterEtAlStrong2025} and the remark following its proof in Appendix~D.2 of Ref.~\cite{SchusterEtAlStrong2025} ]
\label{lem:patch_protocol}
We divide our $n$ qubits into $m=\lfloor n/\xi \rfloor \geq 3$  patches, the first $m-1$ of size $\xi$, and the final one of size $\xi+\zeta$, such that $n=m\xi+\zeta$. Let $\varepsilon \leq 1$. We define the following protocol, which samples each step independently:
\begin{enumerate}
    \item We apply an $\varepsilon_2$-approximate strong $2$-design on all qubits,  with $\varepsilon_2 = \varepsilon^{8}/(B\, n^{8} k^{5})$.
    \item We apply independent $\varepsilon/(3n)$-approximate strong $k$-designs on the
          patches $(2i-1,2i)$, for every $1\leq i\leq\lfloor m/2 \rfloor$.
    \item We apply independent $\varepsilon/(3n)$-approximate strong $k$-designs on the
          patches $(2i,2i+1)$, for every $1\leq i\leq \lceil m/2 \rceil-1$.
    \item We apply an $\varepsilon_2$-approximate strong $2$-design on all qubits.
\end{enumerate} Where in steps 2 and 3 each pair of patches has size $2\xi$, except the pair containing the
final patch, which has size $2\xi + \zeta$.
There exists a universal constant $B\geq 1$ such that the ensemble resulting from this protocol forms a strong
$\varepsilon$-approximate unitary $k$-design when
\begin{equation}
    \xi
    \;\geq\;
    \frac{16}{3}
    \log\!\left(\frac{nk^{2}}{\varepsilon}\right)
    +
    \mathcal{O}(1).
    \label{eq:patch_size_ew}
\end{equation}

\end{lemma}

The final result from strong gluing uses the full trace-norm measurable-error convention.

\subsection{Proof of \cref{thm:strong-k}}
\label{app:proof-strong-k}

\begin{proof}[Proof of \cref{thm:strong-k}]

We prove the theorem by proposing a realization of the four steps of the protocol in
\cref{lem:patch_protocol}, on a partition of the $n$ qubits into
$m=\lfloor n/\xi\rfloor\in\mathbb{N}$ patches (of size $\xi$ except from the last one of size $\xi+\zeta$). We first fix the universal constants
$\xi_{0}$, $A$ and $c_{0}$, and we write
\begin{equation}
    \varepsilon_2 \;:=\; \frac{\varepsilon^{8}}{B\, n^{8} k^{5}},
    \label{eq:shell_accuracy}
\end{equation}
for the strong 2-design shell accuracy demanded by steps 1 and 4 of \cref{lem:patch_protocol}, where
$B$ is the universal constant fixed there.

The depth of the local blocks of steps 2 and 3 will grow linearly with the patch size,
so $\xi$ is defined to sit just above the smallest value that \cref{lem:patch_protocol}
admits. Taking $\xi_{0}$ at least the universal constant hidden in the
$\mathcal{O}(1)$ term of \cref{eq:patch_size_ew} guarantees \cref{lem:patch_protocol}.
Because $\xi$ is minimal up to these constants, it is logarithmic in all parameters:
there is a universal constant $C_{\xi}$ such that
\begin{equation}
    \xi \leq C_{\xi}\log\frac{nk}{\varepsilon}.
    \label{eq:xi_upper}
\end{equation}
\Cref{eq:xi_upper} enters the proof twice: it lets the patches fit into the system, and
it converts the depth of the construction, linear in $\xi$, into the logarithmic depth
bound.

The condition $\log(nk/\varepsilon)\leq c_{0}n$ must guarantee that the system is large
enough for the construction, in three respects: (i) the patches must be small relative
to the system, $4\xi\leq n$ (hence $m\geq 4$, so in particular the requirement
$m \geq 3$ of \cref{lem:patch_protocol} holds); (ii) $n\geq n_{\mathrm{sh}}$, so that
\cref{thm:shell} applies to the shells of steps 1 and 4; and (iii)
$2e^{-c_{\mathrm{sh}}n}\leq \varepsilon_2$, so that the accuracy $\varepsilon_2$
demanded of the shells lies within the accuracy window of \cref{thm:shell}. All three
follow from \cref{eq:xi_upper} and the condition $\log(nk/\varepsilon)\leq c_0 n$ once
$c_{0}$ is a small enough universal constant. For (i), \cref{eq:xi_upper} gives
$4\xi\leq 4C_{\xi}c_{0}n\leq n$ once $c_{0}\leq 1/(4C_{\xi})$. For (ii), the chain
$\log(4n)\leq\log(nk/\varepsilon)\leq c_{0}n$ places $n$ above any prescribed absolute
threshold as $c_{0}$ shrinks. For (iii), by \cref{eq:shell_accuracy},
\begin{equation}
    \log_{2}\frac{2}{\varepsilon_2}
    \;=\; 1 + \log_{2}B + 8\log_{2}n + 5\log_{2}k + 8\log_{2}\frac{1}{\varepsilon}
    \;\leq\; C_{*}\log_{2}\frac{nk}{\varepsilon}
\end{equation}
for a universal constant $C_{*}$, so imposing $C_{*}c_{0}\leq c_{\mathrm{sh}}/\ln 2$
yields $2e^{-c_{\mathrm{sh}}n}\leq\varepsilon_2$.

We now construct an $\varepsilon$-measurable error $k$-design by making a realization
of each of the steps of \cref{lem:patch_protocol}.

\emph{Steps 1 and 4.} Each strong 2-design shell is an independent copy of the ensemble
$\mathcal{E}^{(T)}_{\mathrm{PM}}*\nu_{n,\delta}$ of \cref{thm:shell}, taken with
measurable error $\varepsilon_2$ in full trace-norm ($\eta=\delta=\varepsilon_2/2$) and
depth of the perfect-matching ensemble
$T=\lceil C_{\mathrm{mix}}\log(2n/\varepsilon_2)\rceil$. Since
$2e^{-c_{\mathrm{sh}}n}\leq \varepsilon_2\leq 1$, \cref{thm:shell} applies. Each of the
two is a strong $\varepsilon_2$-approximate unitary $2$-design in full trace-norm
measurable error, of depth at most $C_{\mathrm{sh}}\log(n/\varepsilon_2)$, exactly as
steps 1 and 4 of \cref{lem:patch_protocol} require.

\emph{Steps 2 and 3.} On each of the pairs of adjacent patches, of size $2\xi$ except the pair containing the last patch, of size $2\xi+\zeta\leq 3\xi$, we place an independent one-dimensional random circuit of
\cref{lem:1D_strong_unitary_designs} with target relative error $\varepsilon/(6n)$.
This target is admissible:
$\xi := \left\lceil \frac{16}{3}\log_2\left(\frac{Ank^2}{\varepsilon}\right)\right\rceil+\xi_0$
gives $2^{2\xi}\geq (Ank^{2}/\varepsilon)^{32/3}\geq Ank^{2}/\varepsilon$, so with
$A\geq 12$,
\begin{equation}
 \frac{2k^{2}}{2^{2\xi+\zeta}} \;\leq\;
    \frac{2k^{2}}{2^{2\xi}} \;\leq\; \frac{2k^{2}\varepsilon}{12\,nk^{2}} \;=\;
    \frac{\varepsilon}{6n}.
    \label{eq:local_admissible}
\end{equation}
By the translation between error notions, $\varepsilon_{m}\leq 2\varepsilon_{r}$
(Lemma~7 of Ref.~\cite{SchusterEtAlStrong2025}), each local factor is a strong
$\varepsilon/(3n)$-approximate unitary $k$-design in full trace-norm measurable error,
as steps 2 and 3 require.

The composed ensemble is therefore exactly the protocol of \cref{lem:patch_protocol},
which yields that the composed ensemble is a strong $\varepsilon$-approximate unitary
$k$-design in measurable error, in full trace-norm.

Finally, we bound the depth of the construction. By \cref{eq:shell_accuracy}, each of
the strong 2-design shells has depth at most
\begin{equation}
    C_{\mathrm{sh}}\log\frac{n}{\varepsilon_2}
    \;=\; C_{\mathrm{sh}}\log\frac{B\,n^{9}k^{5}}{\varepsilon^{8}}
    \;\leq\; C_{2}\log\frac{nk}{\varepsilon}
    \label{eq:shell_depth}
\end{equation}
for a universal constant $C_{2}$. Each of the 1D sandwiched brickworks is composed of
strong $k$-designs of at most $2\xi+\zeta\leq 3\xi$ qubits, running in parallel, with target
relative error $\varepsilon/(6n)$, so the depth bound of
\cref{lem:1D_strong_unitary_designs} becomes:
\begin{equation}
    d_{\mathrm{loc}} := C_{1D}\log^{7}(2k)\left((2\xi+\zeta)k+\log\frac{6n}{\varepsilon}\right)
    \leq C_{1D}\log^{7}(2k)\left(3k\xi+\log\frac{6n}{\varepsilon}\right)
    \leq C_{1D}'\,k\log^{7}(2k)\log\frac{nk}{\varepsilon}.
    \label{eq:dloc_new}
\end{equation}
Consequently, the total depth of the construction becomes:
\begin{equation}
    d \;\leq\; 2d_{\mathrm{loc}}+2C_{2}\log\frac{nk}{\varepsilon}
    \;\leq\; C\bigl(1+k\log^{7}(2k)\bigr)\log\frac{nk}{\varepsilon}
    \label{eq:depth_new}
\end{equation}
for a universal constant $C$. Every component acts only on the original $n$ qubits of
the system. For fixed $k\geq 2$ and fixed $\varepsilon < 1/4$, \cref{eq:depth_new} is
$O(\log n)$. The light-cone lower bound stated in \cref{pro:lower_bound} gives
$\Omega(\log n)$. Thus the optimal dependence on $n$ is $\Theta(\log n)$.

\end{proof}

\end{document}